\def\draft{0}
\documentclass[11pt]{article}

\usepackage[utf8]{inputenc}
\usepackage{amsmath,amssymb}
\usepackage{amsthm}
\usepackage{bbm}
\usepackage{bm}
\usepackage{latexsym}
\usepackage[noadjust]{cite}
\usepackage{graphicx}
\usepackage{hyperref}
\usepackage{xcolor}
\usepackage{dirtytalk}
\usepackage{mathtools}
\usepackage[margin=1in]{geometry}
\usepackage{thm-restate}
\makeatletter
\@ifundefined{newcounteralias}{}{%
  \renewcommand\thmt@autorefsetup{%
    \@xa\def\csname\thmt@envname autorefname\@xa\endcsname
    \@xa{\thmt@thmname}}%
}
\makeatother
\usepackage{enumitem}
\usepackage{tikz-cd}
\usepackage{comment}
\usepackage{algorithm}
\usepackage{float}
\floatstyle{boxed}
\restylefloat{algorithm}
\usepackage{arydshln}

\theoremstyle{plain}
\newtheorem{theorem}{Theorem}[section]
\newtheorem{lemma}[theorem]{Lemma}
\newtheorem{claim}[lemma]{Claim}

\newtheorem{proposition}[lemma]{Proposition}
\newtheorem{observation}[lemma]{Observation}
\newtheorem{corollary}[lemma]{Corollary}

\newtheorem{remark}[lemma]{Remark}

\newtheorem{definition}[lemma]{Definition}

\usepackage[nameinlink]{cleveref}

\newcommand{\cD}{\mathcal{D}} 

\newcommand{\R}{\mathbb{R}}

\newcommand{\E}{\mathbb{E}}

\newcommand{\ep}{\eps}

\newcommand{\eps}{\varepsilon}

\newcommand{\gnp}{\mathcal{G}_{n,p}}
\newcommand{\gnhalf}{\mathcal{G}_{n,\frac12}}
\newcommand{\erdosrenyi}{Erd\H{o}s--R\'enyi~}

\ifnum\draft=1
\newcommand{\cnote}[1]{{\color{purple} \;#1 --Cassandra\;}}
\newcommand{\mnote}[1]{{\color{brown} \;#1 --Madhu\;}}
\newcommand{\rnote}[1]{{\color{red} \;#1 --Ronitt\;}}
\newcommand{\ronitt}[1]{{\color{red} \;#1 --Ronitt\;}}
\else 
\newcommand{\cnote}[1]{{}}
\newcommand{\mnote}[1]{{}}
\newcommand{\rnote}[1]{}
\newcommand{\ronitt}[1]{}
\fi

\title{Quality control in sublinear time: a case study via random graphs}
\author{Cassandra Marcussen\thanks{School of Engineering and Applied Sciences, Harvard University, Cambridge, Massachusetts, USA. Email: cmarcussen@g.harvard.edu. This material is based upon work supported by the Air Force Office of Scientific Research under award number FA9550-23-F-0014 in the amount of \$139,400 to Cassandra Marcussen. Supported in part by NSF Award 2152413 and a
Simons Investigator Award to Madhu Sudan.} \\ \and Ronitt Rubinfeld\thanks{Computer Science and Artificial Intelligence Laboratory, MIT, Cambridge, Massachusetts, USA. Email: ronitt@csail.mit.edu. Supported by the NSF TRIPODS program (award DMS-2022448) and CCF-2310818.} \\ \and Madhu Sudan\thanks{School of Engineering and Applied Sciences, Harvard University, Cambridge, Massachusetts, USA. Email: madhu@cs.harvard.edu. Supported in part by a Simons Investigator Award, NSF Award CCF 2152413, and AFOSR award FA9550-25-1-0112.}}
\date{}

\begin{document}

\maketitle

\begin{abstract}

Many algorithms are designed to work well on average over inputs. But when such an algorithm is run on an arbitrary input, the question arises: Can we trust the algorithm on this input? This question is especially relevant in network science, where algorithms for theoretical random graph models are often applied to individual real-world networks.

We identify a new class of algorithmic problems addressing this, which we call ``Quality Control Problems.'' The (verification versions of these) problems are specified by a (positive, real-valued) ``quality function'' $\rho$ and a distribution $\mathcal{D}$ such that, with high probability, a sample drawn from $\mathcal{D}$ is ``high quality,'' meaning its $\rho$-value is near $1$. 
The goal is to accept most inputs $x \sim \mathcal{D}$ and reject potentially adversarially generated inputs $x$ with $\rho(x)$ far from $1$. 
The objective of quality control is thus weaker than either component problem: testing for ``$\rho(x) \approx 1$'' or testing if $x \sim \mathcal{D}$, and offers the possibility of more efficient algorithms.
This definition is inspired by Feige's work on refutation of random constraint satisfaction problems and by Rubinfeld and Vasilyan's work on testable learning.

In this work, we consider the sublinear version of the quality control problem, where $\mathcal{D} \in \Delta(\{0,1\}^N)$ and the goal is to solve the $(\mathcal{D},\rho)$-quality problem with $o(N)$ queries and time.
As a case study, we consider random graphs, i.e., $\mathcal{D} = \mathcal{G}_{n,p}$ (and $N = \binom{n}2$), and the $k$-clique count function $\rho_k := C_k(G)/\mathbb{E}_{G' \sim \mathcal{G}_{n,p}}[C_k(G')]$, where $C_k(G)$ is the number of $k$-cliques in $G$.
Testing if $G \sim \mathcal{G}_{n,p}$ with one sample, let alone with sublinear query access to the sample, is of course impossible.
 In suitable parameter ranges, worst-case approximate $k$-clique counting requires $p^{-\Omega(k^2)}$ queries.
 In contrast, we show that the quality control problem for $\mathcal{G}_{n,p}$ (with $0<p\leq1/2$ and $n\geq p^{-ck}$ for some constant $c$) with respect to $\rho_k$ can be tested with $p^{-O(k)}$ queries and time, showing quality control is provably superpolynomially more efficient in this setting. More generally, for a motif $H$ of maximum degree $\Delta(H)$, we show that the corresponding quality control problem can be solved with $p^{-O(\Delta(H))}$ queries and running time, which is tight up to a universal constant in the exponent.
\end{abstract}

\thispagestyle{empty}

\newpage

\tableofcontents
\thispagestyle{empty}
\addtocontents{toc}{\protect\thispagestyle{empty}}

\newpage

\setlength{\parskip}{\medskipamount}%

\setcounter{page}{1}
\section{Introduction}

In this work, we formalize a new class of problems: \textit{quality control problems}. These problems bridge the gap between algorithms designed for random inputs and adversarially generated examples. We then study a specific class of quality control problems for graph-theoretic properties on random graphs that show the (sublinear) algorithmic possibilities in this space. 

The primary motivation for our work is that there are many settings where it is easier --- either complexity-theoretically or analytically --- to design algorithms for instances drawn according to some distribution. For example, there exist polynomial-time algorithms for finding a satisfying assignment of a $k$-SAT formula and planted cliques of a graph with high probability over various models of generating the inputs \cite{DBLP:conf/focs/ChvatalR92, DBLP:conf/stoc/MekaPW15}.  Unfortunately, many of these algorithms perform poorly on adversarially generated examples. Our goal is to see if one can take advantage of the ``average-case'' assumptions as these algorithms do, without completely losing out on formal guarantees on {\em specific  instances}. 

We begin with a motivating scenario that represents a common goal in network science \cite{ouadah2015degree, brune2025goodness}. Suppose a network scientist is given an input graph $G$, hypothesized to be drawn from the \erdosrenyi random graph distribution $\gnp$, the distribution over graphs with $n$ vertices where each edge appears independently with probability $p$. The network scientist's objective is:
\begin{enumerate}
    \item Accept graphs drawn from $\gnp$, to justify applying subsequent algorithms designed to work well over this distribution.
    \item Reject graphs that deviate significantly from $\gnp$, as the subsequent algorithms may not work well on these graphs.
\end{enumerate}

Is the network scientist's task possible? As given, this question is not well-defined: \textit{every} graph has some probability of being drawn from $\gnp$, and we are only receiving \textit{one input graph}, so it is impossible to statistically test if it was drawn from the distribution. This leads to the main question of the paper: \textit{Can the validation of random models for real-world data be performed more quickly, while still providing some guarantees?} What can those guarantees be?\footnote{\label{fn:reviewer-refs-1}We thank an anonymous reviewer for proposing this motivation and for pointing us to related references: \cite{milo2002network, DBLP:journals/jmlr/ShervashidzeVPMB09, DBLP:journals/im/GleichO12, seshadhri2012community,  DBLP:conf/www/SeshadhriT19}.} 

Such scenarios (see, for instance, \cite{milo2002network, DBLP:journals/jmlr/ShervashidzeVPMB09, DBLP:journals/im/GleichO12, seshadhri2012community,  DBLP:conf/www/SeshadhriT19}) appear routinely in the pipeline of applying theoretical models to real-world data across various disciplines, such as bioinformatics and social network analysis. On the theoretical side, probabilistic models are proposed to explain the structure of the networks and construct algorithms over; one can find many examples in the book \cite{van2024random}. On the practical side, researchers can collect at most a couple of real-world networks to test the validity of these theoretical models. They then run algorithms on these real-world networks \cite{DBLP:journals/cn/BrinP98, alweiss2021subgraph, blondel2008fast,
DBLP:journals/jmlr/ShervashidzeVPMB09, DBLP:journals/toc/KempeKT15,
DBLP:conf/nips/HuFZDRLCL20}, with the hope that the networks align well enough with the theoretical models.

To validate a probabilistic model, a standard practice in network science is to count the number of occurrences of small structures (also called motifs) --- such as triangles and cliques --- in the real-world network and compare this to counts in the model \cite{bickel2011method, DBLP:journals/im/GleichO12}. If the counts do not match, often different models are considered and constructed. (One example is that preferential attachment \cite{price1976general, barabasi1999emergence} differs from real-world social/societal networks in small clique counts \cite{fronczak2002higher}, and thus later models were constructed to match the real-world counts \cite{holme2002growing}.) Designing network models whose motif counts match real-world networks has been the subject of intense study \cite{frank1986markov, snijders2006new, hunter2008goodness, seshadhri2012community, jacob2015spatial}. Furthermore, the practice of validating models via motif counts is so widely adopted that there is an entire field named after this objective, known as ``motif analysis'' in network science; see the textbook chapter \cite{Weihe2015} and survey of \cite{DBLP:conf/www/SeshadhriT19}. Moreover, motif counts are also used as a criterion to certify the quasirandomness of graphs; see \cite{chung1989quasi, krivelevich2006pseudo, zhao2023graph}.

The use of motif counting as a method of validating randomness of graphs inspires our concrete problem: Given an input graph $G$, we want to reject graphs such that the number of counts of a motif $H$ (such as a triangle or $k$-clique) does not match the expected count under $\gnp$, while still accepting most graphs $G \sim \gnp$. This is an instance of what we call a \textit{quality control problem}, which we define in the next section (\Cref{def:quality-control-intro}). Quality control problems are characterized by their asymmetric accept and reject conditions: the acceptance condition is an average-case requirement and the rejection condition is a worst-case requirement. Our definitions are inspired by Feige's work on refutation of random constraint satisfaction problems \cite{DBLP:conf/stoc/Feige02} and Rubinfeld and Vasilyan's work on testable learning \cite{rubinfeld2023testing}; we describe the connection more in \Cref{sec:intro-related-work} and \Cref{sec:appendix-related-work}. Quality control problems have connections to average-case complexity, property testing, quasirandomness, and several other areas of study. 

\subsection{Quality Control problems}

A generic quality control problem is given by a distribution $\mathcal{D}$ supported, say, on $\{0,1\}^n$, and a function  $\rho:\{0,1\}^n \to \R_{\geq 0}$ that takes a value in $1 \pm \varepsilon$ with high probability on the distribution $\mathcal{D}$ and some $\epsilon > 0$. 
Our goal is to design a quality control algorithm $A$ that takes as input $x \in \{0,1\}^n$ and outputs an accept/reject verdict in time $t(n)$.
The completeness property that we would like is that $A(x)$ accepts with high probability when $x \sim \mathcal{D}$, where the probability is taken over the distribution $\mathcal{D}$ and the coins of the algorithm. Thus, $A(x)$ is likely to pass \textit{almost all} $x$ drawn according to $\mathcal{D}$. 
The soundness condition we would like is that, 
for \textit{all} $x$ such that $|\rho(x) - 1| > \epsilon$ (i.e., ``low quality''), $A(x)$ rejects with high probability over the coins of the algorithm. We define quality control algorithms formally below.

\begin{definition}[Quality control algorithms]\label{def:quality-control-intro}
    Let $\mathcal{X} = \{\mathcal{X}_n\}_n$ be a family of domains and $\mathcal{D} = \{\mathcal{D}_n\}$ be a family of distributions, where $\mathcal{D}_n$ is supported over $\mathcal{X}_n$. For function $\rho:\mathcal{X} \to \R_{\geq 0}$ and small constant $\varepsilon > 0$ such that, for all $n \in \mathbb{N}$, $\mathbb{P}_{x \sim \mathcal{D}_n}\left[\left| \rho(x) - 1 \right| \leq \varepsilon \right] \geq 1 - o(1)$, a $(\mathcal{D}, \rho, \ep)$-quality control algorithm $Q: \mathcal{X} \to \{\textsc{Accept}, \textsc{Reject}\}$ must satisfy the following:
    \begin{enumerate}
        \item (Completeness) For all $n \in \mathbb{N}$, $\mathbb{P}_{x \sim \mathcal{D}_n; R}\left[ Q(x) = \textsc{Accept}\right] \geq 1 - o(1)$, where the probability is taken over distribution $\mathcal{D}_n$ and the randomness $R$ used by algorithm $Q$.\footnote{Here, we use $o(1)$ to mean that a quantity approaches $0$ as the runtime of the algorithm tends to infinity. (Formally, the quality control algorithm takes a error parameter $\tau$ as input, and for every $\tau>0$ there exists a constant $c$ such that the quality control algorithm with error parameter $\tau$ accepts $x\sim \mathcal{D}$ with probability at least $1-\tau$ and the runtime of the quality control algorithm with parameter $\tau$ is at most $c$ times larger than the runtime with error parameter set to some fixed constant, say, $1/3$.)}
        \item (Soundness) $\forall x$ satisfying $\left|\rho(x) - 1 \right| > \varepsilon$,  $\mathbb{P}_{R}\left[ Q(x) = \textsc{Reject}\right] \geq 2/3$.
    \end{enumerate}
    The $(\cD, \rho)$-quality control problem is that of designing a $(\cD, \rho, \ep)$-quality control algorithm for every $\ep > 0$. We say that the $(\cD, \rho)$-quality control problem is solvable with time complexity $t(\cdot)$ and query complexity $q(\cdot)$ if, for every $\ep > 0$, there is a $(\cD, \rho, \ep)$-quality control algorithm that runs in $O_{\ep}(t(\cdot))$ time and makes $O_{\ep}(q(\cdot))$ queries.
\end{definition}

The asymmetry of the definition is intended to give us enough room to exploit properties of the presumed distribution to test for quality, but also assures the output preserves high quality. Indeed, for this reason, previous works in the literature, notably the work of Feige~\cite{DBLP:conf/stoc/Feige02}, have studied questions that can be framed as examples of quality control problems, but to our knowledge, this is the first time the {\em class of problems} has been formally introduced.
(We describe the previous works and broader connections to average-case complexity in \Cref{sec:intro-related-work} and \Cref{sec:appendix-related-work}.) In \Cref{sec:intro-techniques}, we explain how the definition's asymmetry leads to a useful composability property for quality control problems that is exploited in our case study. At a high level, this composability property allows us to ``filter out'' instances that do not possess properties held by almost all $x \sim \mathcal{D}$, allowing us to capitalize on these additional
structural properties.

Other models that may seem applicable in this setting turn out not to provide the right answers. For example, one may hope to apply methods of distribution testing to test if the input $x$ was drawn from $\mathcal{D}$. However, this method assumes access to a sequence of samples from the purported distribution, while we have access to \textit{just one} sample. Another model that may apply to our setting is sublinear algorithms for approximating $\rho(x)$. This, of course, would be valid, but our hope is that a quality control algorithm could be \textit{more efficient} than those for approximating/certifying $\rho$ on worst-case inputs. (Again, more details on related models and problems can be found in \Cref{sec:intro-related-work} and \Cref{sec:appendix-related-work}.) In what follows, we introduce the quality control problem for motif counting over random graphs.

\subsection{Counting motifs in random graphs: problems and our results}\label{sec:motifs-description-intro}

We study quality control for the task(s) of ``motif counting'' on well-studied models of random graphs, primarily the \erdosrenyi random graph model. We start by describing our results for the special case where $H$ is a clique on $k$ vertices, before turning to the more general case. For every finite graph $H$, we define a function $\rho_H$ where $\rho_H(G)$ counts the number of induced copies of $H$ in $G$, normalized by its expectation under $\gnp$.  

\vspace{-2ex}
\paragraph{k-Clique Counting} 
When $H = K_k$ is the complete graph on $k$ vertices, we use $\mu_k := \binom{n}k p^{\binom{k}2}$ to denote the expected number of $K_k$ in $\gnp$ and let $\rho_k(G)$ denote the number of $k$-cliques in $G$ divided by $\mu_k$. A quality control algorithm for the problem $(\gnp, \rho_k)$ is thus required to accept graphs from $\gnp$ with probability $1-o(1)$ while rejecting every graph $G$ with $\rho_k(G) \not\in [1-\epsilon,1+\epsilon]$ with high probability (over the algorithm's internal coins). This goal is weaker than either of the component problems of determining whether a graph $G$ is sampled from $\gnp$ (which is not possible with a single graph) or determining whether $\rho_k(G) \in [1-\epsilon,1+\epsilon]$. 

First, are there any natural lower bounds from the literature on the component problems? For worst-case approximate $k$-clique counting, with $\Theta(n^2p)$ edges and $\Theta(n^kp^{\binom{k}{2}})$ cliques, \cite{eden2018faster,DBLP:journals/siamcomp/EdenRS20} prove a query complexity lower bound of $\Omega_k(\min\{n^2p,p^{-k(k-2)/2}\})$. This is $p^{-\Omega(k^2)}$ when $p^{-k(k-2)/2}\leq n^2p$. This lower bound may seem like a natural bottleneck in our setting if one considers the simple algorithm that samples sets of $k$ vertices in $G$ randomly and estimates the frequency with which this sample forms a $k$-clique. Our first theorem shows that we can perform significantly better (in the regime where $k$ is a fixed constant and $p$ is a slowly shrinking function of $n$) and get a query complexity and running time of $p^{-O(k)}$ for the quality control problem $(\gnp, \rho_k)$. Specifically, we show:

\begin{theorem}\label{thm:intro-cliques}
    There exist constants $c_1, c_2, c_3$ such that for every constant $k\geq2$ and $0<p\leq1/2$, the quality control problem $(\gnp, \rho_k)$ corresponding to estimating the number of $k$-cliques in a graph $G$ supposedly drawn according to $\gnp$ is solvable in $O(p^{-c_1 k})$ queries and time given adjacency matrix access to $G$, provided $n\geq B_{k,\varepsilon,\tau}p^{-c_2k}$.  Furthermore, this problem cannot be solved with $o(p^{-c_3 k})$ adjacency matrix queries.
\end{theorem}

This theorem is formalized and proved as \Cref{thm:k-cliques} in \Cref{sec:query-complexity-cliques}.\footnote{We state all of our results for the case of $k$ constant. Extending the running-time bounds to growing $k$ requires quantitative bounds on the constants in the counting lemma.} We stress that the result is non-trivial both in terms of query complexity, where it beats the $p^{-\Omega(k^2)}$ lower bound for sublinear algorithms on worst-case inputs in this range, and in terms of the running time. 

\vspace{-1ex}
\paragraph{Triangle Counting}
While the results above are tight up to a universal constant factor in the exponent, for the specific setting of $k=3$ (triangle counting), we obtain an \textit{exact} bound on the exponent (assuming adjacency list access to the graph). Below, let $\rho_\Delta := \rho_3$ be the normalized triangle counts in a graph.

\begin{theorem}\label{thm:intro-triangles}
    The quality control problem $(\gnp, \rho_\Delta)$ corresponding to estimating the number of triangles in a graph $G$ supposedly drawn according to $\gnp$ is solvable in $\widetilde{O}(1/p)$ queries and time given adjacency list, adjacency matrix, and degree query access to $G$, provided $np=\omega(1)$. Furthermore, this problem cannot be solved with $o(1/p)$ sorted adjacency list queries.
\end{theorem}

Thus up to lesser order terms, we get a tight bound of $\widetilde{\Theta}(p^{-1})$ on the task of triangle counting. Again this outperforms the worst-case query lower bound of $\Omega(\min\{n^2,p^{-3}\})$ for approximate triangle counting~\cite{eden2017approximately}. This theorem is formalized as \Cref{thm:triangles-usingliterature} in \Cref{sec:triangles}.

\vspace{-1ex}
\paragraph{General Motif Counting}
Next, we describe our results for general motif counting in $\gnp$. Our result here pins the exponent of $1/p$ in the query complexity and running time of algorithms counting the number of copies of $H$ (for an arbitrary graph $H$). This exponent is $\Delta(H)$, the maximum degree of any vertex in $H$. So, while counting $k$-cliques or $k$-stars has (query/time) complexity $p^{-\Theta(k)}$ complexity, counting $k$-paths has $p^{-\Theta(1)}$ complexity. 

\begin{theorem}\label{thm:intro-motifs}
    There exist constants $c_5, c_6, c_7$ such that for every fixed graph $H$ on $k\geq2$ vertices and $0<p\leq1/2$, the quality control problem $(\gnp, \rho_H)$ corresponding to estimating the number of copies of $H$ in a graph $G$ supposedly drawn according to $\gnp$ is solvable in $O(p^{-c_5 \Delta(H)})$ queries and running time given adjacency matrix access to $G$, provided $n\geq B_{H,\varepsilon,\tau}p^{-c_6k}$. Furthermore, this problem cannot be solved with $o(p^{-c_{7} \Delta(H)})$ adjacency matrix queries.
\end{theorem}

The implicit constants may depend on the fixed motif, accuracy, and completeness error. This theorem is formalized and proved as \Cref{thm:qc-motif-in-section} in \Cref{sec:query-complexity-motifs}, which in turn is based on techniques that are proven in \Cref{sec:any-random-graph-family}.

\vspace{-1ex}
\paragraph{Beyond \erdosrenyi Graphs}
Our study of quality control of motif counts extends to a broad family of distributions over graphs that have concentration of the counts of small motifs. We prove the following result about $k$-clique counting for several well-studied distributions over graphs. We first define one of these well-studied distributions.

\begin{definition}[Stochastic block model]\label{def:SBM}
    A stochastic block model with $m$ communities and $n$ vertices, $n_i$ in community $i$ for $i \in [m]$, is a random graph model generated according to the following procedure. For each pair $(i, j) \in [m] \times [m]$, fix an edge probability $p_{i j}$. For each edge $e = (u, v)$ in the graph, if $u$ is in community $i$ and $v$ is in community $j$, add an edge independently between $u$ and $v$ with probability $p_{i, j}$.
\end{definition}

\begin{theorem}\label{thm:general-distributions-intro}
    Consider $(\mathcal{D}_n, \rho_k)$-quality control, for $\rho_k(G)  = C_k(G) / \mathbb{E}_{G' \sim \mathcal{D}_n}\left[C_k(G') \right]$. Define parameters $n, p$ where $n \geq p^{-c_{8} k}$ for some constant $c_{8}$. For the following settings of $\mathcal{D}_n$, the query complexity of $(\mathcal{D}_n, \rho_k)$-quality control is $1/p^{O(k)}$:
    \begin{enumerate}
        \item $\mathcal{D}_n$ is a Stochastic Block Model with $n$ vertices, a constant number of communities each of size $\Omega(n)$, and minimum in-community/between-community edge probability at least $p$.
        \item $\mathcal{D}_n$ is the uniform distribution over $d$-regular graphs on $n$ vertices, where $d=np$ is an integer and $nd$ is even.
    \end{enumerate}
\end{theorem}

This theorem is formalized and proved as \Cref{cor:specific-other-models-in-section} in \Cref{sec:any-random-graph-family}. It follows from a more general quantity associated with distribution that we define called ``Frequency increments'' (see \Cref{def:r-j-H}) and a result we prove (see \Cref{thm:general-graph-motif-count-qc}) which shows that our results hold with query complexity shrinking with the frequency increments of the underlying distribution. See Table 1 for our results and a comparison to the worst-case setting.

\begin{table}[t]
\centering
\small
\begin{tabular}{|l|l|l|}
\hline
\textbf{Motif \& distribution} & \textbf{Quality control} & \textbf{Worst-case comparison} \\
\hline
(*) Triangle, $\gnp$ & $\widetilde{\Theta}\left(p^{-1}\right)$ queries/time & $\Theta\left( p^{-3}\right)$ queries/time \cite{eden2017approximately}\\
\hline
$k$-clique, $\gnp$ & $p^{-\Theta(k)}$ queries/time & $p^{
-\Omega(k^2)}$ queries/time \cite{eden2018faster}\\
\hline
Cycle on $k$ vertices, $\gnp$ & $p^{-\Theta(1)}$ queries/time & $p^{
-\Omega(k)}$ queries/time \cite{assadi2018simple} \\
\hline
Star on $k$ vertices, $\gnp$ & $p^{-\Theta(k)}$ queries/time & $\widetilde{O}(n/p)$ queries/time \cite{assadi2018simple}; incomparable \\
\hline
General motif $H$, $\gnp$ & $p^{-\Theta(\Delta(H))}$ queries/time & Depends on motif \\
\hline
\textbf{$k$-clique on non-$\gnp$ examples:} & & \\
\hdashline
$k$-clique, stochastic block model & $p^{-O(k)}$ queries & $p^{-\Omega(k^2)}$ queries/time (from special case of $\gnp$)\\
with minimum edge probability $p$ & & \\
\hdashline
$k$-clique, $np$-regular graph & $p^{-O(k)}$ queries & $p^{-\Omega(k^2)}$ queries/time \cite{eden2018faster}\\
\hline
\end{tabular}
\caption{
The table gives bounds for quality control and compares them with worst-case approximate counting at the corresponding target count. The worst-case query bounds are capped by reading the graph and powers of $1/p$ are stated in the uncapped range. Cycle and star comparisons concern non-induced copies. The lower bound for general cycles concerns colored copies. In the worst-case column, $p$ parameterizes the target motif count $\mathbb{E}_{G' \sim \gnp}[C_H(G')]$, and the input itself is an arbitrary graph. Quality control uses adjacency matrix access except in the row denoted with (*). The worst-case comparisons use the access models of the cited papers, which may also allow adjacency list, degree, and random edge queries. The stochastic block model row assumes constantly many communities each of size $\Omega(n)$, and its worst-case comparison uses the special case in which all edge probabilities equal $p$. For motifs on $k$ vertices, the quality control results hold when $p \geq n^{-c/k}$ for some constant $c$. }
\end{table}

\subsection{Technical contributions}\label{sec:intro-techniques}

Here we describe some of the challenges in proving \Cref{thm:intro-cliques} and the ``quality control'' and ``random graph'' insights that allow us to overcome the challenges.

\vspace{-1em}
\paragraph{Composability.}
We start with a general insight that may be broadly useful in the design of quality control algorithms, which we call the composability property of quality control.
\begin{quote}\label{composability}
    \textbf{The \textit{composability} property of quality control:} Suppose $\rho(x), \tau(x) \approx 1$ for $x \sim \mathcal{D}$ with probability $1 - o(1)$. A quality control algorithm $A$ for a problem $(\mathcal{D}, \rho)$ can always run a quality control algorithm $A'$ for the problem $(\cD,\tau)$ and filter out (reject) instances that are rejected by $A'$, where $\tau$ may be totally unrelated to $\rho$ (as long as the distributions in the two problems are the same).
\end{quote}
This is possible since the filtering does not hurt soundness at all (we simply reject more inputs), while it only affects the completeness by a $o(1)$ additive factor, which is completely acceptable with our definition of completeness in quality control. The stress here is that $\tau$ could be any function, related or unrelated to $\rho$. Of course, to make this idea useful in our context, we need to identify useful functions $\tau$ that make quality control for $(\cD, \rho)$ easier. For example, we use smaller $\ell$-clique counts to filter out for $k$-clique counts, for $\ell < k$, and this turns out to be crucial for our bounds.

\vspace{-1em}
\paragraph{The naive algorithm and concentration challenge.}
Turning to the problem of counting $k$-cliques, the naive idea of sampling random $k$-sets and estimating the frequency with which they are $k$-cliques in $G$ is too slow: it would take $p^{-\Omega(k^2)}$ samples to find even a single clique in $G \sim \gnp$. Our first idea is to sample a single large(r) set $S$ of the vertices and use the frequency of $k$-cliques in $G|_S$ to help estimate the frequency of the number of $k$-cliques in $G$. (Here and later we use the term ``frequency of $k$-clique'' in some graph $G'$ with $n'$ vertices to be the number of $k$-cliques divided by $\binom{n'}k$.) Over the choice of $S$, the expected frequency of $k$-cliques in $G|_S$ is clearly exactly equal to the frequency of $k$-cliques in $G$. However, this frequency is not well-concentrated for worst case $G$ (over the choice of $S$) until $|S| = p^{-\Omega(k^2)}$, which is exactly the setting we wish to avoid. 

A simple example illustrating the absence of concentration is if $G$ is a complete graph on $n/2$ vertices and has $n/2$ isolated vertices. For a typical set $S$, switching one vertex can affect the number of $k$-cliques in $S$ additively by $\approx (|S|/2)^{k-1}$, indicating a very high variance. Of course, such a $G$ is very unlikely to be produced by $\gnp$ (it has too many edges, for instance) and so should be rejected perhaps for other reasons. This would mesh well with the notion of composability. But the central question is what other features (or functions $\tau$) can we use to filter out such $G$s.

\vspace{-1em}
\paragraph{Filtering for small cliques via high-probability bounded differences.} 

Our first idea at this stage is to filter out graphs that differ from $\gnp$ in terms of $\ell$-clique counts for \textit{any} $1 \leq \ell < k$. Using $(\gnp,\rho_\ell)$-quality control, $1 \leq \ell < k$, as pre-filters certainly rules out the specific counterexample above, but does it work generally? 

One potential way to take advantage of the correctness of $\ell$-clique counts is to apply McDiarmid's method of bounded differences (\cite{mcdiarmid1989method}; see \Cref{thm:mcdiarmid}). Suppose that we know that $G|_S$ has exactly the same frequency of $(k-1)$-cliques as $G$, i.e., $p^{\binom{k-1}2}$ and $G$ has been accepted by pre-filter $(\gnp,\rho_{k-1})$. Then the maximum difference that switching a single vertex of $S$ can make to the number of $k$-cliques in $S$ is $O_k(|S|^{k-1}\cdot p^{\binom{k-1}2})$. This is a small enough number that, with McDiarmid's method of bounded differences and some technical care, we can conclude that a sample $S$ of size roughly $p^{-O(k)}$ suffices for the frequency of $k$-cliques in $G|_S$ to be concentrated. 

Unfortunately, the assumption that $G|_S$ has the same frequency of $(k-1)$-cliques as $G$ for \textit{every} $S$ is totally unrealistic, and so the argument above seems doomed. Indeed, the most one can hope is that the frequency of $(k-1)$-cliques in $G|_S$ is roughly the same as in $G$ for some $1-\delta$ fraction of the $S$s. In principle, we could now use a ``bounded difference with high probability'' extension of McDiarmid's inequality due to Combes (\cite{combes2024extension}; see \Cref{thm:mcdiarmid-whpdifferences}) to get concentration for smaller sample sizes. However, even this hope has a quantitative problem. A direct algorithm to estimate $\delta$ would require sampling $1/\delta$ different $S$s, and Combes' theorem in our context requires $\delta = p^{\Omega(k^2)}$. A direct use of Combes' method therefore would require $|S| = p^{-\Omega(k^2)}$.

It turns out we can test the condition $\delta = p^{\Omega(k^2)}$ indirectly using a number of queries that scales with $p^{-O(k)} \cdot \log(1/\delta)$ via an inductive use of Combes' theorem. To this end, we define a graph $G$ to be {\em exponentially robustly quasirandom with respect to $\ell$-clique counts} if the probability that a random sample of $s$ vertices has the wrong $\ell$-clique count (relative to a sample from $\gnp$) is exponentially small in $s$. We prove a concentration lemma (\Cref{lem:concentration-intro}) via Combes' theorem which tells us that if a graph $G$ is exponentially robustly quasirandom with respect to $(\ell-1)$-clique counts and we sample a subgraph of $G$ of size $p^{-O(\ell)}$ then the sampled $\ell$-clique count differs from the scaled global count by a small fraction of its $\gnp$ expectation with probability $1 - \exp(-p^{O(\ell)} s)$. We use this to give an inductive proof that an exponentially robustly quasirandom graph with respect to $(\ell-1)$-clique counts can be ``tested'' (in the quality control sense) for being an exponentially robustly quasirandom graph with respect to $\ell$-clique counts, using $p^{-O(\ell)}$ queries. To do so, we simply count the number of $\ell$-cliques in the sampled subgraph and reject if this number is not in range. We prove this is sufficient for the inductive step via two cases. First, if a graph $G$ is exponentially robustly quasirandom for $(\ell-1)$-cliques and has the (roughly) correct count of $\ell$-cliques, the concentration lemma (\Cref{lem:concentration-intro}) tells us that the sampled subgraph will also have the correct $\ell$-clique count with high probability, and so $G$ is exponentially robustly quasirandom for $\ell$-cliques. On the other hand, if $G$ is exponentially robustly quasirandom for $(\ell-1)$-cliques but its $\ell$-clique count is $\varepsilon$-far from the correct count, the concentration lemma implies that the $\ell$-clique count will be $\varepsilon/2$-far in the sampled subgraph with high probability, and so $G$ will be rejected with high probability. 

With this inductive step, we may now conclude that if a graph $G$ has the wrong $k$-clique bound, then our test must reject it with high probability. Specifically, there must be a smallest $\ell$ for which it also has the wrong $\ell$-clique bound; but, if $G$ has not failed the tests thus far, then it is exponentially robustly quasirandom with respect to $(\ell -1)$-clique counts and so will be rejected in the $\ell$th iteration with high probability.
To our knowledge, this is the first algorithmic application of Combes' theorem outside of statistics (where it has been used previously, see e.g., \cite{DBLP:conf/nips/YanYS17}), despite the result first appearing online in 2015.

\vspace{-1em}
\paragraph{Algorithmic complexity.} But we are not done yet, since estimating the frequency of $k$-cliques in $G|_S$ would naively take $|S|^{O(k)}$ time and this would translate to a $p^{-O(k^2)}$-time algorithm, albeit with only $p^{-O(k)}$ queries.
To obtain a time-efficient algorithm from this, we use tools from \textit{graph quasirandomness} \cite{thomason1987pseudo, chung1989quasi}.
Specifically we use the fact that there is a deterministic graph property, referred to in the literature as {\em jumbledness}, 
such that for graphs possessing this property, $\ell$-clique-counts approximately match $G_{|S|, p}$ for all $1 \leq \ell \leq k$. Jumbledness is a sufficient, but not necessary, condition for the $\ell$-clique counts to match \erdosrenyi graphs. Furthermore, while jumbledness itself may not be easy to test, there exists an even simpler condition that can be verified in $s^{O(1)}$ time which implies jumbledness, and a $1-o(1)$ fraction of \erdosrenyi graphs satisfy this simpler condition; see \Cref{cor:jumbled-subgraphs}. So we can leverage composability again to perform quality control on the simpler sufficient condition for the jumbledness property. Applying this to $G|_S$ where $|S| = p^{-O(k)}$ now gives us a $p^{-O(k)}$ time algorithm for the quality control problem $(\gnp,\rho_k)$. We overview this step also in more depth in \Cref{sec:proof-overviews} and give full proofs in \Cref{sec:algorithmic-complexity}. 
\\\\
We emphasize that the quality control framework is being crucially utilized to obtain the query complexity and algorithmic speedups. Within $S$, we are using composability to certify a sufficient degree/codegree condition for jumbledness. Quality control for the sufficient condition performs quality control for jumbledness. Quality control for jumbledness performs quality control for $\ell$-clique counts, for $2 \leq \ell \leq k$. This can then be lifted to quality control for $\ell$-clique counts in $G$ via exponential quasirandomness. Then, in turn, this acts as a pre-filter for quality control for the count of $k$-cliques compared to \erdosrenyi graphs. Even though our algorithm is simple (pick a subgraph and look at degrees/codegrees), its correctness relies on the filtering power of the quality control framework, which then requires us to develop some new technical machinery. 

\subsubsection{Beyond random graphs and motifs: general techniques}

We highlight three core technical ideas that should be broadly useful to quality control problems, none of which are specific to motif counts on random graphs: The first of these is the \underline{composability} of quality-control algorithms, which we feel is useful well beyond the contexts of graph problems or $\gnp$. The second is the notion of \underline{exponentially robust quasirandomness} which can be applied to other distributions and settings as well. Finally we note that the powerful theory of {\em quasirandomness}, i.e., the study of deterministic conditions satisfied by most objects from a distribution, can be a useful algorithmic tool as well. 

\subsection{Related work}\label{sec:intro-related-work}

There is an extensive amount of work in the literature that relates to quality control. See \Cref{sec:appendix-related-work} for an in-depth review of the connections between these areas and quality control. We briefly list the different areas here.

\begin{description}
     \item[Average-case complexity.] Average-case problems \cite{DBLP:journals/siamcomp/Levin86} can often be framed as quality control problems. On the other hand, quality control can lead to average-case algorithms; see \Cref{sec:average-case-complexity}.
    
    \item [Testing assumptions of algorithms.] Testable learning \cite{rubinfeld2023testing} and testable approximate counting \cite{TestableAlgorithmsEden2025} both ask: given an algorithm, can we output YES if the algorithm is safe on an input and NO if it is not. The quality control model is, instead, asymmetric: it allows algorithms not to accept all safe inputs --- instead, just those coming from a distribution.

    \item[Distribution testing.] Distribution testing is also concerned with the problem of whether an input comes from a distribution. However, distribution testing assumes access to many samples from the distribution, while we focus on \textit{sublinear} access to a \textit{single} sample.

    \item[Quasirandomness.] Quasirandom properties of graphs \cite{thomason1987pseudo, chung1989quasi} are deterministic properties held by \erdosrenyi graphs, and are a useful ingredient in the design of quality control algorithms.

    \item[Motif counting.] Motif counting is essential to graph theory both in theory \cite{zhao2023graph} and practice \cite{DBLP:conf/www/SeshadhriT19}. Sublinear algorithms for motif counting must work for worst-case inputs, while we leverage the distribution of inputs.
\end{description}

\section{Proof overviews and open problems}\label{sec:pf-overview-open-problem}

In this section we describe our quality control algorithms and lower bounds for clique counting and general motif counting and give an overview of their analyses. (This section complements \Cref{sec:intro-techniques} which only highlighted the challenges and hinted at the tools necessary to overcome them. Nevertheless a reader may find this section more readable after reading \Cref{sec:intro-techniques}.)

\subsection{Proof overviews}\label{sec:proof-overviews}

\vspace{-1ex}
\paragraph{Clique Counting: Upper Bound}

The algorithm and its analysis of quality control for $k$-cliques compared to $\gnp$ were motivated in \Cref{sec:intro-techniques}. Here, we give the ingredients of the proof. Notice that our final algorithm, described below, turns out to be extremely simple. While we want to test clique counts via our algorithm, recall from \Cref{sec:intro-techniques} that degree and common neighbor counts in a sampled subgraph imply a property called ``graph jumbledness,'' which in turn implies bounds on approximate clique counts in the subgraph. 
We show that these bounds in the subgraph yield information about the global clique count.

\begin{algorithm}[H]
\caption{Algorithm (informal) to perform quality control on the $k$-clique count of $G$ compared to $\gnp$ in $p^{-O(k)}$ queries. See \Cref{alg:clique-quality-efficient} for the more formal algorithm.}
\vspace{1em} ~~ \textsc{Clique-Quality-Informal}: On inputs $G$, $n, p, \varepsilon$, $\gnp$, and $k$

\begin{quote}
Sample a random (multi-)set $S$ of $s \approx p^{-O(k)}$ vertices of $G$. 

For each $u \in S$, check that the degree of $u$ in $S$ is $\approx (s - 1) p$. For every $(u, v) \in S \times S$ with $u\ne v$, check that the number of common neighbors of $(u, v)$ is $\approx (s - 2) p^2$. 

If any vertex in $S$ or pair in $S\times S$ does not satisfy this, reject and halt. Else, accept. 
\end{quote}
\label{alg:clique-quality-intro}
\end{algorithm}
\vspace{-2ex}

First, looking at the test, it is straightforward to prove (and we do) that $G \sim \gnp$ will pass with high probability. Thus, it suffices to prove the soundness, i.e., analyzing the probability that the algorithm rejects a graph whose $k$-clique count doesn't match $\gnp$.

\noindent\underline{Soundness overview:} Suppose the $k$-clique is $\varepsilon$-far from the correct count. Thus, there exists some smallest $\ell \leq k$ such that the $\ell$-clique is $\varepsilon$-far from the correct count. For this $\ell$ we claim the following: (1) with high probability over $S$, the graph $G|_S$ does not have the right frequency of $\ell$-cliques. (2) If $G|_S$ does not have the right frequency of $\ell$-cliques then our algorithm will reject. We prove these claims in reverse order below. 

\noindent\underline{Proof of (2): Counting $\ell$-cliques in the random set.} Note that if our algorithm accepts $G|S$ then it certifies a certain quasirandomness condition on $G|_{S}$: namely, that the degrees in $G|_S$ are all $\approx (s-1)p$ and common neighbor counts are all $\approx (s-2) p^2$. This allows us to use the theory of graph quasirandomness and jumbledness \cite{thomason1987pseudo, conlon2014extremal} to show that $G|_S$ is ``jumbled'' (a notion we don't define here) and hence the frequency of $\ell$-cliques is roughly that on $\mathcal{G}_{s,p}$. In contrapositive form this implies that if $G|_S$ does not have the right frequency of $\ell$-cliques our algorithm will reject (with probability 1). 

\noindent\underline{Proof of (1): Turning a global statement into a (high probability) local one.} We will now see how to turn the global statement, that $G$ does not have the right $\ell$-clique bound (while having the right $\ell'$-clique count for every $2\le \ell' < \ell$), into a high probability local statement, i.e., that $G|_S$ does not have the right $\ell$-clique bound with high probability over the choice of $S$. 
A notion that turns out to be critical for this step is that of \textit{exponentially robust quasirandomness}, (as motivated in \Cref{sec:intro-techniques}) which we define below. 

\begin{definition}[Exponentially robustly quasirandom (informal)]\label{def:exp-qr-intro}
    We say that a graph $G=([n],E)$ is exponentially robustly quasirandom with respect to the count of $\ell$-cliques if for every sufficiently large $s$ we have that $$\Pr_{S \sim [n]^s}\left[ \text{\# $\ell$-cliques in $S$} \approx \binom{s}{\ell}p^{\binom{\ell}{2}} \right] \geq 1 - \exp(- p^{O(\ell)} s).$$
\end{definition}

The formal definition is given in \Cref{def:exp-qr}. Typical definitions of graph quasirandomness involve some property of
graphs that requires that some local condition holds  {\em with probability one} 
(i.e., for all choices of some variables). As a result, these properties are \textit{not} testable in sublinear queries and time. Exponentially robust quasirandomness is a relaxation of quasirandomness that holds ``with high probability,'' is efficiently certifiable, and is strong enough to give us what we need with
respect to motif counts. In those three ways, it is a ``robust'' version of graph quasirandomness. 

A key step in our use of the notion of exponentially robust quasirandomness is that it can be tested inductively efficiently, as asserted by the following lemma:

\begin{lemma}[Inductive lemma (informal)]\label{lem:inductive-intro}
    Consider any parameter $s > 0$. Define tester $T_{\ell'}$ that samples (multi-set) $S$ with $s \approx p^{-O(\ell')}$ vertices of $G$ and accepts if and only if the $\ell'$-clique count in $S$ is consistent with the expectation for $\gnp$. 
    The following two items hold. 
    \begin{enumerate}
        \item If the graph $G$ is exponentially robustly quasirandom with respect to the count of $(\ell' - 1)$-cliques and $\ell'$-cliques, $T_{\ell'}$ accepts with probability at least $1 - \exp(-p^{O(\ell')} s)$.
        \item If the graph $G$ is exponentially robustly quasirandom with respect to the count of $(\ell' - 1)$-cliques and not exponentially robustly quasirandom with respect to the count of $\ell'$-cliques, $T_{\ell'}$ rejects with probability at least $1 - \exp(-p^{O(\ell')} s)$.
    \end{enumerate}
\end{lemma}

The quantified version of the inductive lemma is given as \Cref{lemma:inductive-cliques}.
Condition 1 of the lemma above is immediate given that $G$ is exponentially robustly quasirandom with respect to $\ell$-clique counts. However, as elaborated on already in \Cref{sec:intro-techniques}, the second condition, with the small query complexity and high rejection probability, is not obvious from the definition. (Indeed the definition only requires that $G|_S$ should be rejected with probability $\exp(-p^{O(\ell')}s)$ and not $1-\exp(-p^{O(\ell')}s)$.) This condition, however, follows easily from the following ``concentration lemma'' which argues that if $G$ is exponentially robustly quasirandom with respect to $(\ell'-1)$-clique counts then with high probability the frequency of $\ell'$ cliques in $G|_S$ differs from that in $G$ by a small fraction of the frequency expected under $\gnp$. 

For a multiset $S$, let $C_{\ell}(S)$ be the number of $\ell$-cliques in $S$ (defined formally in \Cref{def:C-ell-S}). Let $A_{s,\ell}$ scale the global count to its expectation on $S$, and let $E_{s,\ell}$ be the expected sampled count under $\gnp$ (see \Cref{eq:adjustment-global-to-local,eq:expected-ell-clique-copy-amount}).

\begin{lemma}[Concentration lemma (informal)]\label{lem:concentration-intro}
   For sufficiently large $s$, the following holds for some $\varepsilon > 0$. Suppose that the graph $G$ is exponentially robustly quasirandom with respect to the count of $(\ell - 1)$-cliques. Then,
   $$\mathbb{P}_{S \sim [n]^s}\left(\left|C_{\ell}(S)-A_{s,\ell}C_{\ell}(G)\right|\leq\varepsilon E_{s,\ell}\right) \geq 1 - \exp(-p^{O(\ell)} s).$$
\end{lemma}

This lemma is formalized as \Cref{lem:concentration} and proved in \Cref{sec:proof-concentration}.

The proof of the concentration lemma is the technical core of the proof of \Cref{thm:intro-cliques}. It follows from an application of the following concentration inequality for functions satisfying a ``high-probability'' bounded differences property \cite{combes2024extension}: Consider a function $f: \{0, 1\}^N \to \mathbb{R}$ and a subset of the domain $\mathcal{Y} \subseteq \{0, 1\}^N$. Suppose that there exists a $c \geq 0$ such that for all $\vec{x}, \vec{y} \in \mathcal{Y}$, $\left| f(\vec{x}) - f(\vec{y})\right| \leq c \cdot \#(i : x_i \neq y_i).$ Let $X_1, X_2, \dots, X_N$ be 0/1-valued independent random variables; denote $X = (X_1, X_2, \dots, X_N)$. Let $\delta = \mathbb{P}\left( X \not \in \mathcal{Y}\right)$. Then, for bounded nonnegative $f$ and $\tau$ sufficiently large compared to both $\delta\max(f)$ and $\delta Nc$,
    $\mathbb{P}\left( \left|f(X) - \mathbb{E}(f(X))\right| \geq \tau \right) \lessapprox \delta + \exp\left(- \tau^2/(2 N c^2) \right).$
The formal statement is given as \Cref{cor:mcdiarmid-unconditional-expectation}.

Let $\mu_{\ell}(s):=\binom{s}{\ell}p^{\binom{\ell}2}$.  We choose $\tau = \varepsilon\mu_\ell(s)$. 
To prove the concentration lemma using high-probability bounded differences, we let the random variable $X_i$ be the $i$th vertex sampled (with repetition) by the algorithm for the multiset $S$ of size $s$. These are independent. Let $\mathcal{Y}$ be the subset of size-$s$ multisets with $(\ell - 1)$-clique counts proportional to the expected amount from $\gnp$. By the exponentially robust quasirandomness assumption on the count of $(\ell - 1)$-cliques, and the setting of $s$, $\delta = \mathbb{P}(S \not \in \mathcal{Y})$ is very small. We prove that the ``bounded difference'' $c$ satisfies $c = O_\ell(\mu_{\ell-1}(s))$, and so \Cref{cor:mcdiarmid-unconditional-expectation} implies that $C_{\ell}(S)$ differs from its scaled global count by more than $\tau$ with probability \textit{at most} (omitting constants):
$$\delta + \exp\left(- \tau^2/(2 N c^2) \right) \approx \exp\left(- \varepsilon^2 \mu_{\ell}(s)^2 / (s \cdot \mu_{\ell-1}(s)^2 ) \right).$$
Plugging in the values of $\mu_{\ell-1}(s), \mu_{\ell}(s)$ yields the concentration lemma, since $\mu_{\ell}(s)^2 / (s \cdot \mu_{\ell-1}(s)^2 )= p^{O(\ell)} s$. Ultimately, considering all $2 \leq \ell \leq k$, to obtain a low error probability for the algorithm, we need to set  $s \approx p^{-O(k)}$.

\vspace{-2ex}
\paragraph{Triangle Counting}

Our results so far already imply that quality control of the triangle count compared to $\gnp$ takes $1/p^{O(1)}$ adjacency matrix queries and runtime. For triangles, we identify the exact constant in the exponent, proving a $\widetilde{O}(1/p)$ upper bound on the query/time complexity and a $\Omega(1/p)$ lower bound. Our algorithm relies on results from the literature on sublinear approximation of graph parameters \cite{eden2022approximating, eden2018faster}, and correspondingly utilizes more expansive sublinear access to the graph: query access to adjacency lists, the adjacency matrix, and degree queries.

If we were to use an algorithm for estimating the number of triangles in a worst-case graph, this takes $\Theta(\min\{n^2,p^{-3}\})$ queries (ignoring polylogarithmic factors) \cite{eden2017approximately}. Alternatively, if we first perform $(\gnp, \rho_{e})$-quality control, where $\rho_e(G) = e(G)/\mathbb{E}_{G' \sim \gnp}\left[e(G')\right]$, to take advantage of the concentrated number of edges in $\gnp$, then subsequent $(\gnp, \rho_{\Delta})$-quality control on the accepted graphs takes $\Theta(p^{-3/2})$ queries in the uncapped range (ignoring polylogarithmic factors) \cite{eden2017approximately}. To get to $\widetilde{O}(1/p)$, we first perform $(\gnp, \rho_{e})$-quality control and $(\gnp, \rho_{\alpha})$-quality control, where $\rho_{\alpha}(G) = \alpha(G)/\mathbb{E}_{G' \sim \gnp}\left[\alpha(G')\right]$ and in turn $\alpha(G)$ is the arboricity of the graph $G$ (see \Cref{def:arboricity}). On graphs accepted by both of the quality control algorithms, perform $(\gnp, \rho_{\Delta})$-quality control. Using algorithms from \cite{eden2018faster} for the approximation of $k$-clique counts in graphs with a bounded arboricity, we can reduce the bounds to $\widetilde{O}(1/p)$ queries and runtime.

For the lower bound on $(\gnp, \rho_{\Delta})$-quality control, we construct two distributions over graphs: a YES distribution, which is forced to be $\gnp$ since the quality control problem has this as the target distribution, and a NO distribution, which must have a number of triangles out of the range $(1 \pm \varepsilon) \mathbb{E}_{G' \sim \gnp}\left[C_{\Delta}(G')\right]$, where $C_{\Delta}(G)$ is the number of triangles in graph $G$. Since we want the NO distribution to look indistinguishable from $\gnp$ except for in a small region, we plant a clique of size $\ell = O(pn)$ in a graph drawn from $\gnp$. The graph created by planting an $\ell$-clique in $\gnp$ has more than $(1 \pm \varepsilon) \mathbb{E}_{G' \sim \gnp}\left[C_{\Delta}(G)\right]$ triangles, and a query must be made inside the planted clique to identify that the graph does not look like $\gnp$ and is from the NO distribution. Since the input is randomized, it will take $\Omega(1/p)$ queries until a query is made involving a vertex/edge inside the planted clique. The formal proofs can be found in \Cref{sec:triangles}.

\vspace{-2ex}
\paragraph{Motif Counting: Upper Bound}

We extend the principles of the quality control for clique counts to general motifs. We define exponentially robust quasirandomness for general induced labeled motifs, and prove a concentration lemma. The change is the following: The concentration lemma makes an assumption about the count of all $H' \subseteq H$ such that $|H'| = \ell - 1$. 

For motif $H$, let $\mu_H(s)$ be the expected number of labeled induced copies of $H$ in a multiset of size $s$ of vertices chosen from $G \sim \gnp$. The size $s$ of the multiset sampled is chosen so that
$\mu_{H}(s)^2 / (s \cdot \mu_{H'}(s)^2)$ is large for all $H' \subset H, |H'| = \ell - 1$. To understand what $s$ should be, look at the $H'$ maximizing $\mu_{H'}(s)^2$. Since $p \leq 1/2$, this corresponds to the $H'$ obtained by removing the \textit{maximum degree vertex} from $H$. The ratio $\mu_{H}(s)^2 / (s \cdot \mu_{H'}(s)^2)$ will, therefore, depend on $p^{O(\Delta(H))}$, where $\Delta(H)$ is the maximum vertex degree in $H$. From this observation, we'll set $s \approx p^{-O(\Delta(H))}$.

Our quality control results additionally extend to \textit{other distributions over random graphs} that possess good concentration of small motif counts. This is because the soundness of the algorithm is argued over any worst-case graph. The completeness argument only relies on the fact that $\gnp$ has concentration of counts of motifs of size at most $\ell$. The formal proofs can be found in \Cref{sec:query-complexity-motifs} (motifs over $\gnp$) and \Cref{sec:any-random-graph-family} (motifs over general random graph distributions).

\vspace{-2ex}
\paragraph{Motif Counting: Lower Bound}

As in the lower bound for quality control of the triangle count, the YES distribution is $\gnp$. For the NO distribution, plant a structure in a graph drawn from $\gnp$ so that the number of induced labeled copies of a motif $H$ is $> (1 + \varepsilon) \mathbb{E}_{G' \sim \gnp}\left[C_H(G')\right]$, where $C_H(G)$ is the number of copies of $H$ in $G$. We cannot just plant a large $\ell$-clique in $G$ this time, however, since this may not add enough \textit{induced} copies of $H$. Instead, we plant a graph drawn from $\mathcal{G}_{\ell, 1/2}$ for $\ell = O(n p^{\Delta(H)/2})$, and {\em further} connect each of these $\ell$ vertices to other vertices in the graph, with edge probability $1/2$. Given this construction, it takes $\Omega(p^{-\Delta(H)/2})$ queries to distinguish the YES and NO distributions. The choice of $\ell = O(n p^{\Delta(H)/2})$ arises because the vertices in the planted/denser region play the role of the maximum-degree vertex in new $H$ copies we count (and copies of $H' = H$ minus this vertex are taken from existing copies in the $\gnp$ part of the graph). The formal proof can be found in \Cref{sec:motifs-lowerbound}.
\subsection{Open problems and future directions}\label{sec:future-directions}

Our work opens up a variety of future research directions. We begin with open problems that are more directly related to quality control of motif counts for $\gnp$, and then turn to broader questions.

\vspace{-3ex}
\paragraph{Quality control of motif counts} 

First, identifying/proving the precise universal constants in the exponents of our results remains open. Second, when $p = o\left(n^{-1/(c k)} \right)$ for a constant $c$, quality control remains unstudied. Quality control for quality functions that are not concentrated over a distribution has not been explored, but remain interesting. Third, quality control for search instead of decision problems likewise remains interesting and unstudied. Fourth, when sublinear algorithms for counting motifs are studied, it is often natural to study \textit{sampling} motifs as well \cite{biswas2021towards, eden2025approximately, eden2020almost}. While sampling may not fit as naturally with counting in the quality control framework, what can we say about efficient algorithms for sampling motifs in graphs accepted by the algorithm? Finally, there is still more to understand about the runtime for quality control over other random graph distributions.

\vspace{-2ex}
\paragraph{Other directions in quality control} 
In this paper, we performed a case study of our newly-identified class of quality control problems by studying the quality control of motif counts over \erdosrenyi graphs. However, quality control need not be studied only over graphs, nor does it need to be studied in the sublinear access model only. The question of whether an algorithm that works well on average over a distribution can safely be applied to an input is relevant for a variety of different inputs. Quality control is, therefore, natural for various domains, such as problems related to statistics over high-dimensional Gaussian or random matrices. Additionally, understanding implications of quality control in broader algorithmic pipelines --- both from a theoretical and practical perspective --- has many open directions. 

\subsection{Paper outline}

The rest of the paper is organized as follows. 
In \Cref{sec:preliminaries}, we introduce preliminary definitions and results related to quality control, graphs and sublinear access models, quasirandom graphs, properties of \erdosrenyi graphs, and bounded difference concentration inequalities. In \Cref{sec:qc-upperbound}, we study quality control of $k$-clique counts compared to \erdosrenyi graphs. \Cref{sec:qc-upperbound} and \Cref{sec:algorithmic-complexity} incorporate the main technical elements for the query complexity and runtime proofs, which will be applied and extended in later sections for motif counts, other random graph models, and general graph parameters. \Cref{sec:query-complexity-motifs} extends the results to quality control of counts of general motifs compared to \erdosrenyi graphs. The main technical elements of the lower bound proofs in the paper can be found in \Cref{sec:motifs-lowerbound}. We recommend that a reader primarily interested in the core technical ideas of our paper focus on \Cref{sec:qc-upperbound}, \Cref{sec:algorithmic-complexity}, and \Cref{sec:motifs-lowerbound}. \Cref{sec:any-random-graph-family} looks at quality control for motif counts compared to other random graph models with concentration. This section includes the most general forms of our results for quality control of motif counts and the technical lemmas. 

In \Cref{sec:triangles}, we focus on the case of counting triangles compared to \erdosrenyi graphs, relying on results from the literature on sublinear approximate counting to strengthen the quality control result for triangles. Next, in Sections \ref{sec:NC0-properties} and \ref{sec:average-case-complexity} of the Appendix, we provide two alternative perspectives on quality control of motifs. \Cref{sec:NC0-properties} further extends the results for quality control of graph parameters expressible as ``sum-NC$^0$'' circuits (see \Cref{def:sum-nc0}). \Cref{sec:appendix-related-work} contains a more in-depth discussion of the connection between quality control and related areas of study. Next, in \Cref{sec:average-case-complexity}, we describe the interplay between average-case complexity and quality control and show how our quality control algorithm implies an algorithm for approximately counting motifs with average-case runtime $1/p^{O(\Delta(H))}$ where $\Delta(H)$ is the maximum vertex degree of a motif $H$. Lastly, \Cref{sec:appendix} provides proofs of concentration of subgraph counts for the \erdosrenyi graph model and another model used for our lower bound proofs.

\section{Preliminaries}\label{sec:preliminaries}

In this section, we provide a background on relevant definitions, notation, and results related to quality control, graphs, sublinear access models, quasirandom graphs, properties of \erdosrenyi graphs, and concentration inequalities. We also prove some elementary results. The preliminary results that we prove in this paper are \Cref{prop:jumbled-from-deg-codeg}, \Cref{lem:Gnp-motif-concentration}, \Cref{lem:Gnp-clique-concentration},  \Cref{prop:Gnp-jumbled-conditions}, \Cref{cor:Gnp-is-jumbled}, \Cref{cor:jumbled-subgraphs}, and \Cref{cor:mcdiarmid-unconditional-expectation}.

A reader familiar with graph theory and sublinear access models can skip \Cref{sec:preliminary-graph-sublinear}.  \Cref{sec:erdos-renyi-properties} covers various properties we will utilize about \erdosrenyi graphs: motif counts are concentrated around their expectations, arboricity is bounded with high probability, and the graph is jumbled with high probability. \Cref{sec:preliminary-pseudorandom} introduces the notion of \textit{jumbled graphs} from \cite{thomason1987pseudo}, highlighting a known result proving that jumbledness is efficiently verifiable (\Cref{prop:thomason-jumbled}) and implies subgraph counts (\Cref{prop:cfz-jumbled-subgraph}). \Cref{sec:preliminary-concentrationinequality} highlights McDiarmid's inequality \cite{mcdiarmid1989method} and an extension \cite{combes2024extension} that we will crucially rely on.

\subsection{General version of quality control}

Quality control, as defined in \Cref{def:quality-control-intro}, can be formalized in more general scenarios where the quality parameter $\rho$ does not concentrate around $1$ for inputs from the good distribution $\mathcal{D}$.

\begin{definition}[Quality control, general]\label{def:quality-control-general}
    For function $\rho:\mathcal{X} \to \R$, small constant $\varepsilon > 0$, and distribution $\mathcal{D}: \mathcal{X} \to \mathbb{R}_{\geq 0}$, a $(\mathcal{D}, \rho, \ep)$-quality control algorithm $Q: \mathcal{X} \times \R \to \{\textsc{Accept}, \textsc{Reject}\}$ must satisfy the following:
    \begin{enumerate}
        \item (Completeness) $\mathbb{P}_{x \sim \mathcal{D}; R}\left[\forall \tau \text{ such that } |\rho(x) - \tau| \leq \varepsilon, ~~ Q(x, \tau) = \textsc{Accept}\right] \geq 1 - \ep$, where the probability is taken over distribution $\mathcal{D}$ and the randomness $R$ used by algorithm $Q$.
        \item (Soundness) $\forall \tau, \forall x$ satisfying $\left|\rho(x) - \tau \right| > \varepsilon$,  $\mathbb{P}_{R}\left[ Q(x, \tau) = \textsc{Reject}\right] \geq 2/3$.
    \end{enumerate}
    The $(\cD, \rho)$-quality control problem is that of designing a $(\cD, \rho, \ep)$-quality control algorithm for every $\ep > 0$. We say that the $(\cD, \rho)$-quality control problem is solvable with time complexity $t(\cdot)$ and query complexity $q(\cdot)$ if, for every $\ep > 0$, there is a $(\cD, \rho, \ep)$-quality control algorithm that runs in $O_{\ep}(t(\cdot))$ time and makes $O_{\ep}(q(\cdot))$ queries.
\end{definition}

Let's break up this definition to understand why it is the natural generalization of quality control. First, $\mathcal{D}$ is a \textit{good distribution}, as in \Cref{def:quality-control-intro}. Suppose that we have an algorithm $A$ that is designed to perform some task well over inputs drawn from $\mathcal{D}$. Given an arbitrary input $x$, we want to know: if I run algorithm $A$ on input $x$, can I trust the output of the algorithm? The \textit{quality parameter} $\rho(x)$ now represents the typical/expected output of $A$ on an input $x$. The quality control algorithm is given as input $(x, \tau) \in \mathcal{X} \times \mathbb{R}$, which is meant to represent the input we gave to the algorithm $A$ and the output we received from $A$ run on $x$. In the completeness case, as long as $\tau$ accurately represents the output of $A$, we want the quality control algorithm to accept inputs from $\mathcal{D}$ with high probability. In the soundness case, for every input pair $(x, \tau)$ where $\tau$ is far from what we expect from running $A$ on $x$ (i.e., $\rho(x)$), the quality control algorithm must reject with high probability. The quality control algorithm is, therefore, accepting most $(x, \tau)$ corresponding to ``typical'' use of $A$, and rejecting every input $(x, \tau)$ with high probability when the use (either the input $x$ or the output $\tau$) of the algorithm may be untrustworthy.

\subsection{Graphs and sublinear access models}\label{sec:preliminary-graph-sublinear}

We begin by introducing some notation. For a graph $G$, let $V(G)$ be the set of vertices of $G$, and $v(G) = |V(G)|$ be the number of vertices. Let $E(G)$ be the edges, and $e(G) = |E(G)|$ be the number of edges. Let $\alpha(G)$ be the arboricity of the graph $G$, which is defined as the minimal number of forests that are needed to cover the edges of $G$.

For a graph $G$, an induced subgraph $H \subseteq G$ is the subgraph that contains exactly the edges of $G$ whose endpoints are both in $V(H)$.

We will consider algorithms that have \textit{query access} to the input graph $G$. This means that the algorithm does not have full access to the input graph, and instead must perform \textit{queries} to access information about the underlying structure of the graph. The main object of interest in this paper is the \textit{query complexity} of the algorithm, which is the maximum number of queries that the algorithm must make on any input graph in order to produce its output.

We consider algorithms under the following query access models (see Chapters 8 and 9 of \cite{DBLP:books/cu/Goldreich17} on testing graph properties):
\begin{enumerate}
    \item Adjacency list queries: The algorithm makes queries of the form $(u, i)$ for $u \in V(G)$ and $i \in [n-1]$. If $u$ has at least $i$ neighbors, it obtains as a response the $i$-th neighbor of vertex $u$, according to some arbitrary but consistent ordering of the neighbors of $u$. Else, it obtains as a response $\bot$.
    \item Adjacency matrix queries: The algorithm makes queries of the form $(u, v)$ for $u, v \in V(G)$. It obtains as a response a $1$ if $(u, v) \in E(G)$ and $0$ otherwise. 
\end{enumerate}

When the adjacency lists of vertices are sorted (according to some arbitrary, consistent ordering of the vertices), adjacency list access is so powerful that it can simulate adjacency matrix access. Particularly, a query to the adjacency matrix can be answered with $O(\log n)$ queries to a vertex's sorted adjacency list. Most work on sublinear algorithms for approximating motif counts on graphs (e.g., \cite{assadi2018simple, eden2017approximately, eden2018faster}) assumes access to both the adjacency matrix and adjacency lists (and sometimes even more, such as access to random edge samples). On the other hand, literature on hypothesis testing for planted solution problems in sublinear access models typically assumes adjacency matrix access only (e.g., \cite{racz2020finding, huleihel2024random}).

We will primarily study quality control algorithms with adjacency matrix queries. However, we will also study algorithms with adjacency list queries in certain settings (triangle counting); leveraging previous results that use adjacency list queries allows us to obtain improved and optimal query complexity in this setting.

We now define some relevant terminology related to graphs.

\begin{definition}[Graph parameter]\label{def:graph-parameter}
    Let $\mathcal{G}(n)$ be the family of graphs on $n$ vertices. A graph parameter is a function $F: \mathcal{G}(n) \to \mathbb{R}$ that is invariant under graph isomorphism. That is, for any graph $G \in \mathcal{G}(n)$, and any graph $G'$ resulting from permuting the vertices of $G$, $F(G) = F(G')$. 
\end{definition}

\begin{definition}[Motif]
    A motif of size $k$ on a graph $G$ is a labeled induced subgraph on $k$ vertices of $G$. Typically, motifs refer to small graphs that frequently match size-$k$ induced subgraphs of a random graph model being considered.
\end{definition}

\begin{definition}[Motif count]\label{def:C-H-G}
    For a motif $H$ and a graph $G$, let $C_H(G)$ be the number of labeled, induced copies of $H$ in $G$. Equivalently, $C_H(G)$ equals the number of injective maps $\psi: V(H) \to V(G)$ such that for all $(u, v) \in E(H)$, $(\psi(u), \psi(v)) \in E(G)$, and for all $(w, x) \not \in E(H)$, $(\psi(w), \psi(x)) \not \in E(G)$.
\end{definition}

\begin{definition}[Arboricity]\label{def:arboricity}
    The arboricity $\alpha(G)$ of a graph $G$ is the minimal number of forests that the graph's edges can be partitioned into.
\end{definition}

The arboricity of a graph measures the graph's density, in the sense that, in a graph of arboricity at most $\alpha$, every subgraph has average degree at most $O(\alpha)$.

\subsection{Properties of \erdosrenyi graphs}\label{sec:erdos-renyi-properties}

We first discuss the number of induced copies of motifs in random graphs. \cite{janson1990poisson, janson2004upper} previously proved a result regarding the concentration of the number of non-induced copies of motifs in \erdosrenyi graphs. We provide a simple proof for concentration in the case of induced copies of motifs.

\begin{proposition}\label{lem:Gnp-motif-concentration}
    Fix $\varepsilon>0$. Let $H$ be an induced motif on $k$ vertices, suppose $0 < p \leq 1/2$, and define $m_H := \max_{F \subseteq H} e(F)/v(F)$.
    
    Consider parameters $n, p$ with $p = \omega\left(n^{-1/m_H} \cdot \varepsilon^{-2}\right)$. For $G \sim \gnp$, with probability $1 - o_{1/p}(1)$, the number of labeled induced copies of $H$ in $G$ is in the range $(1 \pm \frac{\varepsilon}{100}) \binom{n}{v(H)} v(H)! \cdot p^{e(H)} (1-p)^{\binom{v(H)}{2} - e(H)}$. 
\end{proposition} 

For an edgeless motif, interpret the size condition as $n \to \infty$. The proof of this result is given in \Cref{sec:appendix}. We note that the choice of $1/100$ in the expression above is chosen to best fit our setting, and can be changed to another constant by modifying the constant hidden in the lower bound on $p$.

From this proposition, we can derive the following consequence for the concentration of clique counts for $G \sim \gnp$, using that $m_H = (k-1)/2$ when $H$ is a $k$-clique.\footnote{This corollary is also implied by \cite{janson1990poisson, janson2004upper} since induced and non-induced clique counts agree.}

\begin{corollary}
  \label{lem:Gnp-clique-concentration}
    For $G \sim \gnp$ and $p = \omega\left(n^{-2/(k-1)} \cdot \varepsilon^{-2} \right)$, the number of labeled $k$-cliques in $G$ is in the range $(1 \pm \frac{\varepsilon}{100}) \binom{n}{k} k! \cdot p^{\binom{k}{2}}$ with probability $1 - o_{1/p}(1)$.
\end{corollary} 

Next, we turn to a discussion of the arboricity (\Cref{def:arboricity}) of \erdosrenyi graphs. \cite{gao2018arboricity} proves the following regarding the arboricity $\alpha(G)$ of $G \sim \gnp$ (which indeed is a precise statement that holds asymptotically almost surely, while we provide a weaker version of their statement below for simplicity).

\begin{proposition}[\cite{gao2018arboricity}]\label{prop:Gnp-arboricity-concentration}
    For $p=\omega(1/n)$, there exists an $N$ such that for all $n\geq N$,
    $$\mathbb{P}_{G\sim\gnp}\left(\alpha(G)\leq2\frac{e(G)}{n-1}+1\right)=1-o_{1/p}(1).$$
\end{proposition}

To obtain a bound on the arboricity in terms of $n, p$, we combine \Cref{prop:Gnp-arboricity-concentration} with \Cref{lem:Gnp-clique-concentration} applied for $k = 2$ (i.e., edge counts).

\begin{corollary}\label{cor:Gnp-arboricity-concentration}
    For $p = \omega\left(\varepsilon^{-2} \cdot 1/n\right)$, there exists an $N$ such that for all $n \geq N$, \\
    $$\mathbb{P}_{G \sim \gnp}\left( \alpha(G) \leq 2 n p \right) = 1 - o_{1/p}(1).$$
\end{corollary}

\subsection{Quasirandom graphs}\label{sec:preliminary-pseudorandom}

\subsubsection{Jumbled graphs}

One of the main characterizations of graphs that we leverage is that of \textit{jumbledness}. Initially defined by Thomason \cite{thomason1987pseudo} and extended to more general notions of graph quasirandomness by subsequent papers -- notably, \cite{chung1989quasi} -- jumbledness characterizes some respect in which the edges in the graph are uniform/regular. A sufficient condition for jumbledness can be checked efficiently in the size of the graph and implies efficient algorithms for subgraph counting.

For $X,Y\subseteq V(G)$, define
$$e(X,Y):=\sum_{x\in X}\sum_{y\in Y}
\mathbf{1}_{\{x,y\}\in E(G)}.$$
Thus $e(X,X)=2e(G[X])$.

\begin{definition}[Jumbled graph (\cite{thomason1987pseudo} as presented in \cite{krivelevich2006pseudo})]\label{def:jumbled}
    A graph $G = (V, E)$ is $(p, \beta)$-jumbled if for every pair of vertex subsets $X, Y \subseteq 
    V$ the following holds:
    $$\left| e(X, Y) - p | X| |Y| \right| \leq \beta \sqrt{|X||Y|}.$$
\end{definition}

Let us understand the values of $(p, \beta)$ for different graphs. 
First, \erdosrenyi graphs $\gnp$ are known to be almost surely $(p,O(\sqrt{np}))$-jumbled when $np=\omega(\log n)$ \cite{krivelevich2006pseudo}. 
In general, the definition above implies $\beta \geq 0$. The value of $\beta$ could be high, 
even as much as $\Theta(n p)$ when
$n p$ is the average degree of the graph  -- this can
happen if there exists a very sparse or dense region in the graph.

Setting $\beta \geq n$ is trivial, because for any graph $e(X, Y) \leq |X| |Y|$, and so $\beta \geq n \geq \sqrt{|X||Y|}$ yields an inequality that is always true/provides no meaningful information. Indeed, the property of being $(p, \Theta(np))$-jumbled does not reveal much about the graph, as linear-sized very sparse/dense subgraphs can still exist. On the other hand, $(p, o(np))$-jumbledness in the dense case (constant $p$) implies strong properties about the underlying graph \cite{chung1989quasi}. In general, graph sequences that satisfy these strong properties are called \textit{quasirandom} (or pseudorandom) graphs. See the survey by Krivelevich and Sudakov on pseudorandom graphs \cite{krivelevich2006pseudo} for further information. 

Thomason \cite{thomason1987pseudo} gives efficiently verifiable sufficient conditions for jumbledness.

\begin{proposition}[\cite{thomason1987pseudo} as in \cite{krivelevich2006pseudo}]\label{prop:thomason-jumbled}
    Let $G = (V, E)$ be a graph with $|V| = n$ such that every vertex has degree at least $np$, for $p \in (0, 1)$. If every pair of vertices in $G$ has at most $np^2 + \ell$ common neighbors, then $G$  is $(p, \sqrt{(p + \ell)n})$-jumbled.
\end{proposition}

The proof of the proposition above is not readily found online. We prove a similar statement, given below, thinking of $\ell$ as $\delta n$ for some $\delta > 0$.

\begin{proposition}\label{prop:jumbled-from-deg-codeg}
    Let $G = (V, E)$ where $|V| = n$. Consider a parameter $\delta > 0$. If every vertex in $G$ has degree in $(n-1)(p \pm \delta)$, and every pair of distinct vertices in $G$ has $(n -2)(p^2 \pm \delta)$ common neighbors, then $G$ is $(p, n \sqrt{3\delta} + \sqrt{n(\delta + p)})$-jumbled.
\end{proposition}

\begin{proof}
    We follow and modify the proof that CODEG implies DISC from Chapter 3 of \cite{zhao2023graph}. Consider any two subsets $X, Y \subset V$. We begin by bounding $\left| e(X, Y) - p | X| |Y| \right|$ as follows. By the Cauchy-Schwarz inequality,
    $$\frac{1}{|X|} \left| e(X, Y) - p | X| |Y| \right|^2 \leq \sum_{x \in X}\left(\text{deg}(x, Y) - p|Y| \right)^2.$$
    We can replace the sum over $X$ by a sum over all vertices in $G$ to upper-bound this by:
    $$\leq \sum_{x \in V(G)}\left(\text{deg}(x, Y) - p|Y| \right)^2 = \sum_{x \in V(G)}\left(\text{deg}(x, Y)^2 - 2 p |Y| \text{deg}(x, Y) + p^2 |Y|^2 \right).$$
    Using $\sum_{x \in V(G)} \text{deg}(x, Y)^2 = \sum_{y, y' \in Y}\text{codeg}(y, y')$ (where $\text{codeg}(y, y) = \text{deg}(y)$), and also that $\sum_{x \in V(G)}\text{deg}(x, Y) = \sum_{y \in Y} \text{deg}(y)$, the above expression equals:
    $$= \sum_{y, y' \in Y}\text{codeg}(y, y') - 2 p |Y| \sum_{y \in Y} \text{deg}(y) + p^2 |Y|^2 n.$$
    By the assumptions about vertex degrees and codegrees (shared neighbors) in the lemma, this is:
    $$\leq |Y|^2 (n-2) p^2 \pm |Y|^2 \delta (n-2) + |Y| (n-1) p \pm |Y| \delta (n-1) - 2 p |Y|^2 (n-1)p \pm 2 p |Y|^2 \delta (n-1) + p^2 |Y|^2 n.$$
    Removing terms that cancel each other out,  we are left with:
    $$\leq |Y|^2 \delta (n-2) + 2 p |Y|^2 \delta (n-1) + |Y| (\delta + p) (n-1) \leq 3 |Y|^2 \delta n + |Y| (\delta +  p) n.$$

    Therefore, we have so far found that:
    $$\frac{1}{|X|} \left| e(X, Y) - p | X| |Y| \right|^2  \leq 3 |Y|^2 \delta n + |Y| (\delta +  p) n.$$
    Equivalently, 
    $$\left| e(X, Y) - p | X| |Y| \right| \leq \sqrt{|X||Y| \cdot \left( 3 |Y| \delta n + (\delta + p) n\right)} 
    $$ $$= \sqrt{3|Y|\delta n  + (\delta + p) n} \cdot \sqrt{|X||Y|} \leq \sqrt{3\delta n^2 + (\delta + p) n} \cdot \sqrt{|X||Y|} \leq \left( n \sqrt{3\delta} + \sqrt{n(\delta + p)} \right) \sqrt{|X||Y|}.$$
    Therefore, $G$ is $(p, n \sqrt{3\delta} + \sqrt{n(\delta + p)})$-jumbled.
\end{proof}

Importantly for our goals of testing/counting subgraphs, for the range of $p = p(n)$ we are interested in, jumbled graphs have the property that the counts of motifs approximately match $\gnp$.

We first define $\sigma(H)$ as follows:
\begin{equation}\label{eq:s-H}
\sigma(H) = \min\left\{ \frac{\text{max-deg}(L(H)) + 4}{2}, \frac{D(L(H)) + 6}{2}\right\}.
\end{equation}
It is easy to see that when $H$ is a $k$-clique, $\sigma(H) = k$. For general motifs, \cite{conlon2014extremal} note that 
$\sigma(H) \leq \Delta(H) + 2$.

We observe the following implication of Theorem 1.12 of Conlon, Fox, and Zhao \cite{conlon2014extremal}.

\begin{proposition}[Implication of Theorem 1.12 of \cite{conlon2014extremal}]\label{prop:cfz-jumbled-subgraph}
Let $H$ be a fixed graph on $k \geq 3$ vertices with $m$ edges. 

For every $\ep > 0$, there exists a constant $c = c(H, \ep) > 0$ such that if $G$ is a $(p, \beta)$-jumbled graph on $n$ vertices with 
$$\beta \leq c p^{\sigma(H)} n,$$
then the number of labeled non-induced copies $\widetilde{C}_H(G)$ of $H$ in $G$ satisfies
$$(1 - \ep) p^m \binom{n}{k}k! \leq \widetilde{C}_H(G) \leq (1 + \ep) p^m \binom{n}{k}k!.$$
\end{proposition}

\begin{proof}
   Let $r = \sigma(H)$. Theorem 1.12 of \cite{conlon2014extremal} implies that, for every $\varepsilon > 0$, there exist constants $c_0, \eta > 0$ such that: If $V_1, V_2, \dots V_k$ are disjoint vertex sets such that every $(V_i, V_j)$ is $(p, c_0 p^r \sqrt{|V_i||V_j|})$-jumbled and every pair $(V_i, V_j)$ with $(i, j) \in E(H)$ satisfies $\text{DISC}(p, p, \eta)$\footnote{For a bipartite graph $G$ on $(X, Y)$, the condition $\text{DISC}(q, p, \eta)$ means that, for every $u: X \to [0, 1]$ and $v: Y \to [0, 1]$, $\left|\mathbb{E}_{x \in X, y \in Y}\left[ \left( G(x, y) - q\right) u(x) v(y)\right]\right| \leq \eta p.$}, then 
   $$N_H(V_1, V_2, \dots, V_k) = (1 \pm \varepsilon) p^m \prod_{i = 1}^k |V_i|,$$
   where $N_H(V_1, V_2, \dots, V_k)$ counts labeled copies mapping vertex $i$ of $H$ into $V_i$.

   Fix an ordered equitable partition $V(G) = V_1 \sqcup V_2 \cdots \sqcup V_k$ and identify vertex $i \in V(H)$ with $V_i$. Note that $|V_i| \geq n/(2k)$. Thus the assumption $\beta \leq c p^r n$ implies
   $$\beta \leq 2kcp^r \sqrt{|V_i||V_j|}$$
   for every $i, j$.
   Thus, taking $c = c(H, \varepsilon)$ to be sufficiently small, every pair $(V_i, V_j)$ satisfies the jumbledness hypothesis of Theorem 1.12 of \cite{conlon2014extremal}. Furthermore, jumbledness implies the required $\text{DISC}(p, p, \eta)$ condition, using $$\frac{\beta}{\sqrt{|V_i||V_j|}} \leq 2 k c p^r \leq \eta p.$$

   Applying Theorem 1.12 of \cite{conlon2014extremal} with the host graph and subgraph both equal to $G$, we obtain
   $N_H(V_1, V_2, \dots, V_k) = (1 \pm \varepsilon) p^m \prod_{i = 1}^k |V_i|.$ 

   Fix equitable part sizes $n_1, n_2, \dots, n_k$ with $\sum n_i = n$.
   Choose the ordered equitable partition uniformly at random among ordered partitions satisfying $|V_i| = n_i$. Each fixed labeled copy of $H$ is thus counted with probability $\prod_i n_i/(n)_k$, and
   $$\mathbb{E}[N_H(V_1, V_2, \dots, V_k)] = \widetilde{C}_H(G) \frac{\prod_i n_i}{(n)_k}.$$

   Taking expectations over the uniformly random ordered equitable partition and using the identity above gives
   $$\widetilde{C}_H(G) = (1 \pm \varepsilon) p^m(n)_k = (1 \pm \varepsilon) p^m \binom{n}{k} k!.$$
\end{proof}

We now turn to how we will use the notion of jumbledness and the guarantees it gives us for approximate motif counting in this paper. First, we will reduce the global quality control problem on a graph to the problem of approximately counting motifs in a small subgraph $G|_S$ on around $s$ vertices. We next efficiently (i.e., in $\text{poly}(s)$ time) test for jumbledness, rejecting graphs that are not sufficiently jumbled. If the graph is jumbled, we then use \Cref{prop:cfz-jumbled-subgraph} to argue that the motif count approximately matches $\mathcal{G}_{s, p}$. As we show in \Cref{sec:random-graphs-jumbled}, \erdosrenyi graphs satisfy a jumbledness property on subgraphs of size around $s$, and so for quality control, it is acceptable to reject graphs with subgraphs that are not jumbled.

\subsubsection{Random graphs}\label{sec:random-graphs-jumbled}

When we focus on the algorithmic complexity of quality control for $k$-cliques matching $\gnp$ in Section \ref{sec:algorithmic-complexity}, we will utilize the notion of graph \textit{jumbledness} (\Cref{def:jumbled}) to obtain an efficient runtime. Our algorithm reduces the global quality control problem to a local problem over a subgraph of size $s$. A sufficient condition for jumbledness on this subgraph is efficiently checkable (\Cref{prop:jumbled-from-deg-codeg}), and jumbledness implies subgraph counts (\Cref{prop:cfz-jumbled-subgraph}). For completeness, we need to ensure that $G \sim \gnp$ satisfies the sufficient condition from \Cref{prop:jumbled-from-deg-codeg} for jumbledness over subgraphs of size $s$. With this application in mind, we begin by proving that the sufficient conditions are satisfied by $G \sim \gnp$.

\begin{proposition}\label{prop:Gnp-jumbled-conditions}
    Let $C = C(k, 1/\ep)\geq10$. When $p \geq \left( \frac{1200 C^4\log n}{n }\right)^{1/(4k - 1)}$, for $\delta = \frac{p^{2k}}{12 C^2}$, with probability at least $1 - 1/n$, $G \sim \gnp$ satisfies the following:
    \begin{enumerate}
        \item Every vertex in $G$ has degree in $(n-1)(p \pm \delta)$, and 
        \item Every pair of distinct vertices in $G$ has $(n-2) (p^2 \pm \delta)$ common neighbors.
    \end{enumerate}
\end{proposition}

\begin{proof}
    We first prove (1). Consider a graph $G \sim \gnp$, and any vertex $u \in G$. Let us analyze the degree distribution of this vertex. Let $D$ be the random variable corresponding to the degree of vertex $u$. First, the expected degree of vertex $u$ is $E\left(D\right) = (n-1) p$, and the degree distribution of $u$ is a binomial distribution with parameters $n-1$ and $p$. Applying Chernoff bounds, $$\mathbb{P}\left(\left|D -  (n-1) p\right| \geq \delta (n-1) \right) \leq 2 \exp\left( -(n-1) p \cdot \frac{\delta^2}{3 p^2}\right) = 2 \exp\left( \frac{-(n-1) p^{(4k-1)}}{432 C^4}\right).$$
    When $p \geq \left( \frac{1200 C^4\log n}{n }\right)^{1/(4k - 1)}$, this probability is at most $n^{-2.5}$. By a union bound over all vertices in the graph, we find that with probability at least $1 - 1/(2n)$, every vertex in $G \sim \gnp$ has degree in $(n-1)(p \pm \delta)$.

    Next we prove (2) by a similar argument. Consider a graph $G \sim \gnp$ and any pair of vertices $u, v \in G$. Let's analyze the distribution of the number of common neighbors (i.e. codegree) of $u$ and $v$. Let $N$ be the random variable corresponding to the number of common neighbors of $u$ and $v$. The expected value of $N$ is $(n-2)p^2$, and its distribution is binomial with parameters $n - 2$ and $p^2$. By Chernoff bounds, we find:
    $$\mathbb{P}\left( \left| N - (n-2)p^2\right| \geq \delta(n-2)\right) \leq 2 \exp\left( - (n-2)p^2 \cdot \frac{\delta^2}{3 p^4} \right) = 2 \exp\left(\frac{-(n-2)p^{4k-2}}{432 C^4} \right).$$
    When $p \geq \left( \frac{1200 C^4\log n}{n }\right)^{1/(4k - 1)}$, the probability is then at most $n^{-3.5}$. By a union bound over pairs of vertices, with probability at least $1 - 1/(2n)$, every pair of distinct vertices in $G$ has $(n-2)(p^2 \pm \delta) $ common neighbors. Therefore, we obtain the proposition.
\end{proof}

We observe that, by combining \Cref{prop:jumbled-from-deg-codeg} with \Cref{prop:Gnp-jumbled-conditions}, we immediately obtain the following.

\begin{corollary}\label{cor:Gnp-is-jumbled}
    Let $C = C(k, 1/\ep)\geq10$. For $p \geq \left( \frac{1200 C^4\log n}{n }\right)^{1/(4k - 1)}$, with probability at least $1 - 1/n$, $G \sim \gnp$ is $(p, \beta)$-jumbled for $\beta \leq p^k n / C$.
\end{corollary}

\begin{proof}
    By \Cref{prop:Gnp-jumbled-conditions}, with probability at least $1 - 1/n$, for $\delta = \frac{p^{2k}}{12 C^2}$, every vertex in $G$ has degree in $(n-1)(p \pm \delta)$, and every pair of distinct vertices in $G$ has $(n -2)(p^2 \pm \delta)$ common neighbors. By \Cref{prop:jumbled-from-deg-codeg}, since $n \sqrt{3\delta} + \sqrt{n(\delta + p)} \leq p^k n / C$, this implies that with probability at least $1 - 1/n$, $G \sim \gnp$ is $(p, \beta)$-jumbled for $\beta \leq p^k n / C$.
\end{proof}

Given that, for any $1 \leq t \leq n$, the marginal of a $t$-vertex subgraph of $G \sim \gnp$ is distributed according to $G_{t, p}$, we can use the previous proposition to prove the following. For a graph $G$ and subset $S \subseteq [n]$, let $G|_S$ be the induced subgraph on the vertices in $S$.

\begin{proposition}\label{cor:jumbled-subgraphs}
    Let $C = C(k, 1/\ep)\geq10$. Let $p \geq \left( \frac{1200 C^4\log n}{n }\right)^{1/(4k - 1)}$, $n\geq s$ with $s\geq2400C^4\ln(2C)/p^{5k}$, and $\delta = \frac{p^{2k}}{12 C^2}$. Consider $G \sim \gnp$, and choose a random subgraph $G|_S$ of $G$ on $s$ vertices. Then, with probability $1 - o_{1/p}(1)$, the subgraph on these $s$ vertices satisfies the following:
    \begin{enumerate}
        \item Every vertex in $G|_S$ has degree in $(s-1)(p \pm \delta)$ in $G|_S$, and 
        \item Every pair of distinct vertices in $G|_S$ has $(s-2) (p^2 \pm \delta)$ common neighbors in $G|_S$.
    \end{enumerate}
\end{proposition}

\begin{proof}
    For a subgraph $S$, let us call $G|_S$ ``good'' if it satisfies conditions (1) and (2) above. For $G \sim \gnp$, let us first analyze the expected number of $S \subseteq [n], |S| = s$ such that $G|_S$ is \textit{not} good.
    $$\mathbb{E}_{G \sim \gnp}\left(\sum_{S \subseteq [n] : |S| = s} 1(G|_S \text{ is not good})\right) = \sum_{S \subseteq [n] : |S| = s} \mathbb{P}(G|_S \text{ is not good}).$$
    Next, observe that the marginal of the graph on $S$ ($G|_S$) is distributed according to $\mathcal{G}_{s, p}$. Also observe that for $s \geq 2400 C^4\ln(2C)/p^{5k}$, $p \geq \left( \frac{1200 C^4\log s}{s }\right)^{1/(4k - 1)}$. Therefore, we can apply \Cref{prop:Gnp-jumbled-conditions} to $G|_S$. We find that, with probability at most $1/s$, $G|_S$ is not good.

    Therefore, the expected number of $S \subseteq [n], |S| = s$ such that $G|_S$ is \textit{not} good is at most $\frac{1}{s} \cdot \binom{n}{s} = o\left( \binom{n}{s} \right)$. Therefore, by Markov's inequality, with probability at least $1 - o_{1/p}(1)$, at least a $(1 - o_{1/p}(1))$-fraction of subgraphs of size $s$ are good.

    To obtain the proposition, take a union bound over the events that $G \sim \gnp$ does not satisfy the property that a $(1 - o_{1/p}(1))$-fraction of size-$s$ subgraphs are good and that a non-good subgraph $G|_S$ is chosen at random.
\end{proof}

\subsection{Concentration inequalities}\label{sec:preliminary-concentrationinequality}

McDiarmid's inequality bounds the difference between a function and its expectation in terms of how much the value of a function can change when the value of one of the inputs is varied.

\begin{theorem}[McDiarmid's Inequality \cite{mcdiarmid1989method}]\label{thm:mcdiarmid}
    Consider a function $f: \mathcal{X}_1 \times \mathcal{X}_2 \times \dots \times \mathcal{X}_t \to \mathbb{R}$ such that, for all $i \in [t]$ the \textit{bounded differences} property is satisfied: For all $(x_1, x_2, \dots, x_t) \in \mathcal{X}_1 \times \mathcal{X}_2 \times \dots \times \mathcal{X}_t$, $i \in [t]$ and $y_i \in \mathcal{X}_i$,
    $$\sup_{y_i \in \mathcal{X}_i} \left| f(x_1, \dots, x_{i - 1}, x_i, x_{i+1}, \dots, x_t) - f(x_1, \dots, x_{i - 1}, y_i, x_{i+1}, \dots, x_t) \right| \leq c_i.$$
    Then, for independent random variables $X_1, X_2, \dots, X_t$ where $X_i \in \mathcal{X}_i$ for $i \in [t]$, and for any $\delta > 0$:$$\mathbb{P}\left( \left|f(X_1, X_2, \dots, X_t) -  \mathbb{E}\left( f(X_1, X_2, \dots, X_t) \right)\right| \geq \delta \right) \leq 2 \exp\left( -\frac{2\delta^2}{\sum_{i = 1}^t c_i^2}\right).$$
\end{theorem}

Combes \cite{combes2024extension} proves the following extension of McDiarmid's Inequality, which requires that a ``bounded differences property'' only holds over a high-probability subset of the domain.

\begin{theorem}[McDiarmid's Inequality when differences are bounded whp \cite{combes2024extension}]\label{thm:mcdiarmid-whpdifferences}
     Consider a function $f: \mathcal{X}_1 \times \mathcal{X}_2 \times \dots \times \mathcal{X}_t \to \mathbb{R}$ and a subset of the domain $\mathcal{Y} \subseteq  \mathcal{X}_1 \times \mathcal{X}_2 \times \dots \times \mathcal{X}_t$ such that the following bounded differences property holds over elements of the subset $\mathcal{Y}$:

     For all $(x_1, x_2, \dots, x_t) \in \mathcal{Y}$ and $(y_1, y_2, \dots, y_t) \in \mathcal{Y}$,
     $$\left| f(x_1, x_2, \dots, x_t) - f(y_1, y_2, \dots, y_t) \right| \leq \sum_{i : x_i \neq y_i} c_i.$$
     Let $X_1, X_2, \dots, X_t$ be independent random variables where $X_i \in \mathcal{X}_i$ for $i \in [t]$. Let $q := 1 - \mathbb{P}\left( (X_1, X_2, \dots, X_t) \in \mathcal{Y}\right)$. Assume $q<1$.

     Then, for any $\delta > 0$:
     $$\mathbb{P}\left( \left|f(X_1, X_2, \dots, X_t) -  \mathbb{E}\left( f(X_1, X_2, \dots, X_t)  | (X_1, X_2, \dots, X_t) \in \mathcal{Y} \right)\right| \geq \delta \right) $$ $$\leq 2q + 2 \exp\left( -\frac{2\max\left( 0, \delta - q \sum_{i = 1}^t c_i\right)^2}{\sum_{i = 1}^t c_i^2}\right).$$
\end{theorem}

We will utilize the following corollary of \Cref{thm:mcdiarmid-whpdifferences}, which compares $f(X_1, X_2, \dots, X_t)$ to its expectation instead of the expectation conditioned on $(X_1, X_2, \dots, X_t) \in \mathcal{Y}$.

\begin{corollary}\label{cor:mcdiarmid-unconditional-expectation}
     Consider a bounded, nonnegative function $f: \mathcal{X}_1 \times \mathcal{X}_2 \times \dots \times \mathcal{X}_t \to \mathbb{R}$ and a subset of the domain $\mathcal{Y} \subseteq  \mathcal{X}_1 \times \mathcal{X}_2 \times \dots \times \mathcal{X}_t$ such that the following bounded differences property holds over elements of the subset $\mathcal{Y}$: For all $(x_1, x_2, \dots, x_t) \in \mathcal{Y}$ and $(y_1, y_2, \dots, y_t) \in \mathcal{Y}$,
     $$\left| f(x_1, x_2, \dots, x_t) - f(y_1, y_2, \dots, y_t) \right| \leq \sum_{i : x_i \neq y_i} c_i.$$
     
    Let $X_1, X_2, \dots, X_t$ be independent random variables where $X_i \in \mathcal{X}_i$ for $i \in [t]$. Let $q := 1 - \mathbb{P}\left( (X_1, X_2, \dots, X_t) \in \mathcal{Y}\right)$. Consider any $\delta > 0$. If $q \leq \min\{\frac{\delta}{2 \max(f)}, \frac{\delta}{4 \mathbb{E}(f(X_1, X_2, \dots, X_t))} , \frac{1}{2}\}$, then 
     $$\mathbb{P}\left( \left|f(X_1, X_2, \dots, X_t) -  \mathbb{E}\left( f(X_1, X_2, \dots, X_t) \right)\right| \geq \delta \right) \leq 2q + 2 \exp\left( -\frac{2\max\left( 0, \delta/2 - q \sum_{i = 1}^t c_i\right)^2}{\sum_{i = 1}^t c_i^2}\right).$$
\end{corollary}

\begin{proof}
    If $\mathbb{E}(f)=0$, the conclusion follows since $f$ is nonnegative; fractions with a zero denominator are interpreted as $+\infty$. If $q=0$, apply \Cref{thm:mcdiarmid-whpdifferences} directly. Otherwise, for simplicity of notation, let us denote the vector of random variables $(X_1, X_2, \dots, X_t)$ as $\Vec{X}$.

    We first prove that when $q \leq \min\{\frac{\delta}{2 \max(f)}, \frac{\delta}{4 \mathbb{E}(f(X_1, X_2, \dots, X_t))} , \frac{1}{2}\}$, we have:
    \begin{equation}\label{eq:conditional-vs-unconditional-expectation}
       \left| \mathbb{E}\left( f(\Vec{X} | \Vec{X} \in \mathcal{Y})\right) - \mathbb{E}\left( f(\Vec{X})\right)\right| \leq \delta/2. 
    \end{equation}
    To prove this, we begin by expanding the expectation of $f$ into conditional expectations:
    $$\mathbb{E}\left( f(\Vec{X})\right) = \mathbb{E}\left( f(\Vec{X}) | \Vec{X} \in \mathcal{Y} \right) \cdot \Pr\left( \Vec{X} \in \mathcal{Y}\right) + \mathbb{E}\left( f(\Vec{X}) | \Vec{X} \not \in \mathcal{Y} \right) \cdot \Pr\left( \Vec{X} \not \in \mathcal{Y}\right).$$
    Therefore, we have:
    $$\mathbb{E}\left( f(\Vec{X}) | \Vec{X} \in \mathcal{Y} \right) = \frac{\mathbb{E}\left( f(\Vec{X})\right) - \mathbb{E}\left( f(\Vec{X}) | \Vec{X} \not \in \mathcal{Y} \right) \cdot \Pr\left( \Vec{X} \not \in \mathcal{Y}\right)}{\Pr\left( \Vec{X} \in \mathcal{Y}\right)}$$
    $$= \frac{\mathbb{E}\left( f(\Vec{X})\right) - \mathbb{E}\left( f(\Vec{X}) | \Vec{X} \not \in \mathcal{Y} \right) \cdot q}{1-q}.$$

    We first upper-bound $\mathbb{E}\left( f(\Vec{X}) | \Vec{X} \in \mathcal{Y} \right)$. From the expression above, this is upper-bounded by $\frac{\mathbb{E}\left( f(\Vec{X})\right)}{1-q}$. Since $q \leq 1/2$, this is at most $(1 + 2q)\mathbb{E}(f(\Vec{X})).$
    Next, since $q \leq \frac{\delta}{4 \mathbb{E}(f(\Vec{X}))}$, we have:
    $$\mathbb{E}\left( f(\Vec{X}) | \Vec{X} \in \mathcal{Y} \right) \leq (1 + 2q)\mathbb{E}(f(\Vec{X})) \leq \mathbb{E}(f(\Vec{X})) + \frac{\delta}{2}.$$

    We now lower-bound $\mathbb{E}\left( f(\Vec{X}) | \Vec{X} \in \mathcal{Y} \right)$:
    $$\mathbb{E}\left( f(\Vec{X}) | \Vec{X} \in \mathcal{Y} \right) \geq \mathbb{E}\left( f(\Vec{X})\right) - \mathbb{E}\left( f(\Vec{X}) | \Vec{X} \not \in \mathcal{Y} \right) \cdot q \geq \mathbb{E}\left( f(\Vec{X})\right) - \max(f) \cdot q.$$
    Since $q \leq \frac{\delta}{2 \max(f)}$, this expression is at least $\mathbb{E}\left( f(\Vec{X})\right) - \frac{\delta}{2}$.

    Therefore, we have shown that \Cref{eq:conditional-vs-unconditional-expectation} holds.

    Next, by the reverse triangle inequality and \Cref{eq:conditional-vs-unconditional-expectation}, $\left| f(\Vec{X}) -\mathbb{E}\left(f(\Vec{X})\right)\right| \geq \delta$ implies $\left| f(\Vec{X}) -\mathbb{E}\left(f(\Vec{X})| \Vec{X} \in \mathcal{Y}\right)\right| \geq \delta/2$. Therefore, $$\mathbb{P}\left( \left| f(\Vec{X}) -\mathbb{E}\left(f(\Vec{X})\right)\right| \geq \delta \right) \leq \mathbb{P}\left(\left| f(\Vec{X}) -\mathbb{E}\left(f(\Vec{X})| \Vec{X} \in \mathcal{Y}\right)\right| \geq \delta/2\right).$$
    Applying \Cref{thm:mcdiarmid-whpdifferences} to the probability on the right-hand side yields the corollary.
\end{proof}

\section{Quality control of k-clique count}\label{sec:query-complexity-cliques}

In this section, we prove \Cref{thm:intro-cliques}, restated below. Let $C_k(G)$ be the number of $k$-cliques in a graph $G$. Let $\rho_k := C_k(G)/\mathbb{E}_{G' \sim \gnp}\left[ C_k(G')\right]$.

\begin{theorem}[\Cref{thm:intro-cliques}]\label{thm:k-cliques}  
    There exist constants $c_1, c_2, c_3$ such that for every constant $k\geq2$, $0<p\leq1/2$, and $p = \omega\left( n^{-2/(k-1)}\right)$, the quality control problem $(\gnp, \rho_k)$ corresponding to estimating the number of $k$-cliques in a graph $G$ supposedly drawn according to $\gnp$ satisfies the following:
    \begin{enumerate}
        \item The problem is solvable in $O(p^{-c_1 k})$ adjacency matrix queries to $G$.
        \item When $p$ is additionally $\Omega(\log(n)/n)^{1/(4k-1)}$ and $n\geq B_{k,\varepsilon,\tau}p^{-10k}$, the problem is solvable in $O(p^{-c_1 k})$ time. 
        \item Furthermore, for $n\geq C_kp^{-2k}$, this problem cannot be solved with $o(p^{-c_3 k})$ adjacency matrix queries.
    \end{enumerate}
\end{theorem}

It remains an interesting open question to identify the precise exponent in the case of $k$-cliques, for growing $k$. For triangles, we furthermore achieve a tight bound on the exponent, obtaining a $\Omega(1/p)$ lower bound and $\widetilde{O}(1/p)$ upper bound on the query complexity and runtime (with adjacency list, adjacency matrix, and degree queries). See \Cref{thm:triangles-usingliterature} in \Cref{sec:triangles}.

\subsection{Query complexity upper bound}\label{sec:qc-upperbound}

In this section, we prove \Cref{thm:k-cliques} Part (1). Before giving some notation and the algorithm, we highlight the main algorithmic principles for the query complexity upper bound. Our algorithm iteratively performs quality control on the count of $\ell$-cliques, for $\ell \in [k]$, compared to \erdosrenyi graphs. Once the algorithm verifies that the count of $\ell$-cliques is approximately the expected count for $\gnp$, counting $(\ell + 1)$-cliques becomes more efficient. This is because, as we will prove, graphs whose $\ell'$-clique counts match $\gnp$ for all $\ell' \leq \ell$ satisfy a strong structural property, which we call \textit{exponentially robust quasirandomness}. This property captures that most subgraphs will have $\ell$-clique counts similar to that expected of subgraphs of $\gnp$, and will allow us to use bounded difference inequalities (\Cref{cor:mcdiarmid-unconditional-expectation}) more effectively.

We now set up some notation. It suffices to consider $0<\varepsilon\leq1/2$. Let $0<\tau\leq1/10$ be the completeness error, let $\varepsilon_1=\varepsilon(2/3)^{k-2}$, and define
$$L=\ln\left(\frac{(16k)^{2k}}{\varepsilon_1\tau p^{\binom{k}{2}}}\right).$$
For $\ell\in[k]$, define $s_1=2k$ and
\begin{equation}\label{def:s-ell-size}
  s_{\ell}=\left\lceil\max\{2k,L/\alpha_{\ell}\}\right\rceil \qquad (2\leq\ell\leq k),
\end{equation}
where $\alpha_{\ell}$ is defined in \Cref{def:alpha-ell}.
Define sample size $s_{*}$ as:
\begin{equation}\label{def:s-size} 
  s_{*} := s_k.
\end{equation}
Additionally, for $s \in \mathbb{N}$ and $\ell \in [k]$, let
\begin{equation}\label{eq:adjustment-global-to-local}
    A_{s, \ell} = \binom{s}{\ell} \cdot \ell! \cdot \frac{1}{n^{\ell}}.
\end{equation}
We briefly explain $A_{s, \ell}$. As we will prove later in Claim \ref{claim:expectation-set-ell-cliques}, if the number of $\ell$-cliques in a graph $G$ is $C_{\ell}(G)$, then the expected number of labeled $\ell$-cliques in a size-$s$ multiset of random vertices from $G$ is $C_{\ell}(G) \cdot A_{s, \ell}$. Therefore, $A_{s, \ell}$ is an ``adjustment'' term that we need to multiply by to translate from the global clique count to the expected clique count in a size-$s$ multiset of randomly chosen vertices.

Next, we define $E_{s, \ell}$ as follows.
\begin{equation}\label{eq:expected-ell-clique-copy-amount}
  E_{s, \ell} = \binom{n}{\ell}\cdot \ell! \cdot p^{\binom{\ell}{2}} \cdot A_{s, \ell}.  
\end{equation}
$E_{s, \ell}$ is the expected number of labeled $\ell$-cliques in a multiset of $s$ vertices (where each element is drawn uniformly from the vertex set) when the graph is drawn from $\gnp$ (see Claim \ref{claim:expectation-set-ell-cliques}). Particularly, since we are considering the number of $\ell$-cliques in a \textit{multiset}, if a vertex appears $c$ times in the multiset, the $\ell$-cliques it is a part of are counted $c$ times.

We consider the following algorithm for quality control of the $k$-clique count compared to $\gnp$.

\begin{algorithm}[H]
\caption{Algorithm \textsc{Clique-Quality} performs quality control on the $k$-clique count of $G$ compared to $\gnp$ in $1/p^{O(k)}$ queries.}
\vspace{1em} ~~ \textsc{Clique-Quality}: On inputs $G$, $n, p, \varepsilon$, $\gnp$, and $k$

\begin{quote}
Sample $s_{*}$ vertices, each uniformly at random (with repetition). Let $S$ be the corresponding multiset of vertices. 

For $\ell$ from $2$ to $k$, compute the number of $\ell$-cliques from $S$ in $G$ (i.e., $C_{\ell}(S)$ from \Cref{def:C-ell-S}) through exhaustive search. 

If, for any $\ell<k$, $C_{\ell}(S)$ is not in $(1\pm\varepsilon_{\ell-1})E_{s_*,\ell}$, or $C_k(S)$ is not in $(1\pm\varepsilon/2)E_{s_*,k}$, reject. Otherwise, accept. 
\end{quote}
\label{alg:clique-quality}
\end{algorithm}

By exhaustively searching to compute the number of $\ell$-cliques in $S$, we mean the following: count the $\ell$-element subsets of
sampled positions whose vertices are distinct and form a clique, and multiply by $\ell!$. In \Cref{sec:algorithmic-complexity}, we improve the procedure for counting $k$-cliques by additionally performing quality control for another property possessed by $G \sim \gnp$ with high probability. However, in this section, we focus on proving that $s_* = 1/p^{O(k)}$ \textit{queries} suffice for quality control of the $k$-clique count compared to $\gnp$.

We now define \textit{exponential quasirandomness}, a property of graphs that implies concentration of the number of $\ell$-cliques around the expected amount in $\mathcal{G}_{s, p}$, in most (multi-sets) $S$ of roughly $s$ vertices. We first set up some notation. 

In what follows, let $C_{\ell}(G)$ be the number of labeled $\ell$-cliques in $G$. For a multiset $S$ of size $s$, define $C_{\ell}(S)$ as follows.

\begin{definition}\label{def:C-ell-S}
    Let $S=(X_1,\dots,X_s)$ be a multiset of size $s$. $C_{\ell}(S)$ is defined to be the number of ordered $\ell$-tuples of distinct sampled positions whose vertices are distinct and form a clique in $G$. Since $S$ is a multiset, a labeled clique on vertices $v_1,\dots,v_{\ell}$ is counted $\prod_{j=1}^{\ell} m_S(v_j)$ times, where $m_S(v)$ is the multiplicity of $v$ in $S$.
\end{definition}

We now formally define the ``exponentially robustly quasirandom'' property of graphs. Exponentially robust quasirandomness is a property possessed with high probability by \erdosrenyi graphs. We will show that performing quality control for clique counts is sufficient for verifying exponentially robust quasirandomness, and that exponentially robust quasirandomness with respect to $\ell$-cliques makes quality control for $(\ell + 1)$-clique counts compared to $\gnp$ more efficient.

\begin{definition}[Exponentially robustly quasirandom (see \Cref{def:exp-qr-motif} for general definition)]\label{def:exp-qr}
    Fix $p = p(n)$. We say that a graph $G=([n],E)$ is $(\varepsilon,\alpha, s_0)$-exponentially robustly quasirandom with respect to the count of $\ell$-cliques if for every $s \geq s_0$ we have that 
    $$\Pr_{X_1,\ldots,X_s \sim_{\mathrm{i.i.d.}} [n]}\left[ C_{\ell}(\{X_1,\ldots,X_s\}) \not\in (1\pm \varepsilon) \binom{n}{\ell}\ell! p^{\binom{\ell}{2}} \cdot \binom{s}{\ell} \frac{\ell!}{n^{\ell}}\right] \leq 4^{\ell + 1} \cdot \exp(-\alpha s).$$
\end{definition}

We call this property a ``quasirandom'' property because it is a deterministic structural property possessed by \erdosrenyi graphs with high probability. However, we remark that the term ``quasirandom'' has often been used to refer to a class of equivalent graph properties all possessed by \erdosrenyi graphs \cite{chung1989quasi}. Exponentially robust quasirandomness is generally incomparable with these quasirandom properties. It is a relaxation of the discrepancy condition of \cite{chung1989quasi} in the case of exponentially robust quasirandomness for $\ell = 2$. For the other quasirandom properties, it is somewhat stronger in the sense of caring about counts of motifs in subgraphs (instead of, say, globally). At the same time, it is also somewhat weaker, since the property need not hold for all subgraphs of the graph.

We call this \textit{exponentially} robustly quasirandom because the probability that the $\ell$-clique count is not in some range shrinks as $\exp(-s)$ when our sample size is $s$. This exponentially small probability is important as we compose quality control of $\ell$-cliques for different $\ell \leq k$.

We now develop lemmas related to the implications of exponentially robust quasirandomness for $\ell$-cliques to that of $(\ell + 1)$-cliques. To do so, below, we first introduce the parameters these lemmas will rely on. In what follows, we will specify an $\varepsilon, \alpha$ pair for each $\ell$, setting $\alpha_\ell = O(p^{2(\ell - 1)})$. 

\begin{definition}\label{def:alpha-ell}
    For $1\leq\ell\leq k-1$, define
    $$\varepsilon_{\ell}=\varepsilon\cdot\left(\frac{2}{3}\right)^{k-\ell-1}.$$
    Define $\alpha_1=\infty$, and for $2\leq\ell\leq k$ define
    $$\alpha_{\ell}=\min\left\{\frac{\alpha_{\ell-1}}{2},\frac{p^{2\ell-2}\cdot\varepsilon_1^2}{16384\cdot\ell^2}\right\}.$$
\end{definition}
Observe that, for $s_B \geq s_A$, $(\varepsilon,\alpha, s_A)$-exponentially robust quasirandomness with respect to the count of $\ell$-cliques implies $(\varepsilon,\alpha, s_B)$-exponentially robust quasirandomness.

The key ingredient in the analysis of the algorithm is the following inductive lemma, which states that a graph $G$ is $(\varepsilon_{\ell - 1}, \alpha_{\ell-1}, s_{\ell})$-exponentially robustly quasirandom with respect to $(\ell-1)$-clique counts, then it can be tested for being $(\varepsilon_{\ell}, \alpha_{\ell}, s_{\ell})$-exponentially robustly quasirandom with respect to $\ell$-clique counts in $\widetilde{O}\left( p^{-4\ell + 4}\right)$ queries to the adjacency matrix.

\begin{lemma}[Inductive Lemma]\label{lemma:inductive-cliques}
    Let $2\leq\ell<k$. Suppose that the graph $G$ is $(\varepsilon_{\ell - 1}, \alpha_{\ell-1}, s_{\ell})$-exponentially robustly quasirandom with respect to the count of $(\ell-1)$-cliques.
    
    Define the tester $T_{\ell}$ that samples $s_{\ell}$ (as specified in \Cref{def:s-ell-size}) random vertices of $G$ (with replacement) and accepts if and only if the $\ell$-clique count in the induced subgraph on the sampled vertices is within $(1 \pm \varepsilon_{\ell-1}) E_{s_{\ell}, \ell}$ (as specified in \Cref{eq:expected-ell-clique-copy-amount}). This algorithm uses $\widetilde{O}\left( s_{\ell}^2\right) = \widetilde{O}\left( p^{-4\ell + 4} \varepsilon^{-4}\right)$ adjacency matrix queries.
    
    With probability at least $1 - 4^{\ell + 1}\exp(-\alpha_{\ell} s_{\ell})$, the algorithm accepts graphs that are $(\varepsilon_{\ell-1}, \alpha_{\ell}, s_{\ell})$-exponentially robustly quasirandom with respect to the $\ell$-clique count and rejects graphs that are not 
    $(\varepsilon_{\ell}, \alpha_{\ell}, s_{\ell})$-exponentially robustly quasirandom with respect to the $\ell$-clique count.
\end{lemma}

The induction lemma is (``easily'') proven via the following concentration lemma (given qualitatively in \Cref{lem:concentration-intro}):

\begin{lemma}[Concentration Lemma (see \Cref{lem:concentration-general-P-and-D} for general statement)]\label{lem:concentration}
 Suppose $G$ is $(\gamma, \beta, s)$-exponentially robustly quasirandom with respect to the count of $(\ell-1)$-cliques. Suppose $0<\gamma\leq1$, $n,s\geq2\ell$, and $s$ satisfies $\exp(-\beta s)\leq\gamma p^{\binom{\ell}{2}}/(256\ell^2 8^{\ell})$.

 Then:
 $$\Pr_S\left( \left| C_{\ell}(S) - C_{\ell}(G) \cdot A_{s, \ell}\right| \geq \frac{\gamma}{4} E_{s, \ell} \right) \leq 2 \cdot 4^{\ell}\exp(-\beta s) + 2\exp\left( -\gamma^2 s \cdot \frac{p^{2\ell - 2}}{8192\cdot\ell^2}\right).$$
\end{lemma}

For our application to checking $k$-clique counts, we will repeatedly apply this concentration lemma with the parameters $\varepsilon_{\ell}$ and $\alpha_{\ell}$ from \Cref{def:alpha-ell}. Therefore, for ease of use, we provide the concentration lemma with the specific parameters as a corollary.

\begin{lemma}[Concentration Lemma, instantiated] 
\label{cor:concentration-lemma-with-parameters}
    For $2 \leq \ell \leq k$, consider the definitions of $\varepsilon_{\ell-1}$ and $\alpha_{\ell}$ from \Cref{def:alpha-ell} and the definition of $s_{\ell}$ from \Cref{def:s-ell-size}.

    Suppose $G$ is $(\varepsilon_{\ell - 1}, \alpha_{\ell - 1}, s_{\ell})$-exponentially robustly quasirandom with respect to $(\ell - 1)$-clique counts. Then:
    $$\Pr_S\left( \left| C_{\ell}(S) - C_{\ell}(G) \cdot A_{s_{\ell}, \ell}\right| \geq \frac{\varepsilon_{\ell - 1}}{4} E_{s_{\ell}, \ell} \right) \leq 4^{\ell + 1} \exp(-\alpha_{\ell} s_{\ell}).$$
\end{lemma}

The same bound holds for every sample size $u\geq s_{\ell}$, with $u$ in place of $s_{\ell}$ in the sampled count and the probability bound. It also holds with $\varepsilon_1$ in place of $\varepsilon_{\ell-1}$ under the corresponding quasirandomness assumption. The inductive tester may likewise use any $u\geq s_{\ell}$, with $O(u^2)$ queries, keeping $s_{\ell}$ as its quasirandomness threshold.

To see how to prove \Cref{lemma:inductive-cliques} from \Cref{lem:concentration}, observe that the concentration lemma tells us that if a graph is exponentially robustly quasirandom with respect to $(\ell-1)$-cliques, then with high probability a sampled multiset $S$ of size at least $s_{\ell}$ will have $C_{\ell}(S)$ that is proportional to $C_{\ell}(G)$. We can therefore use the number of $\ell$-cliques in $S$ to determine if $C_{\ell}(G)$ corresponds to the $\gnp$ expected amount or not. If it does, then the concentration lemma implies exponential quasirandomness for $\ell$-cliques. If it does not, we can efficiently reject the input graph.

\subsubsection{Proof of Concentration Lemma}\label{sec:proof-concentration}

In this section, we prove \Cref{lem:concentration}.

We first prove the following two claims. Let $C_{\ell}(G)$ be the number of $\ell$-cliques in the graph $G$. For a multiset $S$, define $C_{\ell}(S)$ to be the number of labeled $\ell$-cliques in $S$, with repetitions for repeated vertices, as defined in \Cref{def:C-ell-S}.

\begin{claim}\label{claim:expectation-set-ell-cliques}
    Let $S$ be a multiset generated by sampling $s$ vertices from the graph, with replacement. Then  $\mathbb{E}_S\left[ C_{\ell}(S)\right] = C_{\ell}(G) \cdot A_{s, \ell}.$
\end{claim}

\begin{proof}
    Let us first write the number of $\ell$-cliques in $S$ (denoted as $C_{\ell}(S)$) as a sum of indicator random variables. For ordered set $T$ and ordered $\ell$-clique $C$, let $T \equiv C$ mean that the ordered set of vertices of $T$ is exactly that of $C$. Let $\mathcal{C}_{\ell}(G)$ be the set of labeled $\ell$-cliques in $G$.

    $$C_{\ell}(S) = \sum_{C \in \mathcal{C}_{\ell}(G)}\sum_{T\in[s]^{\ell}_{\ne}} 1(T \equiv C).$$
    
    Here $[s]^{\ell}_{\ne}$ denotes ordered tuples of distinct sampled positions. Fix such a tuple, and some labeled $\ell$-clique $C \in \mathcal{C}_{\ell}(G)$. The probability that $T \equiv C$ is $1/n^{\ell}$, since $T$ is created by sampling $\ell$ vertices, each time uniformly at random (with repetitions).
    
    By linearity of expectation, 
    $\mathbb{E}\left( C_{\ell}(S) \right)= C_{\ell}(G) \cdot \binom{s}{\ell} \ell! / n^{\ell} = C_{\ell}(G) \cdot A_{s, \ell}.$
\end{proof}

In the following claim, let $ E_{s, \ell - 1}$ be the expected number of $(\ell-1)$-cliques in a size-$s$ multiset of vertices in a graph drawn from $\gnp$, as defined in \Cref{eq:expected-ell-clique-copy-amount}.

\begin{claim}\label{claim:bounded-differences-whp-ell-cliques}
    Consider any two sampled sequences $S=(x_1,\dots,x_s)$ and $S'=(y_1,\dots,y_s)$ such that $C_{\ell-1}(S),C_{\ell-1}(S')\in(1\pm\varepsilon)E_{s,\ell-1}$. Let $f_{\ell}(S)$ be the number of labeled $\ell$-cliques in $S$ (and similarly for $f_{\ell}(S')$). Then,
    $$|f_{\ell}(S)-f_{\ell}(S')|\leq |\{i:x_i\ne y_i\}|\cdot\ell\cdot(1+\varepsilon)E_{s,\ell-1}.$$
\end{claim}

\begin{proof}
    Let $D=\{i:x_i\ne y_i\}$. The only labeled cliques that contribute to the difference are those using a position in $D$. For any sampled position, the number of labeled $\ell$-cliques using it is at most $\ell C_{\ell-1}(S)$, since removing that position leaves a labeled $(\ell-1)$-clique and there are $\ell$ possible roles for the position. The same argument applies to $S'$.

    Therefore, the number of copies lost is at most $|D|\ell C_{\ell-1}(S)$, and the number gained is at most $|D|\ell C_{\ell-1}(S')$. Their absolute difference is at most the larger of these two numbers, giving
    $$|f_{\ell}(S)-f_{\ell}(S')|\leq |D|\ell\max\{C_{\ell-1}(S),C_{\ell-1}(S')\}\leq |D|\ell(1+\varepsilon)E_{s,\ell-1}.\eqno\qedhere$$
\end{proof}

Before proving \Cref{lem:concentration}, we make a brief observation about the values of $\alpha_{\ell}$ for $1 \leq \ell \leq k$.

\begin{observation}\label{obs:alpha-decreases}
    For $p \leq 1/2$, for $1\leq\ell<k$, $\alpha_{\ell+1} \leq \alpha_{\ell}$.
\end{observation}

We are now ready to prove \Cref{lem:concentration}. Let's first give an informal outline of the proof. We assume that $G$ is $(\gamma, \beta, s)$-exponentially robustly quasirandom with respect to $(\ell-1)$-clique counts. This provides a sufficient setting for applying \Cref{thm:mcdiarmid-whpdifferences} (a concentration inequality for functions with high-probability bounded differences). Since exponentially robust quasirandomness implies that the probability of sampling a set of size $s$ with $\not \approx E_{s, \ell - 1}$ cliques is exponentially small, the important component that remains in applying \Cref{thm:mcdiarmid-whpdifferences}  is understanding the \textit{bounded differences} that hold for sets of size $s$ with $\approx E_{s, \ell - 1}$ cliques. Claim \ref{claim:bounded-differences-whp-ell-cliques} relates the bounded differences of the function counting the number of $\ell$-cliques in a multiset $S$ to the number of $(\ell - 1)$-cliques in the multiset. When we plug in the bounded differences guarantee from Claim \ref{claim:bounded-differences-whp-ell-cliques} on sets with around $E_{s, \ell - 1}$ $(\ell - 1)$-cliques, we obtain a term depending on $E_{s, \ell}/E_{s, \ell - 1} = (n-\ell+1)(s-\ell+1)p^{\ell-1}/n$.

There are three main takeaways that this proof, therefore, relies on. First, exponentially robust quasirandomness guarantees that the probability of sampling a ``bad''/irregular multiset is exponentially small. Second, the bounded differences on the ``good'' sets relate to $E_{s, \ell - 1}$, so are \textit{much} smaller than the bounded differences would be in an application of standard McDiarmid's inequality. Third, since using concentration of functions with \textit{high-probability} bounded differences allows us to significantly reduce the constant of the bounded differences, we obtain the desired concentration given in \Cref{lem:concentration}.

\begin{proof}[Proof of \Cref{lem:concentration}]
    Suppose that $G$ is $(\gamma, \beta, s)$-exponentially robustly quasirandom with respect to $(\ell-1)$-clique counts. Let $X_1, \dots, X_s$ be the random variables corresponding to which vertices are selected to be in $S$ (i.e., $X_i$ is the $i$-th vertex added to $S$, with repetitions). Define $f_{\ell-1} : [n]^s \to \mathbb{R}$ to be the number of $(\ell-1)$-cliques in the multiset. Let $E_{s, \ell - 1}$ and $E_{s, \ell}$ be the expected number of $(\ell - 1)$-cliques and $\ell$-cliques, respectively, in a size-$s$ multiset for a graph drawn from $\gnp$, as defined in \Cref{eq:expected-ell-clique-copy-amount}. By \Cref{def:exp-qr} of exponential quasirandomness, for every $s' \geq s$ we have that $$\Pr\left[ f_{\ell - 1} (X_1, X_2, \dots, X_{s'}) \not\in (1\pm \gamma) E_{s', \ell - 1}\right] \leq 4^{\ell} \exp(-\beta s').$$ 

    We will apply McDiarmid's Inequality when differences are bounded with high probability (\Cref{cor:mcdiarmid-unconditional-expectation}) to the function $f_{\ell} : [n]^s \to \mathbb{R}$, which is the number of $\ell$-cliques in the multiset.
    
    Let us set up the relevant definitions for the inequality. First, let the domain be $[n]^s$, and define $\mathcal{Y} \subseteq [n]^s$ to be the multisets of $s$ elements from the domain $[n]$ such that the number of $(\ell-1)$-cliques in $\mathcal{Y}$  is in  $(1\pm \gamma) E_{s, \ell - 1}$. Then, observe that, from the assumption of exponentially robust quasirandomness, $q := 1 - \mathbb{P}\left( (X_1, X_2, \dots, X_s) \in \mathcal{Y}\right) \leq 4^{\ell} \exp(-\beta s)$.
    
    We argue about bounded differences of $f_{\ell}$ on $\mathcal{Y}$. A key insight is that the bounded differences for the $\ell$-clique count $f_{\ell}$ of a multiset can be upper-bounded by the number of $(\ell - 1)$-cliques in this multiset. By definition of $\mathcal{Y}$, the number of $(\ell - 1)$-cliques in such a multiset is in $(1 \pm \gamma) E_{s, \ell - 1}$. Formally, for each $i \in [s]$, define $c_i = \ell \cdot (1 + \gamma) E_{s, \ell - 1}$, which is a good upper-bound on the number of ordered $\ell$-cliques that $i$ can contribute to in a multiset in $\mathcal{Y}$. By Claim \ref{claim:bounded-differences-whp-ell-cliques}, for all $S = (x_1, x_2, \dots, x_s) \in \mathcal{Y}$ and $S' = (y_1, y_2, \dots, y_s) \in \mathcal{Y}$, $$\left|f_{\ell}(S) - f_{\ell}(S')\right| \leq  |\{i:x_i\ne y_i\}| \cdot \ell \cdot (1 + \gamma) E_{s, \ell - 1} = \sum_{i : x_i \neq y_i} c_i,$$
    so the bounded difference property holds with this setting of $c_i$'s.

    We want to apply \Cref{cor:mcdiarmid-unconditional-expectation} with error bound $\delta = \gamma E_{s, \ell}/4$. To do so, we need to verify that a condition needed for \Cref{cor:mcdiarmid-unconditional-expectation} holds, particularly that:
    $$q \leq \min\left\{\frac{\delta}{2 \max(f)}, \frac{\delta}{4 \mathbb{E}(f(X_1, X_2, \dots, X_t))} , \frac{1}{2}\right\}.$$ Using $\delta = \gamma E_{s, \ell}/4$, and since $\max(f) \leq \binom{s}{\ell} \ell!$ and $\mathbb{E}(f(X_1, X_2, \dots, X_t)) \leq \binom{s}{\ell} \ell! $, it suffices to show that:
    \begin{equation}\label{eq:q-bound}
       q \leq \min\left\{ \frac{\gamma E_{s, \ell}}{8 \binom{s}{\ell} \ell!}, \frac{\gamma E_{s, \ell}}{16 \binom{s}{\ell} \ell!}, \frac{1}{2}\right\} = \min\left\{ \frac{\gamma}{16} \cdot \binom{n}{\ell}\ell! p^{\binom{\ell}{2}} \cdot \frac{1}{n^{\ell}}, \frac{1}{2} \right\}. 
    \end{equation}
    For our setting of $q$, the size assumption gives $q\leq\gamma p^{\binom{\ell}{2}}/(256\ell^2 2^{\ell})$. Since $n\geq2\ell$, we have $\binom{n}{\ell}\ell!/n^{\ell}\geq2^{-\ell}$, so \Cref{eq:q-bound} holds.

    We are ready to apply \Cref{cor:mcdiarmid-unconditional-expectation}, which implies that:
    $$\mathbb{P}\left( \left|f_{\ell}(X_1, X_2, \dots, X_s) -  \mathbb{E}\left( f_{\ell}(X_1, X_2, \dots, X_s) \right)\right| \geq \gamma E_{s, \ell} /4 \right)$$
    \begin{equation}\label{eq:mcdiarmid-application}
       \leq 2q + 2 \exp\left( -\frac{2\max\left( 0, \gamma E_{s, \ell}/8 - q \sum_{i = 1}^s c_i\right)^2}{\sum_{i = 1}^s c_i^2}\right). 
    \end{equation}
    Our next key insight will be to reduce the right-hand side of \Cref{eq:mcdiarmid-application} to a term depending on $s$ and a ratio between the expected number of $\ell$-cliques in a multiset of size $s$ taken from $\gnp$ ($E_{s, \ell}$) and the expected number of $(\ell - 1)$-cliques in such a multiset ($E_{s, \ell-1}$). Since $E_{s, \ell}/E_{s, \ell-1} = (n-\ell+1)(s-\ell+1)p^{\ell-1}/n$, this will bring us closer to achieving the desired bound in \Cref{lem:concentration}, which has a $p^{O(\ell)}$ dependence on $p, \ell$ (versus the less optimal dependence of $p^{O(\ell^2)}$ that more naive methods would give).
    
    To formalize this insight, let us analyze the numerator of the second term. Since
    $$\frac{E_{s,\ell}}{E_{s,\ell-1}}=\frac{(n-\ell+1)(s-\ell+1)}{n}p^{\ell-1}\geq\frac{s}{4}p^{\ell-1},$$
    the size assumption and $1+\gamma\leq2$ give
    $$q\sum_{i=1}^s c_i\leq2q s\ell E_{s,\ell-1}\leq\frac{\gamma}{16}E_{s,\ell}.$$
    Therefore, $2\max(0,\gamma E_{s,\ell}/8-q\sum_i c_i)^2\geq\gamma^2 E_{s,\ell}^2/128$.

    Therefore, (also rounding $(1 + \gamma) E_{s, \ell-1}$ up to $2 E_{s, \ell-1}$ in the denominator of the expression) \Cref{eq:mcdiarmid-application} is upper-bounded by:
    $$
        \leq 2q + 2 \exp\left(- \frac{ \gamma^2 E_{s, \ell}^2}{128 \cdot s \cdot (2\ell)^2 E_{s, \ell - 1}^2}\right).
    $$
    By definition of $E_{s, \ell - 1}$ and $E_{s, \ell}$ (\Cref{eq:expected-ell-clique-copy-amount}), the expression above is at most:
    $$
    \leq 2q + 2  \exp\left(- \frac{ \gamma^2}{512 \cdot \ell^2 \cdot s} \cdot \left( (n - \ell + 1) \cdot p^{\ell - 1} \cdot \frac{s - \ell + 1}{n}\right)^2\right).
    $$
    For large enough $n$, and by definition of $s$ (\Cref{def:s-size}), we can use that $(n - \ell + 1)/n \geq 1/2$ and $(s - \ell + 1) \geq s/2$ to obtain that the expression is bounded above by:
    $$\leq 2q + 2 \exp\left(- \frac{\gamma^2}{512 \cdot \ell^2 \cdot s} \cdot \frac{p^{2 \ell - 2} s^2}{16} \right).$$
    By definition of $q$, this is:
    $$\leq 2 \cdot 4^{\ell} \exp(-\beta s) + 2\exp\left( -\gamma^2 s \cdot \frac{p^{2\ell - 2}}{8192 \cdot \ell^2}\right).\eqno \qedhere$$
    \end{proof}

    \begin{proof}[Proof of \Cref{cor:concentration-lemma-with-parameters} given \Cref{lem:concentration}]
        Consider the definitions of $\varepsilon_{\ell},\alpha_{\ell}$ from \Cref{def:alpha-ell}, and the definition of $s_{\ell}$ from \Cref{def:s-ell-size}. We apply the concentration lemma with $\gamma=\varepsilon_{\ell-1}$ and $\beta=\alpha_{\ell-1}$.

        The definition of $L$ gives the size assumption. Also $\alpha_{\ell-1}\geq2\alpha_{\ell}$ and $\varepsilon_{\ell-1}^2p^{2\ell-2}/(8192\ell^2)\geq2\alpha_{\ell}$. Therefore,
        $$2\cdot4^{\ell}e^{-\alpha_{\ell-1}s_{\ell}}+2e^{-\varepsilon_{\ell-1}^2p^{2\ell-2}s_{\ell}/(8192\ell^2)}\leq(2\cdot4^{\ell}+2)e^{-2\alpha_{\ell}s_{\ell}}\leq4^{\ell+1}e^{-\alpha_{\ell}s_{\ell}}.$$
        The same proof applies at every sample size $u\geq s_{\ell}$, and also with $\gamma=\varepsilon_1$. This proves the corollary.
    \end{proof}

\subsubsection{Proof of Inductive Lemma given Concentration Lemma}\label{section:inductive-given-concentration}

In this section, we prove the inductive lemma (\Cref{lemma:inductive-cliques}) from the concentration lemma (\Cref{lem:concentration}). Currently, we have proven in \Cref{lem:concentration} that, if we can assume exponentially robust quasirandomness with respect to $(\ell - 1)$-cliques, on a random multiset $S$, the number of $\ell$-cliques will be close to a scaling of the global count of $\ell$-cliques in the graph. We now want to transfer this to an inductive lemma that states: if we have exponentially robust quasirandomness with respect to $(\ell - 1)$-cliques, the tester that accepts if and only if the $\ell$-clique count in $S$ is approximately $E_{s, \ell}$ will accept graphs that are exponentially robustly quasirandom with respect to the count of $\ell$-cliques and reject graphs that are not (for some parameters).

Towards proving \Cref{lemma:inductive-cliques}, we first present some corollaries of \Cref{lem:concentration}. 

The following corollary proves that if $G$ is $(\varepsilon_{\ell - 1}, \alpha_{\ell-1}, s_{\ell})$-exponentially robustly quasirandom with respect to $(\ell - 1)$-clique counts, a size-$s_{\ell}$ multiset chosen uniformly at random will, with very high probability, reveal what the global number of $\ell$-cliques in $G$ is and whether $G$ is $(\varepsilon_{\ell}, \alpha_{\ell}, s_{\ell})$-exponentially robustly quasirandom with respect to $\ell$-clique counts.

\begin{corollary}[Corollary of \Cref{cor:concentration-lemma-with-parameters}]\label{cor:cor-of-concentration-clique}
    Consider $2 \leq \ell \leq k-1$. If $G$ is $(\varepsilon_{\ell-1}, \alpha_{\ell-1}, s_{\ell})$-exponentially robustly quasirandom with respect to $(\ell-1)$-clique counts, then the following hold:
    \begin{enumerate}
        \item Suppose that $G$ is additionally $(\varepsilon_{\ell-1}, \alpha_{\ell}, s_{\ell})$-exponentially robustly quasirandom with respect to $\ell$-clique counts. Sample a multiset $S$ by selecting $s_{\ell}$ vertices uniformly at random (with repetition). Then with probability at least $1 - 4^{\ell + 1}\exp(-\alpha_{\ell} s_{\ell})$ (over the choice of $S$), $C_{\ell}(S) \in (1 \pm \varepsilon_{\ell - 1}) E_{s_{\ell}, \ell}$.

        This holds for any choice of parameters $\varepsilon_{\ell-1}, \alpha_{\ell}, s_{\ell}$ and is not specific to their definitions in \Cref{def:alpha-ell} and \Cref{def:s-ell-size}.
        \item Sample a multiset $S$ by selecting $s_{\ell}$ vertices uniformly at random (with repetition). If $C_{\ell}(S) \in (1 \pm \varepsilon_{\ell - 1})E_{s_{\ell}, \ell}$, then, except on an event of probability at most $4^{\ell+1}\exp(-\alpha_{\ell}s_{\ell})$ (over the choice of $S$), $C_{\ell}(G) \in (1 \pm \frac{5}{4} \varepsilon_{\ell - 1})\binom{n}{\ell} \ell! \cdot p^{\binom{\ell}{2}}$ and, furthermore, $G$ is $(\varepsilon_{\ell}, \alpha_{\ell}, s_{\ell})$-exponentially robustly quasirandom with respect to $\ell$-clique counts.
    \end{enumerate}
\end{corollary}

\begin{proof}[Proof of \Cref{cor:cor-of-concentration-clique}]
    Suppose that $G$ is $(\varepsilon_{\ell - 1}, \alpha_{\ell-1}, s_{\ell})$-exponentially robustly quasirandom with respect to $(\ell-1)$-clique counts. Then, by \Cref{cor:concentration-lemma-with-parameters}, $$\Pr_S\left( \left| C_{\ell}(S) - C_{\ell}(G) \cdot A_{s_{\ell}, \ell}\right| \geq \frac{\varepsilon_{\ell - 1}}{4}E_{s_{\ell}, \ell} \right) \leq 4^{\ell + 1} \exp(-\alpha_{\ell} s_{\ell}).$$

     First, let's address case (1) of the corollary. Suppose $G$ is $(\varepsilon_{\ell-1}, \alpha_{\ell}, s_{\ell})$-exponentially robustly quasirandom with respect to $\ell$-clique counts. By definition of exponential quasirandomness (\Cref{def:exp-qr}), all but a $4^{\ell + 1}\exp(-\alpha_{\ell} s_{\ell})$ fraction of size-$s$ multisets $S$ satisfy
    $$C_{\ell}(S) \in \left(1 \pm \varepsilon_{\ell-1} \right)E_{s_{\ell}, \ell}.$$

    Let's now move to case (2) of the corollary. \Cref{cor:concentration-lemma-with-parameters} implies that when $G$ is $(\varepsilon_{\ell - 1}, \alpha_{\ell-1}, s_{\ell})$-exponentially robustly quasirandom with respect to $(\ell-1)$-clique counts, with probability at least $1 - 4^{\ell + 1}\exp\left(-\alpha_{\ell} s_{\ell} \right)$, $$C_{\ell} (G) \cdot A_{s_{\ell}, \ell} \in C_{\ell}(S) \pm \frac{\varepsilon_{\ell-1}}{4}E_{s_{\ell}, \ell}.$$
    Equivalently, with probability at least $1 - 4^{\ell + 1}\exp\left(-\alpha_{\ell} s_{\ell} \right)$,
    \begin{equation}\label{eq:ell-clique-G-given-ell-clique-sample}
        C_{\ell} (G) \in C_{\ell}(S) \cdot  \frac{1}{A_{s_{\ell}, \ell}} \pm \frac{\varepsilon_{\ell - 1}}{4}\binom{n}{\ell} \ell! \cdot p^{\binom{\ell}{2}}.
    \end{equation}
    Additionally, by \Cref{cor:concentration-lemma-with-parameters}, under the high-probability event, 
    \begin{equation}\label{eq:ell-clique-sample-given-ell-clique-G}
        C_{\ell}(S) \in C_{\ell}(G) \cdot A_{s_{\ell}, \ell} \pm \frac{\varepsilon_{\ell - 1}}{4}E_{s_{\ell}, \ell}.
    \end{equation}
    
    Suppose that $C_{\ell}(S) \in (1 \pm \varepsilon_{\ell - 1})\binom{n}{\ell} \ell! \cdot p^{\binom{\ell}{2}} \cdot \binom{s_{\ell}}{\ell} \frac{\ell!}{n^{\ell}}$. Therefore, $C_{\ell}(S) \cdot \frac{n^{\ell}}{\binom{s_{\ell}}{\ell}\ell!} \in (1 \pm \varepsilon_{\ell - 1})\binom{n}{\ell} \ell! \cdot p^{\binom{\ell}{2}}$. By \Cref{eq:ell-clique-G-given-ell-clique-sample}, with probability at least $1 - 4^{\ell + 1}\exp\left(-\alpha_{\ell} s_{\ell} \right)$, $C_{\ell}(G) \in (1 \pm \frac{5}{4}\varepsilon_{\ell - 1})\binom{n}{\ell} \ell! \cdot p^{\binom{\ell}{2}}$. We now argue that the graph is exponentially robustly quasirandom. Fix any $u\geq s_{\ell}$ and apply the concentration bound at sample size $u$. Together with the global count bound, this shows that all but a $4^{\ell + 1}\exp(-\alpha_{\ell} s_{\ell})$ fraction of size-$s_{\ell}$ multisets $S'$ satisfy
    $$C_{\ell}(S') \in \left(1 \pm \frac{3}{2}\varepsilon_{\ell - 1}\right)\binom{n}{\ell} \ell! \cdot p^{\binom{\ell}{2}} \cdot \binom{s_{\ell}}{\ell} \frac{\ell!}{n^{\ell}}.$$
    The same displayed bound holds with $u$ in place of $s_{\ell}$, and failure probability $4^{\ell+1}e^{-\alpha_{\ell}u}$, for every $u\geq s_{\ell}$. Since $\varepsilon_{\ell} = \frac{3}{2}\varepsilon_{\ell - 1}$, therefore $G$ is $(\varepsilon_{\ell}, \alpha_{\ell}, s_{\ell})$-exponentially robustly quasirandom with respect to $\ell$-clique counts.
\end{proof}

\begin{proof}[Proof of \Cref{lemma:inductive-cliques}]

Given \Cref{cor:cor-of-concentration-clique}, \Cref{lemma:inductive-cliques} follows from a simple argument, given below.

 Suppose that $G$ is $(\varepsilon_{\ell - 1}, \alpha_{\ell-1}, s_{\ell})$-exponentially robustly quasirandom with respect to $(\ell-1)$-clique counts. Define the tester $T_{\ell}$ that samples $s_{\ell}$ (as specified in \Cref{def:s-ell-size}) random vertices of $G$ (with replacement) and accepts if and only if the labeled $\ell$-clique count in the induced subgraph on the sampled vertices is within $(1 \pm \varepsilon_{\ell - 1}) E_{s_{\ell}, \ell}$.

 By the corollary (\Cref{cor:cor-of-concentration-clique}) of the concentration lemma (\Cref{lem:concentration}), the tester $T_{\ell}$ satisfies the following properties. 

\textbf{Completeness:} Suppose that $G$ is $(\varepsilon_{\ell-1}, \alpha_{\ell}, s_{\ell})$-exponentially robustly quasirandom with respect to $\ell$-cliques. Then, with probability at least $1 - 4^{\ell + 1}\exp(-\alpha_{\ell} s_{\ell})$, $C_{\ell}(S) \in (1 \pm \varepsilon_{\ell - 1}) E_{s_{\ell}, \ell}$.
 
\textbf{Soundness:} If the clique count on the sampled multiset is in $(1 \pm \varepsilon_{\ell - 1}) \binom{n}{\ell} \ell! \cdot p^{\binom{\ell}{2}} \cdot \binom{s_{\ell}}{\ell} \frac{\ell!}{n^{\ell}}$, then, except on an event of probability at most $4^{\ell+1}e^{-\alpha_{\ell}s_{\ell}}$, $G$ is $(\varepsilon_{\ell}, \alpha_{\ell}, s_{\ell})$-exponentially robustly quasirandom with high probability, so accepted graphs satisfy the quasirandomness condition.
\end{proof}

\subsubsection{Proof of Theorem \ref{thm:k-cliques}}

Finally, we use the reasoning of \Cref{lemma:inductive-cliques} to argue about the correctness of the algorithm and the proof of \Cref{thm:k-cliques}. Let $s = s_k$ as in \Cref{def:s-size}, and observe that $s \geq s_{\ell}$ for all $\ell \in [k]$. 

\begin{observation}\label{obs:exp-qr-for-vertices}
    Any input graph $G$ is $(0, \infty, s)$-exponentially robustly quasirandom with respect to the count of vertices (1-cliques), for any $s \in \mathbb{N}$.
\end{observation}

\begin{proof}
    Since we are counting vertices in the multiset with repetitions, the count of vertices in the multiset sampled is always $s$.
\end{proof}

We prove another corollary of \Cref{lem:concentration} which, as opposed to \Cref{cor:cor-of-concentration-clique}, separates the case of $C_k(G) \not \in (1 \pm \varepsilon) \binom{n}{k}k! p^{\binom{k}{2}}$ from the typical $C_k(G)$ range of $G \sim \gnp$ (\Cref{lem:Gnp-clique-concentration}). We remark that \Cref{cor:cor-of-concentration-clique} was suitable for testing quasirandomness for $\ell \leq k - 1$, since the algorithm needed to handle any possible value of $C_{\ell}(S)$ on sample $S$  for $\ell \leq k - 1$, and weaker quasirandomness parameters were necessary for the induction. However, once we are finally ready to test the number of $k$-cliques in the graph, given approximate guarantees about the number of $\ell$-cliques in the graph for $\ell \leq k - 1$, we are dealing with a promise problem where either $G \sim \gnp$ or $C_k(G) \not \in (1 \pm \varepsilon) \binom{n}{k}k! p^{\binom{k}{2}}$, which is reflected in the following corollary.

\begin{corollary}[Corollary of \Cref{lem:concentration}]\label{cor:another-cor-of-concentration}
    Let $s\geq s_k$. Suppose $G$ is $(\varepsilon, \alpha_{k-1}, s)$-exponentially robustly quasirandom with respect to $(k-1)$-clique counts. Sample a multiset $S$ by selecting $s$ vertices uniformly at random (with repetition). If $C_k(G) \in (1 \pm \frac{1}{4} \varepsilon) \binom{n}{k} k! \cdot p^{\binom{k}{2}}$, then with probability at least $1 - 4^{k + 1} \exp(-\alpha_k s)$, $C_k(S) \in (1 \pm \frac{1}{2} \varepsilon) E_{s, k}$. If $C_k(G) \not \in (1 \pm \varepsilon) \binom{n}{k} k! \cdot p^{\binom{k}{2}}$, then with probability at least $1 - 4^{k + 1} \exp(-\alpha_k s)$, $C_k(S) \not \in (1 \pm \frac{3}{4} \varepsilon) E_{s, k}$.
\end{corollary}

\begin{proof}
    The proof follows similarly to that of \Cref{cor:cor-of-concentration-clique}.
\end{proof}

For the completeness argument of \Cref{thm:k-cliques}, we need to argue that \textit{global} concentration of clique counts around the mean allows us to conclude about \textit{local} concentration -- that is, exponential quasirandomness. This is given in the following claim.

\begin{claim}\label{claim:completeness-concentration}
       Suppose that a graph $G$ satisfies the property that for all $\ell\in\{2,\dots,k\}$ the number of labeled $\ell$-cliques is in the range $(1\pm\varepsilon_1/4)\binom{n}{\ell}\ell!p^{\binom{\ell}{2}}$. Then, for all $\ell\in[k]$, $G$ is $(\varepsilon_1,\alpha_{\ell},s_{\ell})$-exponentially robustly quasirandom with respect to $\ell$-clique counts.
\end{claim}

\begin{proof}
    We prove this statement by induction on $\ell$, using a repeated application of \Cref{lem:concentration}.

    \textit{Base case}: Consider $\ell=1$. By Observation \ref{obs:exp-qr-for-vertices}, every graph is $(0,\infty,s_1)$-exponentially robustly quasirandom with respect to the vertex count.

    \textit{Inductive hypothesis}: Suppose that $G$ is $(\varepsilon_1,\alpha_{\ell-1},s_{\ell-1})$-exponentially robustly quasirandom with respect to $(\ell-1)$-clique counts.

    \textit{Inductive argument}: Fix any $u\geq s_{\ell}$. Since $s_{\ell}\geq s_{\ell-1}$, the instantiated concentration lemma with $\gamma=\varepsilon_1$ gives
    $$\Pr_S\left[|C_{\ell}(S)-C_{\ell}(G)A_{u,\ell}|\geq\frac{\varepsilon_1}{4}E_{u,\ell}\right]\leq4^{\ell+1}e^{-\alpha_{\ell}u}.$$
    The global count assumption contributes relative error at most $\varepsilon_1/4$. Thus, outside the exceptional event, $C_{\ell}(S)\in(1\pm\varepsilon_1/2)E_{u,\ell}$. This holds for every $u\geq s_{\ell}$ and proves the induction.
\end{proof}

\begin{proof}[Proof of \Cref{thm:k-cliques} -- query complexity upper bound]
      
    \textbf{Completeness.} By \Cref{lem:Gnp-clique-concentration}, with the accuracy chosen as $\varepsilon_1/4$, $G\sim\gnp$ has all clique counts in the ranges required by Claim \ref{claim:completeness-concentration}, with high probability. Suppose that this event holds. The claim and its proof imply that, for each $\ell\in[k]$, the uniformly sampled multiset $S$ satisfies $C_{\ell}(S)\in(1\pm\varepsilon_1/2)E_{s_*,\ell}$ except with probability $4^{\ell+1}e^{-\alpha_{\ell}s_*}$.

    A union bound over $\ell$ bounds the total error by $\tau$. Since $\varepsilon_1/2$ is at most every acceptance threshold, the algorithm accepts $G\sim\gnp$ with high probability.

    \textbf{Soundness.} 
    We begin by noting that, by Observation \ref{obs:exp-qr-for-vertices}, every input graph $G$ is $(0, \infty, s_1)$-exponentially robustly quasirandom with respect to the count of vertices. Next, suppose that, for some $\ell \in \{2, 3, \dots, k-1\}$, the input graph $G$ is $(\varepsilon_{\ell'}, \alpha_{\ell'}, s_{\ell'})$-exponentially robustly quasirandom with respect to the count of $\ell'$-cliques for all $\ell' < \ell$ and is \textit{not} $(\varepsilon_{\ell}, \alpha_{\ell}, s_{\ell})$-exponentially robustly quasirandom with respect to the count of $\ell$-cliques. Then, by the inductive lemma (\Cref{lemma:inductive-cliques}), with probability at least $1 - 4^{\ell + 1}\exp(-\alpha_{\ell}s_{\ell}) \geq 9/10$, the algorithm rejects on $G$ (on the basis of the count of $\ell$-clique counts in the sampled $S$ being too small or large).

    Suppose that the input graph is $(\varepsilon_{k-1}, \alpha_{k-1}, s_{*})$-exponentially robustly quasirandom with respect to the count of $(k-1)$-cliques. Equivalently, since $\varepsilon_{k-1} = \varepsilon$ (the epsilon as input to our quality control problem), we can assume $G$ is $(\varepsilon, \alpha_{k-1}, s)$-exponentially robustly quasirandom with respect to the count of $(k-1)$-cliques. By \Cref{cor:another-cor-of-concentration}, with probability at least $1 - 4^{k + 1}\exp(-\alpha_k s)$, $C_k(S) \not \in (1 \pm \frac{3}{4} \varepsilon) E_{s, k}$ and the algorithm outputs reject.

    Therefore, the algorithm rejects each $G$ with $C_k(G) \not \in (1 \pm \varepsilon) \binom{n}{k}k! p^{\binom{k}{2}}$ with high probability.

    \textbf{Query complexity:} The algorithm solely queries all adjacencies between vertices in a $S$ of size $s_{*}$ (as defined in \Cref{def:s-size}). Therefore, the algorithm makes $\binom{s_{*}}{2} = \frac{1}{p^{O(k)}} \cdot \frac{1}{\varepsilon^4}$ adjacency matrix queries.

    \textbf{Runtime:} The algorithm stated counts $\ell$-cliques, for all $\ell \leq k$, for the induced subgraph on a set $S$ of $s_{*}$ vertices. For worst-case input graphs, this can take up to $O(s_{*}^k) = \frac{1}{p^{O(k^2)}} \cdot \frac{1}{\varepsilon^{2k}}$ time.
\end{proof}

\subsection{Query complexity lower bound}\label{sec:cliques-lowerbound}

We now prove \Cref{thm:k-cliques} Part (3), more formally given as follows.

\begin{theorem}[\Cref{thm:k-cliques} Part (3)]\label{thm:qc-lowerbound-cliques} Consider parameters $n, p$ such that $0<p\leq1/2$ and $n\geq C_kp^{-2k}$ for a sufficiently large constant $C_k$. The $(\gnp, \rho_k)$-quality control problem cannot be solved with $o\left(p^{-(k-1)/2}\right)$ queries to the adjacency matrix.
\end{theorem}

This theorem is a direct implication of a result we prove in the more general setting of lower bounds for $(\gnp, \rho_{H})$-quality control for any motif $H$. We defer the proof to \Cref{sec:motifs-lowerbound}.

\subsection{Algorithmic complexity}\label{sec:algorithmic-complexity}

We have proven that a query complexity of $\frac{1}{p^{O(k)}}$ suffices for the quality control problem for comparing the $k$-clique count to \erdosrenyi graphs. In this section, we prove \Cref{thm:k-cliques} Part (2), which states that there is an algorithm for $(\gnp, \rho_k)$-quality control with query complexity \textit{and runtime} of $\frac{1}{p^{O(k)}}$. To do so, we enhance the algorithm from \Cref{sec:qc-upperbound} by leveraging the notion of \textit{graph jumbledness} \cite{thomason1987pseudo}. Recall that graph jumbledness (defined in \Cref{def:jumbled}) captures the property that every pair of vertex subsets $X, Y$ has approximately $p |X| |Y|$ edges between $X$ and $Y$. We will utilize the fact that a sufficient condition for graph jumbledness is verifiable in $\text{poly}(N)$-time on an $N$-vertex graph. Once checked, jumbledness implies that the count of motifs approximately matches the count expected from \erdosrenyi graphs (\Cref{prop:cfz-jumbled-subgraph}). We are therefore again leveraging \textit{composability} of quality control (see \Cref{composability}), now performing quality control for jumbledness of a graph. Thus, we can view jumbledness as a \textit{quality control proxy} for counting $\ell$-cliques for $2 \leq \ell \leq k$.

\begin{theorem}[Theorem for $k$-cliques]\label{thm:k-cliques-runtime}
    For constant error parameter $\varepsilon$, and \\ $p = \Omega\left(\log (n) / n\right)^{1/(4k-1)}$ with $n\geq B_{k,\varepsilon,\tau}p^{-10k}$, there is an algorithm for $(\gnp, \rho_k)$-quality control that makes $\frac{1}{p^{O(k)}}$ adjacency matrix queries to the input graph and has a runtime of $\frac{1}{p^{O(k)}}$.
\end{theorem}

Recall the definitions of $\varepsilon_{\ell}, \alpha_{\ell}$ and $s_*$ from \Cref{sec:qc-upperbound} in \Cref{def:alpha-ell} and \Cref{def:s-size}.

Let $C=C(k,\varepsilon,\tau)\geq 10$ be sufficiently large that \Cref{prop:cfz-jumbled-subgraph} gives the required accuracy for all $\ell$-cliques with $\ell\leq k$. Define
$$t=\left\lceil\max\left\{s_*,\frac{2400C^4\ln(2C)}{p^{5k}}\right\}\right\rceil.$$
We use $t$ in place of $s_*$ in this section. For completeness, suppose additionally that $n\geq B_{k,\varepsilon,\tau}p^{-10k}$, so $\binom{t}{2}/n\leq\tau/4$.

Recall that, as defined in \Cref{eq:expected-ell-clique-copy-amount}, $E_{s, \ell}$ is the expected number of (distinct labeled induced) $\ell$-cliques in a subgraph on a multiset of $s$ vertices sampled with replacement, when the underlying graph is drawn from $\gnp$. 

\begin{algorithm}[H]
\caption{Algorithm \textsc{Clique-Quality-Efficient} performs quality control on the $k$-clique count of $G$ compared to $\gnp$ in $1/p^{O(k)}$ queries \textit{and time}.}
\vspace{1em} ~~\textsc{Clique-Quality-Efficient}: On inputs $G, n, p, \gnp, \varepsilon$, and $k$
\begin{quote}
\begin{enumerate}
     \item Sample $t$ vertices, each uniformly at random (with repetition). If a vertex is sampled twice, reject and halt. Let $S$ be the corresponding set of vertices, with induced graph $G|_S$. 
     \item Perform adjacency matrix queries on all pairs in $S$ to find all edges. 
     \item Test that $G|_S$ is 
     $(p, \beta)$-jumbled for $\beta \leq \frac{p^k t}{C}$ for $C=C(k,\varepsilon)$ as above
     by checking the following: 
     Let $\delta = p^{2k}/(12C^2)$. For every $u \in S$, check that the degree of $u$ in $S$ is $(t-1)(p \pm \delta)$. For every $(u, v) \in S \times S$ with $u\ne v$, check that the number of common neighbors of $(u, v)$ is $(t - 2)(p^2 \pm \delta)$. If any vertex in $S$ or pair of distinct vertices in $S\times S$ does not satisfy this, reject and halt. 
     \item Else, accept.
\end{enumerate}
\end{quote}
\label{alg:clique-quality-efficient}
\end{algorithm}

Observe that, the runtime of this algorithm is $p^{-O(k)}$. Also note that the subgraph count from \Cref{prop:cfz-jumbled-subgraph} does not account for repetitions for multiple appearances of the same vertex in $S$; however, we will argue in the proofs below that, since the probability of a repeated vertex is at most $\binom{t}{2}/n\leq\tau/4$, this does not impede the success of the algorithm.

We now prove \Cref{thm:k-cliques-runtime}. From \Cref{sec:qc-upperbound}, we know that it suffices to reduce the global $k$-clique counting problem to the local problem of counting $k'$-cliques, for $k' \leq k$, over subgraphs with $t$ (as in \Cref{def:s-size}) vertices. While our algorithm does not directly count $k'$-cliques for $k' < k$, we still need to argue that the number of vertices $t$ we sample is still sufficient. We need to handle the inclusion of the test for jumbledness in the analysis. We argue that, if the count of $(k'-1)$-cliques is approximately matching \erdosrenyi graphs, we can efficiently count $k'$-cliques. If the count of $(k'-1)$-cliques is not what we need, jumbledness will be violated.

\begin{proof}[Proof of \Cref{thm:k-cliques-runtime}] 
   \textbf{Completeness.} We first argue that, with high probability, $G \sim \gnp$ is accepted by Algorithm \textsc{Clique-Quality-Efficient}. Consider $G \sim \gnp$, and choose a random subgraph of $G$ on $t$ vertices. By \Cref{cor:jumbled-subgraphs}, for $p \geq \left( \frac{1200 C^4\log n}{n }\right)^{1/(4k - 1)}$ and $t\geq2400C^4\ln(2C)/p^{5k}$, with probability $1 - o(1)$, the subgraph on these $t$ vertices satisfies the conditions checked for in Step 3 with high probability, as proven in \Cref{prop:Gnp-jumbled-conditions}. Therefore, Step 3 does not reject $G \sim \gnp$ with high probability. Thus, the algorithm accepts with high probability.

   \textbf{Soundness.} Suppose that $C_k(G)\notin(1\pm\varepsilon)\binom{n}{k}k!p^{\binom{k}{2}}$. Let $\ell<k$ be the first level at which the required quasirandomness fails, if such a level exists. By \Cref{lemma:inductive-cliques} at sample size $t\geq s_{\ell}$, with probability at least $9/10$, $C_{\ell}(S)\notin(1\pm\varepsilon_{\ell-1})E_{t,\ell}$. If no such level exists, \Cref{cor:another-cor-of-concentration} gives $C_k(S)\notin(1\pm\varepsilon/2)E_{t,k}$ with probability at least $9/10$.

   Acceptance gives jumbledness, and \Cref{prop:cfz-jumbled-subgraph} then gives all local clique counts within relative error $\varepsilon_1/16$ of $\binom{t}{\ell}\ell!p^{\binom{\ell}{2}}$, for sufficiently large $C$. The edge count follows directly from the degree checks. For sufficiently large $n$, the factor $\binom{n}{\ell}\ell!/n^{\ell}$ differs from $1$ by at most $\varepsilon_1/16$, so these counts satisfy every query-only acceptance threshold. This contradicts either of the high-probability events above. Therefore, the efficient algorithm rejects every graph with an inaccurate global $k$-clique count with probability at least $9/10$.

   \textbf{Query complexity:} The algorithm queries all pairs of vertices in a set $S$ of size $t$ (\Cref{def:s-size}), therefore using $O_{k,\varepsilon,\tau}(1/p^{O(k)})$ adjacency matrix queries.\\
   
   \textbf{Runtime:} Step 1 of the algorithm takes time $O(t)$. Step 2 takes time $O(t^2)$. Step 3 takes time $O(t^3)$. Therefore, since $t= 1/p^{O(k)}$, the runtime of Algorithm \textsc{Clique-Quality-Efficient} is $1/p^{O(k)}$.
\end{proof}

\section{Quality control of general motif counts}\label{sec:query-complexity-motifs}

In this section, we show how to extend our results to quality control of the count of \textit{general motifs} compared to $\gnp$. We prove the \Cref{thm:intro-motifs}, restated below.

We begin by introducing some necessary notation. First, let $C_H(G)$ be the number of labeled, induced copies of a motif $H$ in the graph $G$. Let $\rho_H(G) = C_H(G)/ \mathbb{E}_{G' \sim \gnp} \left[ C_{H}(G') \right]$. Second, let $m_H := \max_{F \subseteq H} e(F)/v(F)$\footnote{The quantity $m_H$ captures something similar to the \textit{arboricity} of graphs (\Cref{def:arboricity}), up to additive constants in the numerator and denominator of the expression of $m_H$.}. Let $m_H^{\star}$ be the maximum value of $m_{H'}$ over all $H'$ with $|H'| = |H|$ and edges that are a superset of the edges of $H$. The notion of $m_H^{\star}$ is important because, in $\gnp$, $p = \omega\left( n^{-1/m_{H}^{\star}}\right)$ suffices to have concentration of the counts of all submotifs $F$ of some motif $H$. Third, recall the definition of $\sigma(H)$ from \Cref{eq:s-H}.

\begin{theorem}[\Cref{thm:intro-motifs}]\label{thm:qc-motif-in-section}
    There exist constants $c_1, c_2, c_3$ ($c_2 < c_1 < c_3$) such that for every fixed graph $H$ on at least two vertices, $0<p\leq1/2$, and $p = \omega\left( n^{-1/m_{H}^{\star}}\right)$, the quality control problem $(\gnp, \rho_H)$ corresponding to estimating the number of copies of $H$ in a graph $G$ supposedly drawn according to $\gnp$ satisfies the following:
    \begin{enumerate}
        \item The problem is solvable in $O(p^{-c_1 \Delta(H)})$ adjacency matrix queries to $G$.
        \item  When $n$ is additionally at least $B_{H,\varepsilon,\tau}p^{-10(\Delta(H)+3)}$, the problem is also solvable in $O(p^{-c_3 \Delta(H)})$ time and queries.
        \item  Furthermore, for $n\geq C_Hp^{-2\Delta(H)-2}$, this problem cannot be solved with $o(p^{-c_2 \Delta(H)})$ adjacency matrix queries.
    \end{enumerate}
\end{theorem}

To illustrate this bound, we give some implications of this result for specific motifs $H$:
\begin{corollary}\label{cor:specific-motifs}
    The following are implied by \Cref{thm:qc-motif-in-section}.
    \begin{enumerate}
        \item Consider any cycle $H$ on $k$ vertices. $(\gnp, \rho_H)$-quality control takes $1/p^{\Theta(1)}$ adjacency matrix queries and time. Observe that the exponent of $1/p$ is independent of the size of the cycle and $n$. In comparison, approximating the number of non-induced cycles on $k$ vertices in worst-case graphs, at the corresponding target count and with access to random edge samples, takes $1/p^{O(k)}$ queries \cite{assadi2018simple}. For the colored non-induced counting problem, this is matched by a $1/p^{\Omega(k)}$ lower bound in the uncapped range.
        \item Consider a $k$-clique $H$. $(\gnp, \rho_H)$-quality control takes $1/p^{\Theta(k)}$ adjacency matrix queries and time. In comparison, approximate $k$-clique counting in worst-case graphs takes $1/p^{\Theta(k^2)}$ queries to the adjacency lists, adjacency matrix, and degrees in the uncapped range described above \cite{eden2018faster}.
        \item Consider the star graph $H$ on $k$ vertices. $(\gnp, \rho_H)$-quality control takes $1/p^{\Theta(k)}$ adjacency matrix queries and time. 
        In comparison, approximating the number of non-induced $k$-stars in worst-case graphs at the corresponding target count takes $\widetilde{O}(n/p)$ queries, including random edge queries \cite{assadi2018simple}, which is incomparable to our result for quality control.
    \end{enumerate}
\end{corollary}

\subsection{Query complexity upper bound}\label{sec:motifs-upperbound}

In this section, we prove \Cref{thm:qc-motif-in-section} Part (1), which states the query complexity of quality control for the count of a motif $H$ compared to $\gnp$ is $p^{-O(\Delta(H))} \cdot \mathrm{poly}(1/\varepsilon)$, where $\Delta(H)$ is the maximum vertex degree of $H$. 

We prove quality control results for comparing \textit{labeled induced motif} counts to the expected amount from $\gnp$. By an induced motif, we mean the number of \textit{induced subgraphs} in the graph that form a copy of the motif. Recall that an induced subgraph is a graph formed from a subset of the vertices of the graph and \textit{all} of the edges and non-edges between these vertices.

We will show how our results are implied from a deferred result (given in \Cref{sec:any-random-graph-family}) that holds for $(D, \rho_H)$-quality control for any distribution $D$ over graphs that has good concentration of motifs $H' \subseteq H$. This broader result relies on the same principles as the proof for quality control of the $k$-clique count compared to $\gnp$. That is, a notion of \textit{exponential quasirandomness} (see \Cref{def:exp-qr-motif}) is given for motif counts, and the proof iteratively proves (or, at some point, refutes) that exponential quasirandomness holds for the input graph. A concentration lemma (\Cref{lem:concentration-general-P-and-D}) and inductive lemma (\Cref{lemma:inductive-general-P-and-D}) are developed for this context as well.

We now turn to setting up the notation and framework needed to apply \Cref{sec:any-random-graph-family}.

For a motif $H$ on $\ell$ vertices, the expected number of induced motif counts in $G \sim \gnp$ is 
$\binom{n}{\ell} \ell! \cdot p^{e(H)}(1-p)^{\binom{\ell}{2} - e(H)}$.

\paragraph{Notation.} Recall the following notation from \Cref{sec:technical-lemmas}.

For a motif $M$ on $\ell$ vertices, as in \Cref{eq:expected-P-number-general}, let $F_{\mathcal{D}}(M)$ be defined as:
$$F_{\gnp}(M) = \mathbb{E}_{G \sim \gnp}\left[ \# \text{ labeled induced copies of } M \text{ in } G\right] \cdot \frac{1}{\binom{n}{\ell} \ell!}.$$
$F_{\gnp}(H)$ captures the \textit{density}/frequency of $H$ in $\gnp$. By density, we mean the probability that a uniformly random ordered tuple of $\ell$ distinct vertices from $G$ forms a labeled induced copy of $M$. We abbreviate $F_{\gnp}(H)$ as $F(H)$ in this section, when it is clear from context that the input graph is being compared to the distribution $\gnp$ over graphs.

Suppose we sample $s$ vertices uniformly at random with repetition. Let
\begin{equation} \label{eq:s-D-H-motif-gnp}
    S_{\gnp}(H, s) = E_{\gnp}(H) \cdot \binom{s}{\ell} \frac{\ell!}{n^{\ell}}
\end{equation}
be the expected number (over $\gnp$ and the choice of random vertices) of labeled induced copies of $H$ in the subgraph induced by a multiset of size $s$ consisting of vertices chosen uniformly at random (with replacement). It is proven in Claim \ref{claim:expectation-set-P-copies} that this is the correct expression for the expectation. We abbreviate $S_{\gnp}(H, s)$ as $S(H, s)$ in this section, when the underlying distribution over input graphs is clear from context.

For a motif $H$ on $k$ vertices, let $\mathcal{H}_{k - 1}$ be the set of all (labeled induced) subgraphs of $H$ on $k - 1$ vertices. Define $r_{\ell}(H)$ as in \Cref{eq:r-j-H}:
$$r_{\ell}(H) = \min_{H_1 \in \mathcal{H}_{\ell}}\frac{F(H_1)^2}{\max_{\substack{H_2 \subseteq H_1 \\ |H_2| = \ell-1}}  F(H_2)^2}.$$
The quantity $r_{\ell}(H)$ captures the ratio of the frequency of subgraphs $H_1$ on $\ell$ vertices with their most-frequent subgraphs $H_2 \subseteq H_1$ on $\ell - 1$ vertices.

\paragraph{Proving the upper bound.} To apply \Cref{thm:general-graph-motif-count-qc} (quality control over motif counts for a general family of distributions over graphs), we need to understand the values of $r_{\ell}$ for $\ell \in \{2, 3, \dots, k\}$. We prove that, for $G \sim \gnp$, the minimum value of $r_{\ell}(H)$ is at least $p^{2\Delta(H)}(1-p)^{2k-2}$, where $\Delta(H)$ is the maximum vertex degree in $H$.

Let us first understand $F_{\gnp}(M)$ for any motif $M$ in the graph. Suppose $M$ is a motif on $\ell$ vertices. Then 
    $$F_{\gnp}(M) = p^{e(M)} (1-p)^{\binom{\ell}{2} - e(M)}.$$
    Observe that, for $p \leq 1/2$, $F(M) \leq F(M')$ if $|M|=|M'|$ and $e(M) \geq e(M')$.

Let us first understand $r_{\ell}(H)$ for $|H| = k, \ell \in [k]$ when $G \sim \gnp$.

\begin{lemma}\label{lem:r-ell-motif-for-gnp}
    Suppose $G \sim \gnp$, $p \leq 1/2$, and consider a motif $H$ on $k$ vertices. For all $2\leq\ell\leq k$, the following holds:
    $$r_{\ell}(H) = (1-p)^{2 \ell - 2} \min_{H_1 \in \mathcal{H}_{\ell}} \left( \frac{p}{1-p}\right)^{2\Delta(H_1)}.$$
\end{lemma}

\begin{proof}
    Let's expand the expression for $r_{\ell}(H)$ for $G \sim \gnp$:
    $$r_{\ell}(H) = \min_{H_1 \in \mathcal{H}_{\ell}}\frac{F(H_1)^2}{\max_{\substack{H_2 \subseteq H_1 \\ |H_2| = \ell-1}}  F(H_2)^2} =  \min_{H_1 \in \mathcal{H}_{\ell}}\frac{p^{2e(H_1)} (1-p)^{2\binom{\ell}{2} - 2 e(H_1)}}{\max_{\substack{H_2 \subseteq H_1 \\ |H_2| = \ell-1}}  p^{2e(H_2)} (1-p)^{2\binom{\ell-1}{2} - 2 e(H_2)}}$$
    $$= (1-p)^{2 \ell - 2} \min_{H_1 \in \mathcal{H}_{\ell}}\frac{p^{2e(H_1)} (1-p)^{- 2 e(H_1)}}{\max_{\substack{H_2 \subseteq H_1 \\ |H_2| = \ell-1}}  p^{2e(H_2)} (1-p)^{- 2 e(H_2)}}.$$
    Since $p \leq 1/2$, in the denominator, the maximum over $H_2 \subset H_1$ of the expression $p^{2e(H_2)} (1-p)^{- 2 e(H_2)}$ is achieved when $H_2$ has as few edges as possible. That is, $H_2$ is $H_1$ with the highest-degree vertex removed; i.e. the maximum is achieved for $H_2$ satisfying $e(H_2) = e(H_1) - \Delta(H_1)$.

    Therefore, the expression for $r_{\ell}(H)$ becomes:
    $$r_{\ell}(H) = (1-p)^{2 \ell - 2} \min_{H_1 \in \mathcal{H}_{\ell}}\frac{p^{2e(H_1)} (1-p)^{- 2 e(H_1)}}{p^{2e(H_1) - 2 \Delta(H_1)} (1-p)^{- 2 e(H_1) + 2 \Delta(H_1)}} $$ $$ = (1-p)^{2 \ell - 2} \min_{H_1 \in \mathcal{H}_{\ell}} \left( \frac{p}{1-p}\right)^{2\Delta(H_1)},$$
    as stated in the lemma.
\end{proof}

We verify that the values of $r_{\ell}(H)$ decrease as $\ell$ increases, which is a necessary condition for applying \Cref{thm:general-graph-motif-count-qc}. This corresponds to the statement that the density of subgraphs of $H$ on $\ell$ vertices in the subsets of $\ell$ vertices is less than the density of their respective subgraphs on $\ell - 1$ vertices in the subsets of $\ell - 1$ vertices. 

\begin{lemma}
    For all $\ell \in \{2, 3, \dots, k\}$, $r_{\ell-1}(H) \geq r_{\ell} (H)$.
\end{lemma}

\begin{proof}
    For $\ell=2$, the claim follows from $r_2(H)\leq r_1(H)=1$. For $\ell\geq3$, first observe that $(1-p)^{2(\ell-1)-2}\geq(1-p)^{2\ell-2}$. Next, observe that
    \begin{equation}\label{eq:motifs-r-ell-ratio}
        \min_{H_2 \in \mathcal{H}_{\ell-1}} \left( \frac{p}{1-p}\right)^{2\Delta(H_2)} \geq \min_{H_1 \in \mathcal{H}_{\ell}} \left( \frac{p}{1-p}\right)^{2\Delta(H_1)}.
    \end{equation}
    This holds because $\left(p/(1-p)\right)^c$ decreases as $c$ increases, for $c > 0$. As we consider $\Delta(H_1)$ over subsets $H_1$ on more vertices, the maximum maximum-degree $\Delta(H_1)$ can only increase, giving us \Cref{eq:motifs-r-ell-ratio}.

    Combining the two observed inequalities gives us the lemma.
\end{proof}

Next, as noted in \Cref{sec:any-random-graph-family}, the query complexity will depend on $r_{k}(H)^{-1}$ (the ratio of densities of subgraphs of $H$). We observe the following.
\begin{observation}
    $r_{k}(H) = p^{2 \Delta(H)} (1-p)^{2k - 2 - 2 \Delta(H)}$.
\end{observation}

\begin{proof}
    This follows directly from \Cref{lem:r-ell-motif-for-gnp}, observing that, by definition, $\mathcal{H}_k = \{ H\}$ when $H$ is a motif on $k$ vertices.
\end{proof}

\begin{observation}\label{obs:s-value-motifs}
    Define $s_{*}$ as in \Cref{def:s-general-distribution}. Then $s_{*} = \widetilde{O}\left(r_k(H)^{-1} \varepsilon^{-2}\right) = \widetilde{O}\left(p^{-2 \Delta(H)} \varepsilon^{-2}\right).$
\end{observation}

We are now ready to prove the  upper bound in \Cref{thm:qc-motif-in-section} as an implication of \Cref{thm:general-graph-motif-count-qc}.

\begin{proof}
    Consider a motif $H$ on $k$ vertices. By \Cref{lem:Gnp-motif-concentration}, for $p = \omega\left( n^{-1/m_H^{\star}}\right)$ $\gnp$ satisfies the properties that all $H'\subseteq H$ the number of labeled induced copies of $H'$ is in the range $(1\pm\varepsilon_1/4)E_{\gnp}(H')$  with high probability, and that $r_{\ell - 1}(H) \geq r_{\ell}(H)$ for all $\ell \in \{2, 3, \dots, k\}$. Therefore, \Cref{thm:general-graph-motif-count-qc} implies that quality control for the count of $H$ in an input graph $G$ compared to $\gnp$ can be accomplished in $O\left( r_k(H)^{-2} \varepsilon^{-4} \right) =  p^{-O(\Delta(H))} \cdot \mathrm{poly}(1/\varepsilon)$ queries.
\end{proof}

We conclude with the following observation about quality control for unlabeled or non-induced motifs. First, recall the following facts regarding the relationship between induced/non-induced and labeled/unlabeled motifs in a graph. The number of non-induced copies of a motif $H$ is the sum of the count of all induced motifs $H'$ of size $k$ such that the edges of $H$ are a subset of the edges of $H'$. For an unlabeled motif $H$, let $\text{Aut}(H)$ be the set of automorphisms of $H$; that is, the set of all permutations of the vertices of $H$ that preserve the structure/adjacencies of $H$. Let $|\text{Aut}(H)|$ be the number of automorphisms of $H$. Recall that the number of labeled copies of a motif $H$ in a graph $G$ is the number of unlabeled copies times $|\text{Aut}(H)|$.

\begin{observation}\label{obs:induced-to-noninduced-motif-counts}
    Quality control for non-induced motifs $H$ (of size $k$) compared to $\gnp$ can be accomplished via quality control algorithms for labeled counts of motifs $H'$ of size $k$ such that the edges of $H$ are a subset of the edges of $H'$. Quality control for labeled motif counts can be accomplished directly from quality control for unlabeled motif counts, and vice versa.
\end{observation}

\subsection{Query complexity lower bound}\label{sec:motifs-lowerbound}

In this section, we prove \Cref{thm:qc-motif-in-section}, Part (3). That is, we prove a lower bound on the number of queries (and thus the runtime) needed to perform quality control on the count of a general motif as compared to \erdosrenyi graphs.

For graphs $H$ and $G$, let $C_{H}(G)$ (\Cref{def:C-H-G}) count the number of induced labeled copies of $H$ in $G$. For a parameter $p\in [0,1]$ and integer $n$, let $\mu_H = \mu_H^{n,p}$ denote the quantity $\mu_H := \E_{G \sim \gnp}[C_{H}(G)]$. Finally,
let $\rho_H(G) = \rho_H^{n,p}(G)$ equal $C_{H}(G)/\mu_H$. 
For a graph $H$, let $\Delta(H)$ denote the maximum vertex degree in $H$.

\Cref{thm:qc-motif-in-section} Part (3) states that the  $(\gnp,\rho_H)$-quality control problem, i.e., the quality control problem for $\gnp$ with respect to $\rho_H$, requires
$p^{-\Omega(\Delta(H))}$ queries for every graph $H$. Specifically, we have the following theorem. 

\begin{theorem}\label{thm:query-lb}
   For every graph $H$ on at least two vertices, there exists a constant $c > 0$ and polynomial $n_0$ such that for every $p \in (0, 1/2]$ and every $n \geq n_0(1/p)$ , the $(\gnp,\rho_H)$-quality control problem cannot be solved with query complexity less than $c \cdot p^{-\Delta(H)/2}$.
\end{theorem}

Our results follow the standard paradigm for lower bounds in graph property testing. By Yao's lemma, to show a lower bound on the number of queries any (possibly randomized) algorithm needs for a problem, it suffices to show that no deterministic algorithm can distinguish between inputs from two certain distributions. One distribution must be supported on instances that a quality control algorithm should (usually) accept, and one distribution should be supported on instances that a quality control algorithm should (usually) reject.

Specifically, we consider two distributions on graphs: a ``Yes'' distribution\footnote{In our case, since we are performing quality control for $\gnp$, we are forced to set $D_Y = \gnp$.} $D_Y= D_Y^n := \gnp$ which is mostly supported on graphs $G$ that have $\rho(G) \approx 1$, and a family of ``No'' distributions $D_N = D_N^{n,\ell}$ that aims to be mostly supported on graphs with $|\rho(G)-1| \gg \epsilon$. Our (easy) ``indistinguishability'' lemma (which is independent of $H$) shows that no deterministic $o(n/\ell)$ query algorithm gets much advantage in distinguishing instances drawn from $D_Y^n$ from instances drawn from $D_N^{n, \ell}$. The (slightly more) involved ``parameter-separation'' lemma shows that for an appropriate choice of $\ell$, depending on $H$, $D_N$ does place most of its mass on instances with $|\rho(G)-1| \gg \epsilon$. Putting the two together yields a proof of \Cref{thm:query-lb}.

\Cref{def:no-distribution} below defines our distributions $D_N^{n,\ell}$. It is followed by \Cref{lem:dy-dn-indist}, which proves the indistinguishability of $D_Y^n$ and $D_N^{n, \ell}$, and then by \Cref{lem:param-sep}. which proves that $D_N$ is mostly supported on instances with $\rho$-value far from $1$. We conclude the section with a proof of \Cref{thm:query-lb}.

\begin{definition}\label{def:no-distribution}
     We describe $D_N$ by showing how to sample a graph $G'$ according to this distribution: $G' \sim D_N^{n,\ell}$ is sampled as follows: We first sample $G \sim \gnp$. Then we sample $T \sim \gnhalf$. Finally we sample $S \subseteq [n]$ with $|S| = \ell$ uniformly among all such sets. Our output graph $G' = G'(G,S,T)$ is set to be equal to $G$ on $\overline{S} \times \overline{S}$ and equal to $T$ otherwise. In other words, $(u,v) \in E(G')$ if $\{u,v\} \cap S = \emptyset$ and $(u,v)\in E(G)$ or if $\{u,v\} \cap S \ne \emptyset$ and $(u,v) \in E(T)$.
\end{definition}

\begin{lemma}\label{lem:dy-dn-indist}
    Let $A$ be a deterministic algorithm making $q$ queries to the adjacency matrix of an $n$-vertex graph and outputting $0/1$. Then 
    $$\left|\Pr_{G \in D_Y^n}[A(G) = 1] - \Pr_{G'\in D_N^{n,\ell}}[A(G')=1] \right| \leq \frac{2q\ell}{n-\ell}.$$
\end{lemma}

\begin{proof}
    We prove something even stronger, namely that for every graph $G$, with high probability over the choice of $S \subseteq [n]$ with $|S|=\ell$ and for every graph $T$ we have 
        $$\Pr_S\left[A(G'(G,S,T))\ne A(G)\right] \leq \frac{2q\ell}{n-\ell},$$
    where the notation $G'(G,S,T)$ is as in \Cref{def:no-distribution}. The lemma follows by sampling $G \sim \gnp$ and $T \sim \gnhalf$. 

    Note every query of $A$ to its input graph is of the form ``Is $(u,v) \in E$?''. Let $(u_i,v_i)$, for $i = 1$ to $q$ denote the $q$ queries of $A$ on input $G$. (Since $A$ is deterministic and the input graph is fixed, these queries are completely determined.) Let $Q = \{u_i ~|~ i\in [q]\} \cup \{v_i ~|~ i \in [q]\}$. We observe that if $Q \cap S = \emptyset$ then $A(G) = A(G'(G,S,T))$ since every query gets the same answer in $G'$ as in $G$ and $A$'s queries and output are deterministic functions of previous queries and their answers. The lemma follows from the simple claim that for every set $Q \subseteq [n]$ with $|Q|\leq 2q$ the probability that it intersects a random set $S \subseteq [n]$ of size $\ell$ is at most $2q\ell/(n-\ell)$.     
\end{proof}

\begin{lemma}\label{lem:param-sep}
    For every motif $H$ of maximum degree $\Delta$,  there exists a polynomially growing function $f$ such that for every $p \in (0,1/2]$ and $n\geq C_H p^{-\Delta}$ for a sufficiently large constant $C_H$, and every $\ell \geq f(1/p)$ and $n/2 \geq \ell \geq 2^{2\binom{k}2+1} n p^{\Delta/2}$, the following hold:
    $$\Pr_{G \sim D_Y^n}[| \rho_H(G) - 1| \leq 0.1] \geq 1 - 1/12, $$
    $$\mbox{ and } \Pr_{G' \sim D_N^{n,\ell}}[ \rho_H(G') \geq 1.9] \geq 1 - 1/12.$$  
\end{lemma}

\begin{proof} 
    We prove the lemma for $f(x)=C_Hx^{\Delta}$, with $C_H$ sufficiently large. 
    
    On the one hand we have that the expected number of induced labeled copies of $H$ in $G \sim \gnp$ is $\mu_H = \binom{n}k k! \cdot p^{e(H)} (1-p)^{\binom{k}{2} - e(H)} = n^k p^{e(H)} (1-p)^{\binom{k}{2} - e(H)} (1\pm o_n(1))$, and this is concentrated under the stated size assumption by \Cref{lem:Gnp-motif-concentration}, where $m_H = \max_{H' \subseteq H} \frac{e(H')}{v(H')}$. This holds for $n\geq C_Hp^{-\Delta(H)}$. 
    
    As a consequence, we have
    $$\Pr_{G \sim D_Y^n}[| C_{H}(G)-\mu_H| \leq 0.1 \cdot \mu_H ] = 1 - o_n(1) \geq 1 - 1/12. $$

    We now turn to analyzing $\rho_H(G')$ for $G'\sim D_N$. We want to show that $\rho_H(G')$ is $\geq 1.9$ with high probability (equivalently, that $C_H(G') \geq 1.9 \mu_H$ with high probability). Let the vertices of $H$ be denoted as $h_1,\ldots,h_k$, with the degree of $h_k$ being $\Delta$. Let $e(H)$ denote the number of edges of $H$. Let $H'$ be the induced subgraph on vertices $h_1,\ldots,h_{k-1}$, and let $\mu_{H'} = \mathbb{E}_{G' \sim \gnp}[C_{H'}(G')]$. Note that the number of edges of $H'$ is exactly $e(H') = e(H) -\Delta$. We lower bound the expected number of induced copies of $H$ in $G'$ by first considering ordered copies of $H'$ and then extending them.

    Fix distinct vertices $v_1,\ldots,v_{k-1}$. Conditional on the planted set $S$, every required edge among them has probability at least $p$, and every required non-edge has probability at least $1/2$. At least $\ell-k+1$ possible extensions $v_k$ lie in $S$. For each of these extensions, the $k-1$ new edge or non-edge requirements have probability exactly $2^{-(k-1)}$, independently of the requirements among $v_1,\ldots,v_{k-1}$. Summing over the ordered choices of the first $k-1$ vertices gives
    $$\mathbb{E}[C_H(G')]\geq(n)_{k-1}(\ell-k+1)p^{e(H')}2^{-\binom{k}{2}}\geq\mu_{H'}(\ell-k+1)2^{-\binom{k}{2}}.$$
    To compare this with $\mu_H$, note that $\mu_H = \mu_{H'} \cdot n \cdot p^\Delta (1-p)^{k-1-\Delta} (1\pm o_n(1))$. Thus  
    \begin{align*}
        \E[C_{H}(G')] & \geq \mu_{H'} (\ell-k+1) 2^{-\binom{k}2} \\
                    & \geq \mu_H n^{-1} p^{-\Delta} (1-p)^{-(k-1-\Delta)} (1\pm o_n(1)) (\ell-k+1) 2^{-\binom{k}2} \\
                    & \geq \mu_H \frac{\ell}{2^{\binom{k}2+2}n p^{\Delta}}
    \end{align*}
    where the final inequality uses $(1-p)^{-(k-1-\Delta)}\geq 1$, $\ell \geq 2k$ and that $n$ is large enough so that the $o_n(1)$ term is $\leq 1/2$.
    Thus if $\ell \geq 2^{\binom{k}2+3} p^{\Delta} n$ we get that $\E_{G'}[C_{H}(G')] \geq 2\mu_H$. By applying \Cref{thm:sbm-H-concentration} with concentration error $1/40$, we get that  
    $$\Pr_{G' \sim D_N^{n,\ell}}[ C_{H}(G')  \geq 1.9\mu_H] \geq 1 - 1/12,$$
    provided $n\geq C_Hp^{-\Delta(H)}$, $\ell \geq (1/p)^{\Delta}$, and $n/2 \geq \ell \geq 2^{2\binom{k}2+1} n p^{\Delta/2}$.
\end{proof}

\begin{proof}[Proof of \Cref{thm:query-lb}.]  
    If $\Delta=0$, any algorithm must make at least one query: the complete graph has no induced copy of $H$, and an algorithm making no queries cannot distinguish it from $\gnp$. Now suppose $\Delta\geq1$. Let $B=2^{2\binom{k}{2}+1}$ and $p_*=(4B)^{-2/\Delta}$. For $p\geq p_*$, the complete graph has normalized count $0$ unless $H$ is a clique, in which case its normalized count is $p^{-e(H)}\geq2$. Thus at least one query is required, giving the claimed bound after decreasing its constant.

    For $p<p_*$, choose $n_0(x)=C_Hx^{2\Delta+2}$, with $C_H$ sufficiently large, and set $\ell=\lceil Bnp^{\Delta/2}\rceil$. Then $\ell\leq n/2$, and all the size conditions of \Cref{lem:param-sep} hold. We may take $c=1/(96B)$.

    Let $A$ be an algorithm making $q$ queries that accepts $G \sim \gnp$ with probability at least $2/3$ while accepting $G$ with $\rho_H(G) > 1.5$ with probability at most $1/3$. Now note that by our choice of $\ell$ we have, by \Cref{lem:param-sep}, that $\Pr_{G'\sim D_N^{n,\ell}}[\rho_H(G')<1.9]\leq1/12$ (by our choice of $n$).
    So we get $\Pr_{G\sim \gnp,R}[A(G) = 1] - \Pr_{G'\sim D_N^{n, \ell},R}[A(G')=1] \geq 1/6$, where the randomness $R$ denotes the random coin tosses of the algorithm $A$. In particular there exists $R_0$ such that $\E_{G\sim \gnp}[A'(G) = 1] - \E_{G'\sim D_N^{n, \ell}}[A'(G')=1] \geq 1/6$ where $A'$ denotes the deterministic version of $A$ obtained by fixing its randomness to $R_0$. Applying \Cref{lem:dy-dn-indist}, we get that $2q\ell/(n-\ell) \geq 1/6$, which in turn implies $q\geq(n-\ell)/(12\ell)\geq p^{-\Delta/2}/(48B)$.
\end{proof}
\subsection{Algorithmic complexity}

We now move to proving \Cref{thm:qc-motif-in-section} Part (2) regarding the algorithmic complexity of quality control of motifs compared to $\gnp$. We re-state the result below. Recall the definition of $\sigma(H)$ from \Cref{eq:s-H}. Let $m_H^{\star}$ be the maximum value of $m_{H'}$ over all $H' \supseteq H$ with $|H'| = |H|$ and edges that are a superset of the edges of $H$.

\begin{theorem}[\Cref{thm:qc-motif-in-section} Part (2)]\label{thm:motifs-runtime}
    For a fixed motif $H$, constant error parameter $\varepsilon$, and completeness error $\tau>0$, there exists a constant $B=B(H,\varepsilon,\tau)$ such that, for $0<p\leq1/2$ and $n\geq Bp^{-10(\Delta(H)+3)}$, there is an algorithm for $(\gnp,\rho_H)$-quality control that makes $1/p^{O(\Delta(H))}$ adjacency matrix queries and has runtime $1/p^{O(\Delta(H))}$.
\end{theorem}

\subsubsection{Non-induced motifs}\label{sec:noninduced-motifs-algo}
We begin by studying non-induced motifs, moving to induced motifs in \Cref{sec:induced-motifs-algo}. Our query complexity upper bounds so far hold for the setting of non-induced motifs, as noted in Observation \ref{obs:induced-to-noninduced-motif-counts}. 
We will extend \Cref{sec:algorithmic-complexity}, using similar ideas of testing graph jumbledness and then applying \Cref{prop:cfz-jumbled-subgraph}.

We begin by proving \Cref{thm:motifs-runtime} for the quality control of the counts of \textit{non-induced} motifs, and then bound the difference between induced and non-induced counts using one-edge extensions. Considering non-induced motifs allows us to more directly apply \Cref{prop:cfz-jumbled-subgraph} (from \cite{conlon2014extremal}) for approximate counting of motifs in sufficiently jumbled graphs. 

If $H$ is edgeless, its non-induced count is $\binom{n}{k}k!$ on every graph, so accept without queries. Otherwise, let $j=\lceil\sigma(H)\rceil$ and $|H|=k$.

In this section we are considering the case where $p \geq \left(\frac{1200 C^4 \log n}{n} \right)^{1/(4j-1)}$, where $C=C(H,\varepsilon)\geq10$ is sufficiently large for \Cref{prop:cfz-jumbled-subgraph}. In this setting of $p$, we can use the jumbledness of small induced subgraphs of $\gnp$.

For a motif $H$, recall the definition of $\sigma(H)$ from \Cref{eq:s-H}. Suppose that $j = \lceil \sigma(H) \rceil$. Let $S(H, s)$ be the expected number of copies of $H$ in a multiset on $s$ vertices (with repetition), as in \Cref{eq:s-D-H-motif-gnp}.

Define $s_*$ as in \Cref{def:s-general-distribution}, using non-induced counts and frequencies $F(P)=p^{e(P)}$. Their increments are at least $p^{2\Delta(H)}$.

Let $C=C(H,\varepsilon,\tau)\geq10$ be sufficiently large for \Cref{prop:cfz-jumbled-subgraph}, and define
$$t=\left\lceil\max\left\{s_*,\frac{2400C^4\ln(2C)}{p^{5j}}\right\}\right\rceil.$$
For completeness of the non-induced algorithm, suppose additionally that $n=\omega(p^{-10j})$, so the probability of a repeated sampled vertex is at most $\binom{t}{2}/n=o(1)$.

We apply the following algorithm, which is an extension of the algorithm in \Cref{sec:algorithmic-complexity}. 

\begin{algorithm}[H]
\caption{Algorithm \textsc{NonInduced-Motif-Quality-Efficient} performs quality control on the count of a labeled non-induced motif $H$ in $G$ compared to $\gnp$ in $1/p^{O(\Delta(H))}$ queries and time, where $\Delta(H)$ is the maximum degree of $H$.}
\vspace{1em} ~~ \textsc{NonInduced-Motif-Quality-Efficient}: On inputs $G, n, p, \gnp, \varepsilon$, and non-induced motif $H$
\begin{quote}
\begin{enumerate}
     \item Sample $t$ vertices, each uniformly at random (with repetition). If a vertex is sampled twice, reject. Let $S$ be the corresponding set of vertices, with induced graph $G|_S$. 
     \item Perform adjacency matrix queries on all pairs in $S$ to find all edges. 
     \item Test that $G|_S$ is 
     $(p, \beta)$-jumbled for $\beta \leq \frac{p^j t}{C}$ for the constant $C=C(H,\varepsilon,\tau)$ chosen above 
     by checking the following:
     Let $\delta = p^{2j}/(12C^2)$. For every $u \in S$, check that the degree of $u$ in $S$ is $(t-1)(p \pm \delta)$. For every $(u, v) \in S \times S$ with $u\ne v$, check that the number of common neighbors of $(u, v)$ is $(t - 2)(p^2 \pm \delta)$. If any vertex in $S$ or pair of distinct vertices in $S\times S$ does not satisfy this, reject. 
     \item Else, accept.
\end{enumerate}
\end{quote}
\label{alg:noninduced-motif-quality-efficient}
\end{algorithm}

\begin{proof}[Proof of \Cref{thm:motifs-runtime} for non-induced motifs]
    This proof follows the same steps as the proof of \Cref{thm:k-cliques-runtime} from \Cref{sec:algorithmic-complexity} (using our setting of parameters $\varepsilon_{\ell},\alpha_{\ell}$ for this section).

    The two changes are the following. For completeness, now apply \Cref{cor:jumbled-subgraphs} for $\beta\leq p^j t/C$ (that is, replace $k$ by $j=\lceil\sigma(H)\rceil$), and let $C = C(H, \ep)$ as before. We can apply \Cref{cor:jumbled-subgraphs} because $t$ is set so that $t\geq2400C^4\ln(2C)/p^{5j}$. 
    
    Additionally, instead of relying on the inductive and concentration lemmas from \Cref{sec:qc-upperbound} for $k$-cliques (\Cref{lemma:inductive-cliques} and \Cref{lem:concentration}, respectively), rely on the inductive and concentration lemmas from \Cref{sec:technical-lemmas} (\Cref{lemma:inductive-general-P-and-D} and \Cref{lem:concentration-general-P-and-D}, respectively), which hold for general motifs and a range of distributions over graphs.

    For soundness, apply the concentration and induction directly to non-induced counts. For every submotif $P\subseteq H$, $\sigma(P)\leq\sigma(H)$, since deleting edges or vertices cannot increase the maximum degree or degeneracy of the line graph. Thus, by choosing $C$ sufficiently large, acceptance certifies all non-induced local counts to relative error at most $\varepsilon_1/16$. For sufficiently large $B$, the factor $\binom{n}{\ell}\ell!/n^{\ell}$ is within $\varepsilon_1/16$ of $1$ for every $\ell\leq k$, so these counts satisfy the acceptance thresholds of the query-only algorithm. Edges follow directly from the degree checks, and on the distinct sampled set, edgeless non-induced counts are deterministic.

   We now argue about the query complexity and runtime. Since the algorithm queries all pairs of vertices in a set $S$ of size $t$ (\Cref{def:s-general-distribution}; see Observation \ref{obs:s-value-motifs}), the algorithm uses $O(t^2)$ adjacency matrix queries and $O(t^3)$ time. Since $j\leq\Delta(H)+2$, these are $1/p^{O(\Delta(H))}$ when $\Delta(H)\geq1$.
\end{proof}

\subsubsection{Moving to induced motifs}\label{sec:induced-motifs-algo}

We can transfer the results for non-induced motifs to induced motifs. The inclusion-exclusion formula below expresses an induced count in terms of non-induced counts of edge supersets. To keep the dependence on $\Delta(H)$, we will only check one-edge extensions when $p$ is small, and sample induced copies directly otherwise.

Suppose $H$ is supported on $k$ vertices. Let $r=\binom{k}{2}-e(H)$ and $p_0=\varepsilon/(8\max\{1,r\})$. Below, $H+a$ denotes the motif obtained by adding one missing edge $a$ to $H$.

\begin{algorithm}[H]
\caption{Algorithm \textsc{Induced-Motif-Quality-Efficient} performs quality control on the count of a labeled induced motif $H$ in $G$ compared to $\gnp$ in $1/p^{O(\Delta(H))}$ queries and time, using the algorithm \textsc{NonInduced-Motif-Quality-Efficient} and direct sampling.}
\vspace{1em} ~~ \textsc{Induced-Motif-Quality-Efficient}: On inputs $G, n, p, \gnp, \varepsilon$, and induced $H$
\begin{quote}
\begin{enumerate}
    \item If $p>p_0$ or $H$ is edgeless, sample independent uniform ordered $k$-tuples of distinct vertices to estimate the fraction inducing $H$ to additive accuracy $\varepsilon F_{\gnp}(H)/8$, with failure probability at most $\min\{\tau/2,1/6\}$. Accept if the estimate is in $(1\pm\varepsilon/2)F_{\gnp}(H)$, and reject otherwise.
    \item Otherwise, for $H'=H$ and for every $H'=H+a$ with $a\notin E(H)$:
    \begin{enumerate}
        \item Run Algorithm \textsc{NonInduced-Motif-Quality-Efficient} with inputs $G, n, p, \gnp, \frac{\varepsilon}{8}$, and non-induced $H'$. Use completeness error $\tau/(2(r+1))$. If Algorithm \textsc{NonInduced-Motif-Quality-Efficient} rejects, then reject. Else, proceed.
    \end{enumerate}
    \item Accept.
\end{enumerate}
\end{quote}
\label{alg:induced-motif-quality-efficient}
\end{algorithm}

The inclusion-exclusion formula is as follows. For a motif $P$, let $C_{P}(G)$ be the number of labeled induced copies of $P$ in $G$, and let $\widetilde{C}_{P}(G)$ be the number of labeled non-induced copies of $P$ in $G$.

\begin{observation}
    Consider any motif $H$ on $k$ vertices, and any graph $G$. Then
    $$C_H(G) = \sum_{i = 0}^{\binom{k}{2} - e(H)} (-1)^{i} \cdot \sum_{\substack{H' \supseteq H \\ |H'| = k \\ e(H') = e(H) + i}} \widetilde{C}_{H'}(G).$$
\end{observation}

We also observe that the number of labeled non-induced copies of every $H'$ used by the algorithm concentrates around $\binom{n}{k}k!p^{e(H')}$ in the stated range of $n$. This follows from the same second-moment argument as \Cref{lem:Gnp-motif-concentration}, omitting non-edge constraints. The choice of the constant in the concentration error can be made sufficiently small for the non-induced tests.

We now prove \Cref{thm:motifs-runtime} for induced motifs.

\begin{proof}[Proof of \Cref{thm:motifs-runtime}]
\textbf{Completeness:} Suppose $G\sim\gnp$. In the sparse case, each call to \textsc{NonInduced-Motif-Quality-Efficient} accepts with the required probability. A union bound over the $1+r$ calls gives acceptance with probability at least $1-\tau$, for sufficiently large $B$. In the direct-sampling case, \Cref{lem:Gnp-motif-concentration} gives a global induced count within relative error $\varepsilon/8$ of its expectation. The sampled estimate then lies within relative error $\varepsilon/4$ of $F_{\gnp}(H)$ and the algorithm accepts, outside the allocated failure events.

\textbf{Soundness:} Let $\mu_H = \mathbb{E}_{G' \sim \gnp}\left[ C_H(G')\right] = \binom{n}{k}k!p^{e(H)} (1 - p)^{\binom{k}{2} - e(H)}$ (induced count) and $\widetilde{\mu}_H = \mathbb{E}_{G' \sim \gnp}\left[ \widetilde{C}_H(G')\right] = \binom{n}{k}k!p^{e(H)}$ (noninduced count). Suppose that $G$ satisfies $C_H(G) \not \in (1 \pm \varepsilon) \mu_H$. 
Every non-induced copy of $H$ that is not induced contains an extra edge. Therefore,
$$\widetilde{C}_H(G)-\sum_{a\notin E(H)}\widetilde{C}_{H+a}(G)\leq C_H(G)\leq\widetilde{C}_H(G).$$
Suppose that every checked non-induced count is within relative error $\eta=\varepsilon/8$ of its expectation. Since $\widetilde{\mu}_{H+a}=p\widetilde{\mu}_H$, the bounds above give
$$[1-\eta-(1+\eta)rp]\widetilde{\mu}_H\leq C_H(G)\leq(1+\eta)\widetilde{\mu}_H.$$
For $p\leq p_0$, we have $rp\leq\varepsilon/8$ and $\mu_H=(1-p)^r\widetilde{\mu}_H$. Thus the lower bound is at least $(1-\varepsilon)\mu_H$, and
$$C_H(G)\leq\frac{1+\varepsilon/8}{1-\varepsilon/8}\mu_H\leq(1+\varepsilon)\mu_H.$$
Therefore, if $C_H(G)\notin(1\pm\varepsilon)\mu_H$, at least one checked non-induced count is inaccurate. Its quality control algorithm rejects with probability at least $2/3$, so the induced algorithm rejects with at least this probability.

In the direct-sampling case, an additive estimation error of at most $\varepsilon F_{\gnp}(H)/8$ separates a count outside $(1\pm\varepsilon)\mu_H$ from the acceptance interval. Thus, the algorithm rejects every inaccurate graph with high constant probability.

\textbf{Query complexity and runtime:} In the sparse case, there are $1+r$ calls, all on motifs of maximum degree at most $\Delta(H)+1$. Thus, the query complexity and runtime are $1/p^{O(\Delta(H))}$ for $\Delta(H)\geq1$. Since $\lceil\sigma(H+a)\rceil\leq\Delta(H)+3$, the stated condition on $n$ controls collisions in every call.

In the direct-sampling case, $F_{\gnp}(H)\geq p_0^{e(H)}2^{-r}$ when $p>p_0$, and $F_{\gnp}(H)\geq2^{-r}$ when $H$ is edgeless. Hoeffding's inequality gives $O_{H,\varepsilon,\tau}(1)$ sampled tuples and constant query complexity and runtime.
\end{proof}

\section{Quality control for general random graph models}\label{sec:any-random-graph-family}

In this section, we generalize our results to the setting of quality control of any motif $H$ as compared to any distribution $\mathcal{D}_n$ over graphs with a good concentration of motif counts. Particularly, we will consider distributions $\mathcal{D}_n$ over graphs such that, for $G \sim \mathcal{D}_n$, for all $P \subseteq H$ the number of labeled induced copies of $P$ is in the range $(1 \pm \frac{\varepsilon_1}{4}) E_{\mathcal{D}_n}(P)$ for some $E_{\mathcal{D}_n}(P)$ with high probability.

Before stating the main theorem we prove in this section (\Cref{thm:general-graph-motif-count-qc}), we set up some notation and terminology.

\paragraph{Notation and terminology}
In what follows, for distribution $\mathcal{D}_n$ over graphs and motif $P$ on $\ell$ vertices, define $E_{\mathcal{D}_n}(P), F_{\mathcal{D}_n}(P)$ as follows:
\begin{equation}\label{eq:expected-P-number-general}
    E_{\mathcal{D}_n}(P) = \mathbb{E}_{G \sim \mathcal{D}_n}\left[ \# \text{ labeled induced copies of } P \text{ in } G\right] = \binom{n}{\ell} \ell! \cdot F_{\mathcal{D}_n}(P).
\end{equation}
We will refer to $F_{\mathcal{D}_n}(P)$ as the \textit{density} of $P$ in $G \sim \mathcal{D}_n$, since it reflects the probability that a uniformly random ordered tuple of $\ell$ distinct vertices induces $P$ with its prescribed labels.

Sample $s \in \mathbb{N}$ vertices uniformly at random with repetition. Let
\begin{equation}\label{eq:adjustment-general}
    A_{s, \ell} = \binom{s}{\ell} \cdot \ell! \cdot \frac{1}{n^{\ell}}.
\end{equation}
Then, 
\begin{equation} \label{eq:s-D-P}
    S_{\mathcal{D}_n}(P, s) = E_{\mathcal{D}_n}(P) \cdot A_{s, \ell}
\end{equation}
is the expected number of labeled induced copies of $P$ in the subgraph induced by the multiset of vertices (as proven below in Claim \ref{claim:expectation-set-P-copies}). \\\\
We begin by defining the notion of the \textit{frequency increments} of a motif, which will appear in the query complexity upper bound achieved for $(\mathcal{D}_n, \rho_H)$-quality control.

We begin by defining the \textit{frequency increments} $r_{\ell}(H)$ of a motif $H$. 

\begin{definition}[Frequency increment]\label{def:r-j-H}
Consider a motif $H$ on $k$ vertices. For $1 \leq \ell \leq k$, let $\mathcal{H}_{\ell}$ consist of all (labeled, induced) subgraphs of $H$ on $\ell$ vertices. Set $r_1(H)=1$. For $2\leq\ell\leq k$, let the $\ell$th frequency increment $r_{\ell}(H)$ of $H$ be:
\begin{equation}\label{eq:r-j-H}
    r_{\ell}(H) = \min_{H_1 \in \mathcal{H}_{\ell}}\frac{F_{\mathcal{D}_n}(H_1)^2}{\max_{\substack{H_2 \subseteq H_1 \\ |H_2| = \ell-1}}  F_{\mathcal{D}_n}(H_2)^2}.
\end{equation}
This is the minimum squared ratio between the density of any subgraph $H_1$ of $H$ on $\ell$ vertices and the densest subgraph $H_2 \subset H_1$ of $H_1$ on $\ell - 1$ vertices.
\end{definition}

For a motif $H$ on $k$ vertices, we simply call the $k$th frequency increment of $H$ its frequency increment.

To understand the definition of the frequency increment, let us return briefly to the case of $\gnp$ (with $p \leq 1/2$). For any $P$ ($|P| = j$), the maximum $\max_{\substack{P' \subseteq P \\ |P| = j-1}}  F_{\mathcal{D}_n}(P')^2$ will be achieved by removing the highest-degree vertex in $P$. This means, for example, that $r_k(H)$ (for $H$ with $|H| = k$) satisfies $r_k(H) \geq p^{\Theta\left(\Delta(H)\right)} $, where $\Delta(H)$ is the maximum vertex degree in $H$.

Let $C_H(G)$ be the number of labeled induced copies of $H$ in the graph $G$. Let $\rho_H$ be defined as $\rho_H(G) := C_H(G)/\mathbb{E}_{G' \sim \mathcal{D}_n}\left[ C_H(G')\right]$.

\begin{theorem}\label{thm:general-graph-motif-count-qc}
    Consider a motif $H$ on $k\geq2$ vertices, let $\varepsilon_1=\varepsilon(2/3)^{k-2}$, and suppose $n\geq2k$ and all submotifs have positive expected counts. Consider any distribution $\mathcal{D}_n$ over graphs satisfying the following properties:
    \begin{enumerate}
        \item For $G \sim \mathcal{D}_n$, for all $\ell \in [k]$ and all $P\subseteq H$ on $\ell$ vertices, the number of labeled induced copies of $P$ is in the range $(1 \pm \frac{\varepsilon_1}{4}) E_{\mathcal{D}_n}(P)$ with probability $1 - o(1)$.
        \item For all $\ell \in \{2, 3, \dots, k\}$, $r_{\ell}(H) \leq r_{\ell - 1}(H)$.
    \end{enumerate}
    Then, there exists an algorithm for $(\mathcal{D}_n, \rho_H)$-quality control problem that uses $\widetilde{O}\left( r_k(H)^{-2} \right)$ adjacency matrix queries to the input graph.
\end{theorem}

Assuming the motif expectations and frequency increments are given, the algorithm has a runtime of $\widetilde{O}\left(r_k(H)^{-2k}\right)$. Additionally, the bounds with the dependence on $\varepsilon$ are $\widetilde{O}\left( r_k(H)^{-2} \varepsilon^{-4} \right)$ for the query complexity and $\widetilde{O}\left(r_k(H)^{-2k}\varepsilon^{-4k}\right)$ for the runtime.

\paragraph{Other random graph models:} We now apply \Cref{thm:general-graph-motif-count-qc} to other random graph models. For illustrative purposes, we focus on the case of $k$-cliques. The condition we need out of the random graph models we can study is the concentration of the counts of all $k'$-cliques within the model, for $k' \leq k$. The bound on the query complexity also requires an understanding of a lower-bound on $\mathbb{E}\left[ C_k(G)\right]/\mathbb{E}\left[ C_{k-1}(G)\right]$, where the expectations are taken with respect to the other models/distributions considered.

\Cref{thm:general-graph-motif-count-qc} implies the following results, when we specify the distribution $\mathcal{D}_n$ that we apply the result to.

\begin{corollary}[Implication of \Cref{thm:general-graph-motif-count-qc}]\label{cor:specific-other-models-in-section}
Consider $(\mathcal{D}_n, \rho_k)$-quality control, for $\rho_k(G)  = C_k(G) / \mathbb{E}_{G' \sim \mathcal{D}_n}\left[C_k(G') \right]$. Define parameters $n, p$ where $p = \omega\left(1/n^{2/(k-1)} \right)$ and $p \leq 1/2$. For the following settings of $\mathcal{D}_n$, the query complexity of $(\mathcal{D}_n, \rho_k)$-quality control is $1/p^{O(k)}$:

\begin{enumerate}
    \item $\mathcal{D}_n = \gnp$.
    \item $\mathcal{D}_n$ is a Stochastic Block Model (as defined in \Cref{def:SBM}) with a constant number of communities each of size $\Omega(n)$ and minimum in-community/between-community edge probability at least $p$.
    \item $\mathcal{D}_n$ is the uniform distribution over $d$-regular graphs on $n$ vertices, where $d=np$ is an integer and $nd$ is even.
\end{enumerate}
\end{corollary}

\begin{proof}
    Part (1) is already proven in \Cref{sec:query-complexity-cliques}. It follows because concentration of $k'$-clique counts for $k' \leq k$ holds for $p = \omega\left(1/n^{2/(k-1)} \right)$. Additionally, $r_k(H) = p^{2(k-1)}$, so the query complexity is $1/p^{O(k)}$.

    For Part (2), concentration of $k'$-clique counts for $k' \leq k$ can be proven similarly to the proof for a specific stochastic block model in \Cref{sec:appendix-sbm}, when $p=\omega(n^{-2/(k-1)})$, since every community has linear size. For each fixed assignment of clique vertices to communities, the overlap proof in \Cref{sec:appendix-sbm} gives a variance ratio bounded by $O_k(\sum_{i=2}^k n^{-i}p^{-\binom{i}{2}})=o(1)$. Summing over the constantly many assignments gives concentration. We can show that $\mathbb{E}_{G \sim \mathcal{D}_n}\left[C_k(G) \right] \geq p^{k-1} (n-k+1) \mathbb{E}_{G \sim \mathcal{D}_n}\left[C_{k-1}(G) \right]$:
    $$\mathbb{E}_{G \sim \mathcal{D}_n}\left[C_k(G) \right] = \mathbb{E}\left[\sum_{S : |S| = k - 1} \sum_{i \in [n], i \not \in S} \mathbbm{1}(S \equiv (k-1)\text{-clique}) \cdot \mathbbm{1}((i, j) \in e(G)  ~ \forall j \in S)\right].$$
    Since the multiplied indicator random variables in each term of the sum are independent, this is at least:
    $$\geq \mathbb{E}_{G \sim \mathcal{D}_n}\left[C_{k-1}(G) \right] \cdot p^{k-1} (n-k+1).$$
    The same argument at every level gives $r_{\ell}(H)\geq p^{2(\ell-1)}$. Thus, the minimum-increment version of the theorem gives query complexity $1/p^{O(k)}$.

    In Part (3), the concentration and computation of the expected value when $p = \Theta(1)$ follow from \cite{krivelevich2001random} and \cite{krivelevich2006pseudo} (Theorem 4.10). The edge count is deterministic. For clique sizes at least three and $p=o(1)$, the concentration and expected value computation come from \cite{DBLP:journals/dm/KimSV07}, Corollary 3.1. Particularly, in both cases of $p$, the expected value is around $\binom{n}{k} k! \cdot p^{\binom{k}{2}}$, and the same estimate at every level gives $\min_{\ell}r_{\ell}(H)\geq p^{O(k)}$. Therefore, the minimum-increment version gives query complexity $1/p^{O(k)}$.
\end{proof}

We remark that there are other well-studied random graph models for which our results may apply, such as random geometric graphs and inhomogeneous random graphs (i.e., Chung-Lu random graphs). See \cite{bachmann2018concentration, devroye2011high} for references about the concentration and expectation of $k$-clique counts in random geometric graphs, and \cite{DBLP:journals/siamdm/00010KS24} for inhomogeneous random graphs.

\paragraph{The algorithm:} We consider the following algorithm, which extends the algorithm for $k$-clique counts compared to $\gnp$. Let the sample size be $s_{*} = \widetilde{O}(r_k(H)^{-1} \varepsilon^{-2})$.

\begin{algorithm}[H]
\caption{Algorithm \textsc{Motif-Quality} performs quality control on the count of a labeled induced motif $H$ in $G$ compared to a distribution $\mathcal{D}_n$ that possesses concentration of the count of $H$. The query complexity is $\widetilde{O}(r_k(H)^{-2} \varepsilon^{-4})$ (see \Cref{def:r-j-H}).}
\vspace{1em} ~~ \textsc{Motif-Quality}: 
On inputs $G$, $n, \varepsilon$, $\mathcal{D}_n$, and $H$
\begin{quote}
\begin{enumerate}
    \item Sample $s_{*}$ vertices, each uniformly at random (with repetition). Let $S$ be the corresponding multiset of vertices.
    \item For each $H' \subseteq H$ on $\ell$ vertices, count the number of labeled induced copies of $H'$ from $S$ in $G$ (i.e., $C_{H'}(S)$ from \Cref{def:C-P-S}).
    \item If, for any $2\leq\ell<k$, there exists an $H'$, $|H'|=\ell$, such that $C_{H'}(S)$ is not in $(1\pm\varepsilon_{\ell-1})S_{\mathcal{D}_n}(H',s_*)$, or $C_H(S)$ is not in $(1\pm\varepsilon/2)S_{\mathcal{D}_n}(H,s_*)$, reject. 
    \item Otherwise, accept.
\end{enumerate}
\end{quote}
\label{alg:motif-quality}
\end{algorithm}

\paragraph{Parameters.} We will analyze concentration and inductive lemmas over general parameters, then later set the parameters to the following (which are extensions of $\alpha_{\ell}, \varepsilon_{\ell}, s_{\ell}$ from previous sections).
\begin{definition}\label{def:general-parameters}
    For $1\leq\ell\leq k-1$, define
    $$\varepsilon_{\ell}=\varepsilon\cdot\left(\frac{2}{3}\right)^{k-\ell-1}.$$
    Define $\alpha_1=\infty$, and for $2\leq\ell\leq k$ define
    $$\alpha_{\ell}=\min\left\{\frac{\alpha_{\ell-1}}{2},\frac{r_{\ell}(H)\cdot\varepsilon_1^2}{16384\cdot\ell^4}\right\}.$$
\end{definition}

We will specifically consider the following setting of $s_{\ell}$. We use the same notation as the sample multiset size from $k$-cliques because the $s_{\ell}$ in this setting is a generalization of the $s_{\ell}$ from the $k$-cliques case.

\begin{definition}\label{def:s-ell-general}
    Suppose that $H$ is a motif on $k$ vertices. Let $0<\tau\leq1/10$ be the completeness error, and define
    $$L=\ln\left(\frac{(16k)^{2k}}{F_{\mathcal{D}_n}(H)\cdot\varepsilon_1\tau}\right).$$
    Define $s_1=2k$, and for $2\leq\ell\leq k$ define
    $$s_{\ell}=\left\lceil\max\{2k,L/\alpha_{\ell}\}\right\rceil.$$
\end{definition}

\begin{definition}[Sample size ($s_{*}$)]\label{def:s-general-distribution}
    Suppose that $H$ is a motif on $k$ vertices. Define $s_*=s_k$.
\end{definition}

\begin{observation}\label{obs:s-general}
    The parameters satisfy $s_{\ell-1}\leq s_{\ell}$, $\alpha_{\ell-1}\geq2\alpha_{\ell}$, and $\alpha_{\ell}s_{\ell}\geq L$. Also $F_{\mathcal{D}_n}(P)\geq F_{\mathcal{D}_n}(H)$ for every $P\subseteq H$.
\end{observation}

\begin{definition}\label{def:C-P-S}
    For a multiset $S=(X_1,\dots,X_s)$ of size $s$, and a motif $P$ on $\ell$ vertices, let $C_P(S)$ be the number of ordered $\ell$-tuples of distinct sampled positions whose vertices are distinct and induce $P$ with its prescribed labels. Since $S$ is a multiset, a labeled copy on vertices $v_1,\dots,v_{\ell}$ is counted $\prod_{j=1}^{\ell}m_S(v_j)$ times, where $m_S(v)$ is the multiplicity of $v$ in $S$.
\end{definition}

\subsection{Technical Lemmas}\label{sec:technical-lemmas}

In this section, we prove general versions of the inductive lemma (\Cref{lem:inductive-intro}) and  concentration lemma (\Cref{lem:concentration-intro}) for general motif counts and general models of random graphs. 

We begin by defining \textit{exponentially robust quasirandomness} for general $\mathcal{D}_n$ and $P$. Exponentially robust quasirandomness captures the property that, among multisets of size $s$, the number of copies of $P$ is very well concentrated around its expected value.

Recall the definition of $C_P(S)$ from \Cref{def:C-P-S}, for motif $P$ and multiset $S$.

\begin{definition}[Exponentially robustly quasirandom]\label{def:exp-qr-motif}
    We say that a graph $G=([n],E)$ is $(\varepsilon,\alpha, s_0)$-exponentially robustly quasirandom with respect to distribution $\mathcal{D}_n$ and motif $P$ on $\ell$ vertices if for every $s \geq s_0$ we have that 
    $$\Pr_{X_1,\ldots,X_s \sim_{\mathrm{i.i.d.}} [n]}\left[ C_P(\{X_1,\ldots,X_s\}) \not\in (1\pm \varepsilon) S_{\mathcal{D}_n}(P, s)\right] \leq 4^{\ell+1} \cdot \exp(-\alpha s).$$
\end{definition}

We prove the following lemmas. Below, let $\alpha_{\ell}, \varepsilon_{\ell}, s_{\ell}$ be parameters from \Cref{def:general-parameters} and \Cref{def:s-ell-general}.

\begin{lemma}[Inductive Lemma]\label{lemma:inductive-general-P-and-D}
    Consider a distribution $\mathcal{D}_n$ over graphs, and a motif $P\subseteq H$ on $2\leq\ell<k$ vertices. Let $\mathcal{P}_{\ell - 1}$ be the set of (labeled, induced) subgraphs of $P$ on $\ell - 1$ vertices.
    
    Suppose that the graph $G$ is $(\varepsilon_{\ell - 1}, \alpha_{\ell-1}, s_{\ell})$-exponentially robustly quasirandom with respect to counts of each $P' \in \mathcal{P}_{\ell-1}$ compared to $\mathcal{D}_n$.
    
    Define the tester $T_{\ell}$ that samples $s_{\ell}$ (as specified in \Cref{def:s-ell-general}) random vertices of $G$ (with replacement) and accepts if and only if the count of $P$ in the induced subgraph on the sampled vertices is in $(1 \pm \varepsilon_{\ell-1}) S_{\mathcal{D}_n}(P, s_{\ell})$. (Recall that, as in \Cref{eq:s-D-P}, $S_{\mathcal{D}_n}(P, s_{\ell})$ is the expected number of copies of $P$ in a multiset of $s_{\ell}$ vertices in graphs chosen according to $\mathcal{D}_n$.)

    $T_{\ell}$ uses $\widetilde{O}\left( r_{\ell}(H)^{-2} \varepsilon^{-4} \right)$ adjacency matrix queries. With probability at least $1 - 4^{\ell + 1}\exp(-\alpha_{\ell} s_{\ell})$, $T_{\ell}$ accepts graphs that are $(\varepsilon_{\ell-1}, \alpha_{\ell}, s_{\ell})$-exponentially robustly quasirandom with respect to counts of $P$ compared to $\mathcal{D}_n$ and rejects graphs that are not 
    $(\varepsilon_{\ell}, \alpha_{\ell}, s_{\ell})$-exponentially robustly quasirandom with respect to counts of $P$ compared to $\mathcal{D}_n$.
\end{lemma}

The following lemma provides a concentration inequality regarding the probability that the count of $P$ in a random multiset $S$ is far from its expectation in terms of the \textit{assumed quasirandomness} on smaller sub-motifs and \textit{ratio between motif densities}.

\begin{lemma}[Concentration Lemma]\label{lem:concentration-general-P-and-D}

Let $P$ be a motif on $\ell\geq2$ vertices, let $n,s\geq2\ell$, and suppose $0<\gamma\leq1$. Suppose that graph $G$ is $(\gamma, \beta, s)$-exponentially robustly quasirandom with respect to distribution $\mathcal{D}_n$ and all motifs $P' \subseteq P$ on $\ell - 1$ vertices. Suppose the following two conditions hold:
\begin{enumerate}
    \item \textbf{Size of s:} $s$ satisfies $\exp(-\beta s) \leq \gamma\cdot F_{\mathcal{D}_n}(P)/(256\ell^3 8^{\ell})$.
    \item \textbf{Lower bound on frequency increment:} $r_{\ell}(P) \cdot \frac{\gamma^2}{8192 \cdot \ell^4} \geq \kappa$ for some $\kappa > 0$.
\end{enumerate}
 Then:
 $$\Pr_S\left( \left| C_{P}(S) - C_{P}(G) \cdot A_{s, \ell}\right| \geq \frac{\gamma}{4} S_{\mathcal{D}_n}(P, s) \right) \leq 2\ell\cdot4^{\ell}\exp(-\beta s) + 2 \exp\left(- \kappa s\right).$$
\end{lemma}

We show that this lemma implies the following corollary, which is applicable for a specific setting of parameters (given by the choice of $P$ and $\mathcal{D}_n$), as given in \Cref{def:general-parameters}.

\begin{corollary}\label{cor:concentration-lemma-with-parameters-general}
    Let $P\subseteq H$ be a motif on $\ell$ vertices and $\mathcal{D}_n$ be a distribution over graphs. Given $P$ and $\mathcal{D}_n$, define $s_{\ell}$ as \Cref{def:s-ell-general}. For $2\leq\ell\leq k$, consider the definitions of $\varepsilon_{\ell}$ and $\alpha_{\ell}$ from \Cref{def:general-parameters}.

    Suppose that graph $G$ is $(\varepsilon_{\ell - 1}, \alpha_{\ell - 1}, s_{\ell})$-exponentially robustly quasirandom with respect to distribution $\mathcal{D}_n$ and all motifs $P' \subseteq P$ on $\ell - 1$ vertices.
    
    Then:
    $$\Pr_S\left( \left| C_{P}(S) - C_{P}(G) \cdot A_{s_{\ell}, \ell}\right| \geq \frac{\varepsilon_{\ell - 1}}{4}S_{\mathcal{D}_n}(P, s_{\ell}) \right) \leq 4^{\ell + 1} \exp(-\alpha_{\ell} s_{\ell}).$$
    The same bound holds for every sample size $u\geq s_{\ell}$, with $u$ in place of $s_{\ell}$ in the sampled count and probability bound. It also holds with $\varepsilon_1$ in place of $\varepsilon_{\ell-1}$ under the corresponding quasirandomness assumption.
\end{corollary}

\subsubsection{Concentration Lemma}

In this section, we prove the general concentration lemma (\Cref{lem:concentration-general-P-and-D}).

\paragraph{Proof of \Cref{lem:concentration-general-P-and-D}}

We first prove the following two claims. Let $C_{P}(G)$ be the number of induced labeled copies of $P$ in the graph $G$. For a subgraph $S$ of $G$, define $C_P(S)$ similarly.

\begin{claim}\label{claim:expectation-set-P-copies}
    Let $S$ be a multiset generated by sampling $s$ vertices from the graph, with replacement. Then the expected number of labeled induced copies of $P$ in $S$ is $C_{P}(G) \cdot A_{s, \ell}$.
\end{claim}

\begin{proof}
    Suppose that $P$ is an $\ell$-vertex motif. Let us first write the number of labeled induced copies of $P$ in $S$ (denoted as $C_{P}(S)$, defined in \Cref{def:C-P-S}) as a sum of indicator random variables. For ordered set $T$ and ordered $P$, let $T \equiv P$ mean that the ordered set of vertices of $T$ is exactly that of $P$. Let $\mathcal{C}_{P}(G)$ be the set of induced labeled copies of $P$ in $G$. Observe that:
    $$C_{P}(S) = \sum_{P \in \mathcal{C}_{P}(G)}\sum_{T\in[s]^{\ell}_{\ne}} 1(T \equiv P).$$
    Here $[s]^{\ell}_{\ne}$ denotes ordered tuples of distinct sampled positions, and $T\equiv P$ refers to their sampled vertices. Fix such a tuple and some $P \in \mathcal{C}_{P}(G)$. The probability that $T \equiv P$ is $1/n^{\ell}$, since $T$ is created by sampling $\ell$ vertices, each time uniformly at random (with repetitions).
    
    By linearity of expectation, 
    $\mathbb{E}\left( C_{P}(S) \right)= C_{P}(G) \cdot \binom{s}{\ell} \ell! / n^{\ell} = C_{P}(G) \cdot A_{s, \ell}.$
\end{proof}

\begin{claim}\label{claim:bounded-differences-whp-motifs}
    Suppose that $P$ is a motif on $\ell$ vertices. Consider any two sampled sequences $S=(x_1,\dots,x_s)$ and $S'=(y_1,\dots,y_s)$ whose counts of every $(\ell-1)$-vertex submotif are at most $(1+\varepsilon)S_{\mathcal{D}_n}(P',s)$. Then,
    $$|C_P(S)-C_P(S')|\leq |\{i:x_i\ne y_i\}|\cdot(1+\varepsilon)\ell^2\max_{P'\subseteq P,\,|P'|=\ell-1}S_{\mathcal{D}_n}(P',s).$$
\end{claim}

\begin{proof}
    Let $D=\{i:x_i\ne y_i\}$. The only copies contributing to the difference use a position in $D$. For any such position, removing its role in $P$ leaves an embedding of one of the $\ell$ motifs $P-v$. Thus, the number of labeled copies using that position is at most $\ell\max_{P'}C_{P'}(S)$, and similarly for $S'$.

    The number of copies lost is at most $|D|\ell\max_{P'}C_{P'}(S)$, and the number gained is at most $|D|\ell\max_{P'}C_{P'}(S')$. Their absolute difference is at most the larger of these numbers. In particular, it is bounded by the stated expression, with the looser factor $\ell^2$.
\end{proof}

\begin{proof}[Proof of \Cref{lem:concentration-general-P-and-D}]
    Consider a motif $P$ on $\ell$ vertices with maximum degree $\Delta(P)$. Suppose that $G$ is $(\gamma,\beta, s)$-exponentially robustly quasirandom with respect to distribution $\mathcal{D}_n$ and all motifs $P' \subseteq P$ on $\ell - 1$ vertices. Let $X_1, \dots, X_s$ be the random variables corresponding to which vertices are selected to be in $S$ (i.e., $X_i$ is the $i$-th vertex added to $S$, with repetitions).  
    
    By \Cref{def:exp-qr-motif} of exponentially robust quasirandomness for $\mathcal{D}_n$ and all motifs $P' \subseteq P$, $|P'| = \ell - 1$, for every $s$ and $P' \subseteq P$ with $|P'| = \ell - 1$, we have that $$\Pr\left[ C_{P'} (X_1, X_2, \dots, X_s) \not\in (1\pm \gamma) S_{\mathcal{D}_n}(P', s)\right] \leq 4^{\ell} \exp(-\beta s).$$ 

    We will apply McDiarmid's Inequality when differences are bounded with high probability (\Cref{cor:mcdiarmid-unconditional-expectation}) to the function $C_{P} : [n]^s \to \mathbb{R}$, which is the number of $P$ copies in the multiset.
    
    Let us set up the relevant definitions for the inequality. First, let the domain be $[n]^s$, and define $\mathcal{Y} \subseteq [n]^s$ to be the multisets of the domain such that, for each $P' \subseteq P$ on $\ell - 1$ vertices, the number of copies of $P'$ in $\mathcal{Y}$  is in  $(1\pm \gamma) S_{\mathcal{D}_n}(P', s)$. Then, observe that exponential-quasirandomness guarantees that $q := 1 - \mathbb{P}\left( (X_1, X_2, \dots, X_s) \in \mathcal{Y}\right) \leq \ell \cdot 4^{\ell} \exp(-\beta s)$.

    We argue about bounded differences of $C_{P}$ on $\mathcal{Y}$. For each $i \in [s]$, define $$c_i = (1 + \gamma) \ell^2 \cdot \max_{\substack{P' \subseteq P \\ |P'| = \ell-1}}  S_{\mathcal{D}_n}(P', s).$$ By Claim \ref{claim:bounded-differences-whp-motifs}, for all $S = (x_1, x_2, \dots, x_s) \in \mathcal{Y}$ and $S' = (y_1, y_2, \dots, y_s) \in \mathcal{Y}$, $$\left|C_{P}(S) - C_{P}(S')\right| \leq  |\{i:x_i\ne y_i\}| \cdot (1 + \gamma)\ell^2 \cdot \max_{\substack{P' \subseteq P \\ |P'| = \ell-1}}  S_{\mathcal{D}_n}(P', s) = \sum_{i : x_i \neq y_i} c_i,$$
    so the bounded difference property holds with this setting of $c_i$'s.

    We want to apply \Cref{cor:mcdiarmid-unconditional-expectation} with error bound $\delta = \gamma S_{\mathcal{D}_n}(P, s)/4$. To do so, we need to verify that $$q \leq \min\left\{\frac{\delta}{2 \max_S(C_P(S))}, \frac{\delta}{4 \mathbb{E}(C_P(X_1, X_2, \dots, X_s))} , \frac{1}{2}\right\}.$$ 
    Using $\delta = \gamma S_{\mathcal{D}_n}(P, s)/4$ and $S_{\mathcal{D}_n}(P, s) = F_{\mathcal{D}_n}(P) \cdot \binom{n}{\ell} \binom{s}{\ell} \cdot \frac{(\ell!)^2}{n^{\ell}}$ from \Cref{eq:expected-P-number-general} and \Cref{eq:s-D-P}, and since $\max_S(C_P(S)) \leq \binom{s}{\ell} \ell!$ and $\mathbb{E}(C_P(X_1, X_2, \dots, X_s)) \leq \binom{s}{\ell} \ell!$, it suffices to show that:
    $$q \leq \min\left\{ \frac{\gamma S_{\mathcal{D}_n}(P, s)}{8 \binom{s}{\ell} \ell!}, \frac{\gamma S_{\mathcal{D}_n}(P, s)}{16 \binom{s}{\ell} \ell!}, \frac{1}{2}\right\} = \min\left\{ \frac{\gamma}{16} \cdot F_{\mathcal{D}_n}(P) \cdot \binom{n}{\ell} \ell! \cdot \frac{1}{n^{\ell}}, \frac{1}{2} \right\}.$$
    Since $q \leq \ell \cdot 4^{\ell} \cdot \exp\left(- \beta s\right) \leq \gamma \cdot F_{\mathcal{D}_n}(P) /(16 \cdot 2^{\ell})$ by assumption, this inequality holds.

    We are ready to apply \Cref{cor:mcdiarmid-unconditional-expectation}, which implies that:
    $$\mathbb{P}\left( \left|C_{P}(X_1, X_2, \dots, X_s) -  \mathbb{E}\left( C_{P}(X_1, X_2, \dots, X_s) \right)\right| \geq \gamma S_{\mathcal{D}_n}(P, s) /4 \right)$$
    \begin{equation}\label{eq:mcdiarmid-application-motif}
       \leq 2q + 2 \exp\left( -\frac{2\max\left( 0, \gamma S_{\mathcal{D}_n}(P, s)/8 - q \sum_{i = 1}^s c_i\right)^2}{\sum_{i = 1}^s c_i^2}\right). 
    \end{equation}
    Let us analyze the second term. First, $$q \sum_{i = 1}^s c_i \leq \ell \cdot 4^{\ell} \cdot \exp(-\beta s) \cdot s \cdot (1 + \gamma) \ell^2 \cdot \max_{\substack{P' \subseteq P \\ |P'| = \ell-1}}  S_{\mathcal{D}_n}(P', s).$$ 

    We prove that this is at most $\gamma S_{\mathcal{D}_n}(P,s)/16$. First, the exact ratio of the expectations is
    $$\frac{S_{\mathcal{D}_n}(P,s)}{\max_{P'}S_{\mathcal{D}_n}(P',s)}=\frac{(n-\ell+1)(s-\ell+1)}{n}\cdot\frac{F_{\mathcal{D}_n}(P)}{\max_{P'}F_{\mathcal{D}_n}(P')}\geq\frac{s}{4}F_{\mathcal{D}_n}(P),$$
    where the maxima are over $(\ell-1)$-vertex submotifs of $P$, and we used $F_{\mathcal{D}_n}(P')\leq1$.

    The size assumption gives $q\leq\gamma F_{\mathcal{D}_n}(P)/(128\ell^2)$. Since $1+\gamma\leq2$, we obtain
    $$q\sum_{i=1}^s c_i\leq2q s\ell^2\max_{P'}S_{\mathcal{D}_n}(P',s)\leq\frac{\gamma}{16}S_{\mathcal{D}_n}(P,s).$$
    Therefore, $2\max(0,\gamma S_{\mathcal{D}_n}(P,s)/8-q\sum_i c_i)^2\geq\gamma^2S_{\mathcal{D}_n}(P,s)^2/128$.

    This implies that (also rounding $(1 + \gamma) \max_{P'} S_{\mathcal{D}_n}(P', s)$ up to $2 \max_{P'} S_{\mathcal{D}_n}(P', s)$ in the denominator of the expression) \Cref{eq:mcdiarmid-application-motif} is upper-bounded by:
    $$
        \leq 2q + 2 \exp\left(- \frac{ \gamma^2 S_{\mathcal{D}_n}(P, s)^2}{128 \cdot s \cdot (2 \ell^2)^2 \cdot \max_{\substack{P' \subseteq P \\ |P'| = \ell-1}}  S_{\mathcal{D}_n}(P', s)^2}\right).
    $$
    By definition of $S_{\mathcal{D}_n}(P', s)$ and $S_{\mathcal{D}_n}(P, s)$ (combining \Cref{eq:expected-P-number-general} and \Cref{eq:s-D-P}), the expression above equals:
    $$
    = 2q + 2  \exp\left(- \frac{ \gamma^2}{512 \cdot \ell^4 \cdot s} \cdot \left( \frac{(n - \ell + 1)(s - \ell + 1)}{n} \cdot \frac{F_{\mathcal{D}_n}(P)}{\max_{\substack{P' \subseteq P \\ |P'| = \ell-1}}  F_{\mathcal{D}_n}(P')}\right)^2\right).
    $$
    For large enough $n$, and by definition of $s$ (\Cref{def:s-general-distribution}), we can use that $(n - \ell + 1)/n \geq 1/2$ and $(s - \ell + 1) \geq s/2$ to obtain that the expression is bounded above by:
    $$\leq 2q + 2 \exp\left(- \frac{\gamma^2}{512 \cdot 16 \cdot \ell^4} \cdot s \cdot \frac{F_{\mathcal{D}_n}(P)^2}{\max_{\substack{P' \subseteq P \\ |P'| = \ell-1}}  F_{\mathcal{D}_n}(P')^2} \right).$$
    By definition of $q$ and $r_{\ell}(P)$ this is:
    $$\leq 2\ell\cdot4^{\ell}\exp(-\beta s) + 2 \exp\left(- \frac{\gamma^2}{8192 \cdot \ell^4} \cdot s \cdot r_{\ell}(P) \right).$$ 
    By assumption (2) of the concentration lemma we are proving, this is $$
    \leq  2\ell\cdot4^{\ell}\exp(-\beta s) + 2 \exp\left(- \kappa s\right). \eqno \qedhere$$
    \end{proof}

We now prove \Cref{cor:concentration-lemma-with-parameters-general} given the concentration lemma proven above.

\begin{proof}[Proof of \Cref{cor:concentration-lemma-with-parameters-general}]
We apply \Cref{lem:concentration-general-P-and-D} with $\gamma=\varepsilon_{\ell-1}$, $\beta=\alpha_{\ell-1}$, and sample size $u\geq s_{\ell}$. The definition of $L$ verifies the size assumption, since $F_{\mathcal{D}_n}(P)\geq F_{\mathcal{D}_n}(H)$.

By definition, $r_{\ell}(P)\geq r_{\ell}(H)$ and $\alpha_{\ell-1}\geq2\alpha_{\ell}$. Set $\kappa=r_{\ell}(P)\gamma^2/(8192\ell^4)\geq2\alpha_{\ell}$. Therefore,
$$2\ell4^{\ell}e^{-\alpha_{\ell-1}u}+2e^{-\kappa u}\leq(2\ell4^{\ell}+2)e^{-2\alpha_{\ell}u}\leq4^{\ell+1}e^{-\alpha_{\ell}u},$$
where the last inequality uses $\alpha_{\ell}u\geq L$. The same argument applies with $\gamma=\varepsilon_1$.
\end{proof}

\subsubsection{Inductive Lemma}

We now prove the general inductive lemma (\Cref{lemma:inductive-general-P-and-D}).

\paragraph{Proof of \Cref{lemma:inductive-general-P-and-D} given \Cref{lem:concentration-general-P-and-D}}\label{section:inductive-given-concentration-general}

Towards proving the inductive lemma (\Cref{lemma:inductive-general-P-and-D}) from the concentration lemma (\Cref{lem:concentration-general-P-and-D}), we first present some corollaries of \Cref{lem:concentration-general-P-and-D}.

The following corollary proves that if $G$ is $(\varepsilon, \alpha_{\ell-1}, s_{\ell})$-exponentially robustly quasirandom with respect to counts of every $P' \in \mathcal{P}_{\ell - 1}$ compared to $\mathcal{D}_n$,  then a size-$s_{\ell}$ multiset chosen uniformly at random will, with very high probability, reveal what the global number of $P$ copies in $G$ is and whether $G$ is $(\varepsilon, \alpha_{\ell}, s_{\ell})$-exponentially robustly quasirandom with respect to counts of $P$ counts (compared to $\mathcal{D}_n$).

\begin{corollary}[Corollary of \Cref{cor:concentration-lemma-with-parameters-general}]\label{cor:cor-of-concentration-general}
    Consider $2 \leq \ell \leq k-1$. Consider a distribution $\mathcal{D}_n$ over graphs, and a motif $P$ on $\ell$ vertices, with $P\subseteq H$. Let $\mathcal{P}_{\ell - 1}$ be the set of (labeled, induced) subgraphs of $P$ on $\ell - 1$ vertices.
    
    If a graph $G$ is $(\varepsilon_{\ell - 1}, \alpha_{\ell-1}, s_{\ell})$-exponentially robustly quasirandom with respect to counts of each $P' \in \mathcal{P}_{\ell-1}$ compared to $\mathcal{D}_n$, then the following hold:
    \begin{enumerate}
        \item Suppose that $G$ is $(\varepsilon_{\ell-1}, \alpha_{\ell}, s_{\ell})$-exponentially robustly quasirandom with respect to the count of $P$. Sample a multiset $S$ by selecting $s_{\ell}$ vertices uniformly at random (with repetition). Then with probability at least $1 - 4^{\ell + 1} \exp(-\alpha_{\ell} s_{\ell})$ (over the choice of $S$), $C_{P}(S) \in (1 \pm \varepsilon_{\ell - 1}) S_{\mathcal{D}_n}(P, s_{\ell})$.
        \item Sample a multiset $S$ by selecting vertices uniformly at random (with repetition). If $C_{P}(S) \in (1 \pm \varepsilon_{\ell - 1}) S_{\mathcal{D}_n}(P, s_{\ell})$, then, except on an event of probability at most $4^{\ell+1}\exp(-\alpha_{\ell}s_{\ell})$ (over the choice of $S$), $C_{P}(G) \in (1 \pm \frac{5}{4} \varepsilon_{\ell - 1})E_{\mathcal{D}_n}(P)$ and, furthermore, $G$ is $(\varepsilon_{\ell}, \alpha_{\ell}, s_{\ell})$-exponentially robustly quasirandom with respect to counts of $P$ compared to $\mathcal{D}_n$.
    \end{enumerate}
\end{corollary}

\begin{proof}[Proof of \Cref{cor:cor-of-concentration-general}]
    Suppose that $G$ is $(\varepsilon_{\ell - 1}, \alpha_{\ell-1}, s_{\ell})$-exponentially robustly quasirandom with respect to counts of each $P' \in \mathcal{P}_{\ell-1}$ compared to $\mathcal{D}_n$. Then, by \Cref{cor:concentration-lemma-with-parameters-general} of \Cref{lem:concentration-general-P-and-D}, $$\Pr_S\left( \left| C_{P}(S) - C_{P}(G) \cdot A_{s_{\ell}, \ell}\right| \geq \frac{\varepsilon_{\ell - 1}}{4}S_{\mathcal{D}_n}(P, s_{\ell}) \right) \leq 4^{\ell + 1} \exp(-\alpha_{\ell} s_{\ell}).$$

     First, let's address case (1). Suppose $G$ is $(\varepsilon_{\ell-1}, \alpha_{\ell}, s_{\ell})$-exponentially robustly quasirandom with respect to the count of $P$ compared to $\mathcal{D}_n$. By definition of exponentially robust quasirandomness (\Cref{def:exp-qr-motif}), all but a $4^{\ell + 1}\exp(-\alpha_{\ell} s_{\ell})$ fraction of size-$s$ multisets $S$ satisfy
    $$C_{P}(S) \in \left(1 \pm \varepsilon_{\ell-1} \right)S_{\mathcal{D}_n}(P, s).$$

    Let's now move to case (2). \Cref{cor:concentration-lemma-with-parameters-general} implies that when $G$ is $(\varepsilon_{\ell - 1}, \alpha_{\ell-1}, s_{\ell})$-exponentially robustly quasirandom with respect to counts of all $P' \in \mathcal{P}_{\ell - 1}$ compared to $\mathcal{D}_n$, with probability at least $1 - 4^{\ell + 1}\exp\left(-\alpha_{\ell} s_{\ell} \right)$, $$C_{P} (G) \cdot A_{s_{\ell}, \ell} \in C_{P}(S) \pm \frac{\varepsilon_{\ell-1}}{4}S_{\mathcal{D}_n}(P, s_{\ell}).$$
    Equivalently, using that $S_{\mathcal{D}_n}(P, s_{\ell}) = E_{\mathcal{D}_n}(P) \cdot \binom{s_{\ell}}{\ell} \frac{\ell!}{n^{\ell}}$ from \Cref{eq:s-D-P}, with probability at least $1 - 4^{\ell + 1}\exp\left(-\alpha_{\ell} s_{\ell} \right)$,
    \begin{equation}\label{eq:ell-clique-G-given-ell-clique-S-sample}
        C_{P} (G) \in C_{P}(S) \cdot \frac{1}{A_{s_{\ell}, \ell}} \pm \frac{\varepsilon_{\ell - 1}}{4}E_{\mathcal{D}_n}(P).
    \end{equation}
    Additionally, by \Cref{lem:concentration-general-P-and-D}, under the high-probability event, 
    \begin{equation}\label{eq:ell-clique-S-sample-given-ell-clique-G}
        C_{P}(S) \in C_{P}(G) \cdot A_{s_{\ell}, \ell} \pm \frac{\varepsilon_{\ell - 1}}{4} S_{\mathcal{D}_n}(P, s_{\ell}).
    \end{equation}
    
    Suppose that $C_{P}(S) \in (1 \pm \varepsilon_{\ell - 1}) S_{\mathcal{D}_n}(P, s_{\ell})$. Therefore, $$C_{P}(S) \cdot \frac{1}{A_{s_{\ell}, \ell}} \in (1 \pm \varepsilon_{\ell - 1})E_{\mathcal{D}_n}(P).$$ By \Cref{eq:ell-clique-G-given-ell-clique-S-sample}, with probability at least $1 - 4^{\ell + 1}\exp\left(-\alpha_{\ell} s_{\ell} \right)$, $C_{P}(G) \in (1 \pm \frac{5}{4}\varepsilon_{\ell - 1})E_{\mathcal{D}_n}(P)$. We now argue that the graph is exponentially robustly quasirandom. Fix any $u\geq s_{\ell}$ and apply concentration at sample size $u$. When $C_{P}(G) \in (1 \pm \frac{5}{4}\varepsilon_{\ell - 1})E_{\mathcal{D}_n}(P)$, by \Cref{eq:ell-clique-S-sample-given-ell-clique-G}, all but a $4^{\ell + 1}\exp(-\alpha_{\ell} s_{\ell})$ fraction of size-$s_{\ell}$ multisets $S'$ satisfy
    $$C_{P}(S') \in \left(1 \pm \frac{3}{2}\varepsilon_{\ell - 1}\right)S_{\mathcal{D}_n}(P, s_{\ell}).$$
    The same displayed bound holds with $u$ in place of $s_{\ell}$, and failure probability $4^{\ell+1}e^{-\alpha_{\ell}u}$, for every $u\geq s_{\ell}$. Since $\varepsilon_{\ell} = \frac{3}{2}\varepsilon_{\ell - 1}$, therefore $G$ is $(\varepsilon_{\ell}, \alpha_{\ell}, s_{\ell})$-exponentially robustly quasirandom with respect to the count of $P$ compared to $\mathcal{D}_n$.
\end{proof}

\begin{proof}[Proof of \Cref{lemma:inductive-general-P-and-D}]

Given \Cref{cor:cor-of-concentration-general}, \Cref{lemma:inductive-general-P-and-D} follows from a simple argument, given below.

 Suppose that $G$ is $(\varepsilon_{\ell - 1}, \alpha_{\ell-1}, s_{\ell})$-exponentially robustly quasirandom with respect to the counts of all motifs $P' \in \mathcal{P}_{\ell - 1}$. Define the tester $T_{\ell}$ that samples $s_{\ell}$ (as specified in \Cref{def:s-ell-general}) random vertices of $G$ (with replacement) and accepts if and only if the motif count in the induced subgraph on the sampled vertices is within $(1 \pm \varepsilon_{\ell - 1}) S_{\mathcal{D}_n}(P, s_{\ell})$.

 By the corollary \Cref{cor:cor-of-concentration-general} of the concentration lemma (\Cref{lem:concentration-general-P-and-D}), the tester $T_{\ell}$ satisfies the following properties. 

\textbf{Completeness:} Suppose that $G$ is $(\varepsilon_{\ell-1}, \alpha_{\ell}, s_{\ell})$-exponentially robustly quasirandom with respect to counts of $P$ compared to $\mathcal{D}_n$. Then, with probability at least $1 - 4^{\ell + 1}\exp(-\alpha_{\ell} s_{\ell})$, $C_{P}(S) \in (1 \pm \varepsilon_{\ell - 1}) S_{\mathcal{D}_n}(P, s_{\ell})$.
 
\textbf{Soundness:} If the motif count on the sampled multiset is in $(1 \pm \varepsilon_{\ell - 1}) S_{\mathcal{D}_n}(P, s_{\ell})$, then, except on an event of probability at most $4^{\ell+1}e^{-\alpha_{\ell}s_{\ell}}$, $G$ is $(\varepsilon_{\ell}, \alpha_{\ell}, s_{\ell})$-exponentially robustly quasirandom with high probability, so accepted graphs satisfy the quasirandomness condition.
\end{proof}

The inductive tester may use any sample size $u\geq s_{\ell}$ with $O(u^2)$ queries. Its probability bound then has $u$ in place of $s_{\ell}$, while the quasirandomness conclusions keep the threshold $s_{\ell}$. All the soundness statements above bound the joint event of acceptance and a false conclusion, without conditioning on acceptance.

\subsection{Proof of Theorem \ref{thm:general-graph-motif-count-qc}}

We now prove \Cref{thm:general-graph-motif-count-qc} from the general concentration and inductive lemmas proven above, following the framework of the proof of \Cref{thm:k-cliques} from the concentration and inductive lemmas for $k$-cliques.

We begin with the following corollary that extends \Cref{cor:another-cor-of-concentration}.
\begin{corollary}[Corollary of \Cref{lem:concentration-general-P-and-D}]\label{cor:another-cor-of-concentration-general}
    Consider a distribution $\mathcal{D}_n$ over graphs, and a motif $H$ on $k$ vertices. Let $\mathcal{H}_{k-1}$ be the set of labeled induced subgraphs of $H$ on $k-1$ vertices.    
    
    Suppose $G$ is $(\varepsilon, \alpha_{k-1}, s_{*})$-exponentially robustly quasirandom with respect to counts of each $H' \in \mathcal{H}_{k-1}$ compared to $\mathcal{D}_n$. Sample a multiset $S$ by selecting $s_{*}$ vertices uniformly at random (with repetition). If $C_H(G) \in (1 \pm \frac{1}{4} \varepsilon) E_{\mathcal{D}_n}(H)$, then with probability at least $1 - 4^{k + 1} \exp(-\alpha_k s_{*})$, $C_H(S) \in (1 \pm \frac{1}{2} \varepsilon) S_{\mathcal{D}_n}(H, s_{*})$. If $C_H(G) \not \in (1 \pm \varepsilon) E_{\mathcal{D}_n}(H)$, then with probability at least $1 - 4^{k + 1} \exp(-\alpha_k s_{*})$, $C_H(S) \not \in (1 \pm \frac{3}{4} \varepsilon) S_{\mathcal{D}_n}(H, s_{*})$.
\end{corollary}

\begin{proof}
    The proof follows similarly to that of \Cref{cor:cor-of-concentration-general}.
\end{proof}

Next, we show that Claim \ref{claim:completeness-concentration} from in \Cref{sec:query-complexity-cliques} extends rather directly to the setting of motifs and general graph distributions. Interestingly, the fact that the proof highlights that all we needed about $\gnp$ previously was the concentration of $\ell$-clique counts for all $\ell \in [k]$.

\begin{claim}\label{claim:completeness-concentration-general}
       Consider a distribution $\mathcal{D}_n$ and motif $H$ as in \Cref{thm:general-graph-motif-count-qc}. Suppose that, for every $P\subseteq H$, $C_P(G)\in(1\pm\varepsilon_1/4)E_{\mathcal{D}_n}(P)$. Then, for every $P\subseteq H$ on $\ell$ vertices, $G$ is $(\varepsilon_1,\alpha_{\ell},s_{\ell})$-exponentially robustly quasirandom with respect to the count of $P$.
\end{claim}

\begin{proof}
    The proof is a direct extension of Claim \ref{claim:completeness-concentration}, by induction on $\ell$. The vertex count is deterministic. For the inductive step, apply \Cref{cor:concentration-lemma-with-parameters-general} with $\gamma=\varepsilon_1$ at every sample size $u\geq s_{\ell}$. Outside an event of probability $4^{\ell+1}e^{-\alpha_{\ell}u}$, the sampled count differs from $E_{\mathcal{D}_n}(P)A_{u,\ell}$ by relative error at most $\varepsilon_1/4+\varepsilon_1/4=\varepsilon_1/2$. This proves the claim.
\end{proof}

We are ready to prove \Cref{thm:general-graph-motif-count-qc}, along the lines of the proof of \Cref{thm:k-cliques} in \Cref{sec:qc-upperbound}.

\begin{proof}[Proof of \Cref{thm:general-graph-motif-count-qc}]
    \textbf{Completeness.} By assumption, with high probability $G\sim\mathcal{D}_n$ has all motif counts in the ranges required by Claim \ref{claim:completeness-concentration-general}. On this event, the claim and its proof imply that each sampled count is within relative error $\varepsilon_1/2$ of $S_{\mathcal{D}_n}(P,s_*)$, except with probability $4^{|P|+1}e^{-\alpha_{|P|}s_*}$. A union bound over at most $2^k$ submotifs bounds the total error by $\tau$. These errors are smaller than the acceptance thresholds, so the algorithm accepts $G\sim\mathcal{D}_n$ with high probability.

    \textbf{Soundness.} 
    If there exists some $\ell \in \{2, 3, \dots, k-1\}$ and $H' \subseteq H$ on $\ell$ vertices such that the input graph $G$ is $(\varepsilon_{\ell'}, \alpha_{\ell'}, s_{\ell'})$-exponentially robustly quasirandom with respect to the count of $H'' \subseteq H, |H''| = \ell'$ for all $\ell' < \ell$ and is \textit{not} $(\varepsilon_{\ell}, \alpha_{\ell}, s_{\ell})$-exponentially robustly quasirandom with respect to the count of $H'$, then \Cref{lemma:inductive-general-P-and-D} at sample size $s_*\geq s_{\ell}$ implies that with probability at least $1 - 4^{\ell + 1}\exp(-\alpha_{\ell}s_{\ell}) \geq 9/10$, the algorithm rejects on $G$.

    If, instead, the input graph is $(\varepsilon_{k-1} = \varepsilon, \alpha_{k-1}, s_k = s_{*})$-exponentially robustly quasirandom with respect to the count of all $H' \subseteq H, |H'| = k - 1$. In this case, \Cref{cor:another-cor-of-concentration-general} implies that with probability at least $1 - 4^{k + 1}\exp(-\alpha_k s_{*})$, $C_H(S) \not \in (1 \pm \frac{3}{4} \varepsilon) S_{\mathcal{D}_n}(H, s_{*})$, in which case the algorithm rejects.
    
    Thus, the algorithm rejects each $G$ with $C_H(G) \not \in (1 \pm \varepsilon) E_{\mathcal{D}_n}(H)$ with high probability.
    
    \textbf{Query complexity:} The adjacency matrix query complexity is $\binom{s_{*}}{2} = \widetilde{O}\left(r_k(H)^{-2} \varepsilon^{-4}\right)$.

    \textbf{Runtime:} The algorithm counts the number of labeled induced copies of $H'$, for all $H' \subseteq H$. For worst-case input graphs and arbitrary distributions $\mathcal{D}_n$ over graphs, this can take up to time $O(s_{*}^k) \leq \widetilde{O}\left(r_k(H)^{-2k}\varepsilon^{-4k}\right)$.
\end{proof}

The same proofs apply to labeled non-induced counts, using their corresponding expectations and frequencies. The same query bound holds with $\min_{2\leq\ell\leq k}r_{\ell}(H)$ in place of $r_k(H)$, without assuming that the increments are nonincreasing.
\section{Acknowledgments}

We thank an anonymous reviewer for significant comments, references, and suggestions on the presentation in the introduction of the paper. (See \Cref{fn:reviewer-refs-1}.)

\bibliographystyle{alpha}
\bibliography{bibliography}

\appendix
\section{Quality control of triangle counts}\label{sec:triangles}

We prove that, under the adjacency list query access model (as defined in \Cref{sec:preliminaries}), quality control of triangle counts as compared to $\gnp$ requires $\Theta\left( \frac{1}{p}\right)$ queries, up to $\mathrm{polylog}(n)$ and $\mathrm{poly}(1/\varepsilon)$ multiplicative factors.

Recall that, for a graph $G$, $C_3(G)$ is the number of triangles in the graph.

\begin{theorem}\label{thm:triangles-usingliterature}
    Let $\rho_{\Delta}= C_3(G)/\mathbb{E}_{G' \sim \gnp}\left[C_3(G') \right]$ and consider parameters $n, p$ such that $p = \omega\left(1/n\right)$. Solving the $(\gnp, \rho_{\Delta})$-quality control problem can be accomplished in $O\left(\frac{1}{p} \right) \cdot \mathrm{poly}\left(\log(n) \right)$ adjacency matrix, adjacency list, and vertex degree queries and $O\left(\frac{1}{p} \right) \cdot \mathrm{poly}\left(\log(n) \right)$ runtime, in expectation (over the random choices of the algorithm). Moreover, $(\gnp, \rho_{\Delta})$-quality control requires $\Omega\left( \frac{1}{p}\right)$ queries.
\end{theorem}

Studying triangles is interesting because, beyond being one of the most elementary structures possible, the algorithm we construct will directly rely on results from the literature \cite{eden2018faster, eden2022approximating}. Studying triangles will, therefore, set us up for understanding the limitations we would face in applying the approaches from the literature to the context of quality control for larger motifs.

\subsection{Query complexity upper bound}

In this section, we prove the upper bound in \Cref{thm:triangles-usingliterature}, restated below.

\begin{lemma}[Upper bound in \Cref{thm:triangles-usingliterature}]\label{lem:triangles-fromliterature-UB}
     Consider parameters $n, p, \varepsilon$ such that $p = \omega\left(\varepsilon^{-2} \cdot 1/n\right)$.  There is a $(\gnp, \rho_{\Delta}, \varepsilon)$-quality control algorithm using $O\left(\frac{1}{p} \right) \cdot \mathrm{poly}\left(\log(n), 1/\varepsilon \right)$ adjacency matrix, adjacency list, and vertex degree queries and $O\left(\frac{1}{p} \right) \cdot \mathrm{poly}\left(\log(n), 1/\varepsilon \right)$ runtime, in expectation (over the random choices of the algorithm).
\end{lemma}

To prove this, we rely on and combine algorithms from the literature \cite{eden2022approximating, eden2018faster} to obtain a quality control algorithm for triangle counts that uses $O(1/p) \cdot \text{poly}(\log n, \frac{1}{\varepsilon})$ queries. At a high level, the algorithm does the following. The algorithm will first approximate the arboricity; if the approximation is less than $2 n p$, proceed, and else reject. Next, approximate the number of triangles using the sublinear approximation algorithm of \cite{eden2018faster} whose query complexity depends on the arboricity of the input graph. Accept if the approximation is in the range $(1 \pm \varepsilon/2) \binom{n}{3} p^3$; else, reject. The completeness (accepting $G \sim \gnp$ with high probability) relies on the concentration of the arboricity and number of triangles for \erdosrenyi graphs (see \Cref{prop:Gnp-arboricity-concentration} and \Cref{lem:Gnp-clique-concentration}). The soundness follows because we reject all graphs with $\rho_{\Delta}\not \in (1 \pm \varepsilon) \mathbb{E}_{G' \sim \gnp}\left[ \rho_{\Delta}(G')\right]$ (and possibly other graphs, due to the nature of the quality control framework).

Towards this general goal, we present subroutines from \cite{eden2022approximating, eden2018faster} that provide ``decision versions'' of the problem of estimating the arboricity and triangle count. That is, the following lemmas state that it is possible to detect when the arboricity and triangle count are in a certain range, versus sufficiently far from the range.

\begin{lemma}[Claim 3.23 of \cite{eden2022approximating}]\label{lem:claim3.23-eden2022}
    Consider an input graph $G$ with arboricity $\mathrm{arb}(G)$, and a parameter $\alpha$. There exists an algorithm (call it \textbf{Algorithm A}) that performs $O(n/\alpha)$ adjacency list and degree queries and has runtime $O(n/\alpha)$ with the following guarantees. If $\mathrm{arb}(G) \leq \alpha$, then with probability at least $1 - 1/n^3$, the algorithm accepts. If $\mathrm{arb}(G) > 100 \log^2(n) \cdot \alpha$, then with probability at least $1 - 20 \log(n)/n^4$, the algorithm rejects.
\end{lemma}

We will apply this with $\alpha = \Theta\left( np \right) $, so the query complexity of the algorithm is $O(1/p)$.

We also observe that deciding if the average degree is $O(np)$ can be done efficiently, with high probability.

\begin{lemma}[\cite{feige2004sums}, simplified]\label{lem:feigesums}
    There is an algorithm that uses $\widetilde{O}\left(\sqrt{n/d}\right)$ degree queries and time and estimates the average degree of an input graph $G$ up to a multiplicative factor of $3$, with probability $1 - o(1)$, as long as $G$ has average degree at least $d$.
\end{lemma}

We rely on the following corollary, which turns this into a decision version of the problem.

\begin{corollary}\label{cor:avg-deg-test}
    Consider an input graph $G$ with $e(G)$ edges, and edge probability $p$. There exists an algorithm (call it \textbf{Algorithm B}) that performs expected $O(\log(n)/p)$ queries and has expected runtime $O(\log(n)/p)$ with the following guarantees. If $e(G) \in \left[\frac{1}{3} n^2 p, n^2 p\right]$, the algorithm accepts with high probability. If $e(G) < \frac{1}{5} n^2 p$ or $e(G) > 10 n^2 p$, the algorithm rejects with high probability.
\end{corollary}

\begin{remark}
    Let $\rho_{e}(G) = e(G)/\mathbb{E}_{G' \sim \gnp}\left[ e(G')\right]$. While the focus in this section is on triangles, a modification of this corollary already implies a $(\gnp, \rho_{e})$-quality control algorithm for comparing the number of edges to $\gnp$, with $\widetilde{O}(\sqrt{1/p})\cdot \text{poly}(1/\varepsilon)$ adjacency list queries and runtime. Particularly, \cite{goldreich-ron-2008} provides stronger guarantees than \cite{feige2004sums}, giving a $\widetilde{O}(n^2/e(G))^{1/2} \cdot \text{poly}(1/\varepsilon)$-time algorithm that outputs a $(1 \pm \varepsilon)$-approximation of the average degree in a graph $G$. We can use this to prove a more refined version of \Cref{cor:avg-deg-test}. For quality control, the completeness case follows from the concentration of the number of edges for $G' \sim \gnp$. Soundness follows from rejecting graphs whose edge count is out of $(1 \pm \varepsilon) \binom{n}{2} p$.
\end{remark}

\begin{proof}[Proof of \Cref{cor:avg-deg-test}]
    Consider the algorithm that does the following: first it samples $s = O(\log(n)/p)$ random vertices' degrees and rejects if the sum is less than $\frac{1}{2} n p s$. Then, it runs the procedure from \cite{feige2004sums} (\Cref{lem:feigesums}) with lower-average-degree parameter $d_0=2np/5$, accepting if the estimated average degree is at most $6np$ and rejecting otherwise. Reject if the procedure has not terminated within $n^2$ steps.

    First, we can test if $e(G) < \frac{1}{5} n^2 p$ with the first step of the algorithm. Suppose we sample a multiset of $s = O(\log (n)/p)$ vertices. Let $X_i$ be the degree of the $i$th vertex chosen. The expectation of $X_i$ is $2e(G) / n < \frac{2}{5} np$. Observe that, since $\max(X_i) \leq n$, $\mathbb{E}[X_i^2] \leq n \mathbb{E}[X_i] < \frac{2}{5} n^2 p$.

    By Bernstein's inequality, the probability that the sum of the $s$ random vertices' degrees is at least $\frac{1}{2} n p s$ is at most
    $$\exp\left( - \frac{\frac{1}{2} \left(\frac{1}{10} n p s \right)^2}{\frac{2}{5} s n^2 p + \frac{1}{3} \frac{1}{10} n^2 p s}\right) = \exp \left(- C s p\right),$$
    for some constant $C$. For $s = O(\log(n)/p)$, this is at most $1/n^2$. Therefore, if $e(G) < \frac{1}{5} n^2 p$, the algorithm rejects with high probability. If $e(G)\in[n^2p/3,n^2p]$, the expected degree is in $[2np/3,2np]$, and the same inequality shows that the first step accepts with high probability.

    If the first step does not reject, we run the second step, i.e. the procedure from \cite{feige2004sums}. If the input graph has $e(G) \geq \frac{1}{5} n^2 p$, by \Cref{lem:feigesums} this step takes $O(\sqrt{1/p})$ time. Otherwise, the algorithm can take $O(n^2)$ time; this will happen with probability at most $1/n^2$. The algorithm has expected $O(\log(n)/p)$ queries and has expected runtime $O(\log(n)/p)$. This procedure has the guarantee that if $e(G) \in \left[\frac{1}{3} n^2 p, n^2 p\right]$, the algorithm accepts with high probability, and if $e(G) > 10 n^2 p$, the algorithm rejects with high probability.
\end{proof}

For a graph $G$, let $C_3(G)$ be the number of unlabeled triangles in the graph. In \cite{eden2018faster}, the authors construct an algorithm for estimating the number of $k$-cliques in a graph with low arboricity. Their algorithm repeatedly calls an oracle-based algorithm, which takes in an estimate for the number of $k$-cliques and returns different output values depending on the relationship of the number of $k$-cliques in $G$ to this estimate. In the following two lemmas, we state the oracle versions of the results from this paper, which we will need for our context of quality control.

The following two lemmas are from \cite{eden2018faster}, presented in a simplified form focused on the case of triangle counting.

\begin{lemma}[Lemma 7.1, \cite{eden2018faster}]\label{lem:eden-7-1}
    Consider an input graph $G$ (with $n$ vertices and $m$ edges) with $C_3(G)$ triangles, and parameters $\alpha'$, $t$, $\delta$, $\gamma$. There exists an algorithm (call it \textbf{Algorithm C}) with the following guarantees, when $\mathrm{arb}(G) \leq \alpha'$:
    \begin{enumerate}
        \item  If $C_3(G) < t$, then with probability at least $1 - \gamma$, the algorithm outputs a number $\widehat{t}$ that is less than $t$.
        \item If $t \in [C_3(G)/4, C_3(G)]$, then with probability at least $1 - \gamma$, the algorithm outputs a number $\widehat{t}$ that is in the range $(1 \pm \delta) C_3(G)$.
    \end{enumerate}

    The algorithm performs adjacency list, adjacency matrix, and degree queries. The expected query complexity and runtime of the algorithm are:
    $$O\left( \min\left\{ \frac{n \cdot (\alpha')^2}{t}, \frac{n}{C_3(G)^{1/3}}\right\} + \frac{m \cdot \alpha'}{t} \cdot \frac{C_3(G)}{t}\right) \cdot \mathrm{poly}\left( \log(n/\gamma), 1/\delta\right).$$
\end{lemma}

\begin{lemma}[Corollary 7.2, \cite{eden2018faster}]\label{lem:eden-cor}
    Consider an input graph $G$ (with $n$ vertices and $m$ edges) with $C_3(G)$ triangles, and parameters $\alpha'$ and $t$. There exists an algorithm (call it \textbf{Algorithm D}) with the following guarantee.

    When $\mathrm{arb}(G) \leq \alpha'$, if $C_3(G) \geq t$, then with high probability, the algorithm outputs a number $\widehat{t}$ that is in the range $(1 \pm \delta) C_3(G)$.

    The algorithm performs adjacency list, adjacency matrix, and degree queries. When $C_3(G)\geq t$, the expected query complexity and runtime of the algorithm are:
    $$O\left( \min\left\{ \frac{n \cdot (\alpha')^2}{t}, \frac{n}{C_3(G)^{1/3}}\right\} + \frac{m \cdot \alpha'}{t} \right) \cdot \mathrm{poly}\left( \log(n), 1/\delta\right).$$
\end{lemma}

We describe the algorithm below.

\begin{algorithm}[H]
\caption{Algorithm \textsc{Triangle-Quality} performs quality control on the triangle count of $G$ compared to $\gnp$ in $\widetilde{O}(1/p)$ queries and time.}
\vspace{1em} ~~ \textsc{Triangle-Quality}: On inputs $G$, $n, p, \varepsilon$, and $\gnp$
\begin{quote}
\begin{enumerate}
    \item Run Algorithm A (from \Cref{lem:claim3.23-eden2022}) with (arboricity bound) parameter $\alpha = 2 n p$. If Algorithm A rejects, then reject. 
    \item Run Algorithm B (from \Cref{cor:avg-deg-test}). If Algorithm B rejects, then reject. With the cutoff chosen as in the proof, also reject if it has not terminated after $O(1/p)\cdot\mathrm{poly}(\log(n),1/\varepsilon)$ steps.
    \item Run Algorithm C (from \Cref{lem:eden-7-1}) with parameters $\alpha' = 200 \log^2(n) \cdot n p$ (arboricity bound), $t = (1 - \varepsilon) \binom{n}{3} p^3$, and estimation accuracy $\delta = \varepsilon/4$. If the output $\widehat{t}$ is $< t$, then reject. With the cutoff chosen as in the proof, if the algorithm has not terminated after $O\left(\frac{1}{p} \right) \cdot \mathrm{poly}\left(\log(n), 1/\varepsilon \right)$ steps, then reject.
    \item Run Algorithm D (from \Cref{lem:eden-cor})  with parameters $\alpha' = 200 \log^2(n) \cdot n p$ (arboricity bound), $t = (1 - \varepsilon) \binom{n}{3} p^3$, and estimation accuracy $\delta = \varepsilon/4$. If the output $\widehat{t}$ is $> (1 + \frac{\varepsilon}{2}) \binom{n}{3} p^3$, then reject. With the cutoff chosen as in the proof, if the algorithm has not terminated after $O\left(\frac{1}{p} \right) \cdot \mathrm{poly}\left(\log(n), 1/\varepsilon \right)$ steps, then reject.
    \item Otherwise, accept the graph.
\end{enumerate}
\end{quote}
\label{alg:triangle-quality}
\end{algorithm}

We now prove that this algorithm implies \Cref{lem:triangles-fromliterature-UB}. 

\begin{proof}
    \textbf{Completeness:} We first reason about completeness, i.e., that Algorithm \textsc{Triangle-Quality} accepts $G \sim \gnp$ with high probability (both over the distribution over graphs and the algorithm's randomness). By \Cref{cor:Gnp-arboricity-concentration} and \Cref{lem:Gnp-clique-concentration}, with high probability, since $p = \omega\left(\varepsilon^{-2} \cdot 1/n\right)$, $G \sim \gnp$ satisfies $\mathrm{arb}(G) \leq 2 n p$, $m \in (1 \pm \frac{\varepsilon}{100}) \binom{n}{2} p$, and $C_3(G) \in (1 \pm \frac{\varepsilon}{100}) \binom{n}{3} p^{3}$. 

    We first argue that, when $G \sim \gnp$, all algorithms called as subroutines are able to terminate successfully (and therefore we can utilize their high-probability guarantees). To see this, first note that, given the high-probability bounds on the edge count, triangle count, and arboricity of $G \sim \gnp$, all algorithms called as subroutines run in $O\left(\frac{1}{p} \right) \cdot \mathrm{poly}\left(\log(n), 1/\varepsilon \right)$ expected time. For any fixed completeness error $\tau>0$, choose the subroutines' failure probabilities at most $\tau/8$ and their cutoffs at least $8/\tau$ times the corresponding expected runtime bounds. By Markov's inequality and a union bound, with probability at least $1-\tau$, the runtime of the algorithm for $G$ with $\mathrm{arb}(G) \leq 2 n p$, $m \in (1 \pm \frac{\varepsilon}{100}) \binom{n}{2} p$, and $C_3(G) \in (1 \pm \frac{\varepsilon}{100}) \binom{n}{3} p^{3}$ is $O\left(\frac{1}{p} \right) \cdot \mathrm{poly}\left(\log(n), 1/\varepsilon \right)$. These cutoffs preserve the stated complexity for fixed $\tau$.

    Thus, outside the allocated failure events, the algorithmic subroutines called all terminate for $G \sim \gnp$. With this in mind, we proceed to the correctness of the steps of the algorithm. When the conditions (edge count, triangle count, arboricity) are satisfied, first, for Steps 1 and 2, by the correctness of Algorithms A and B, with high probability Algorithms A and B both accept. Next, in Step 3, since $\frac{1}{4} C_3(G) \leq (1 - \varepsilon) \binom{n}{3} p^3 \leq C_3(G)$ for $\varepsilon \leq 3/(4 + 1/100)$, Algorithm C outputs $$\widehat{t} \in \left(1 \pm \delta\right) C_3(G) \in \left(1 \pm \frac{\varepsilon}{4}\right) \left(1 \pm \frac{\varepsilon}{100}\right) \binom{n}{3} p^{3} \in (1 \pm \varepsilon) \binom{n}{3} p^{3}.$$ So, the output of Algorithm C is greater than $t$. Next, in Step 4, Algorithm D outputs a $\widehat{t}$ satisfying:
    $$\widehat{t} \leq \left( 1 + \frac{\varepsilon}{4}\right) \left( 1 + \frac{\varepsilon}{100}\right) \binom{n}{3} p^3.$$
    Since the algorithm does not reject in any of Steps 1 through 3, it accepts the graph in Step 4, with high probability over $G \sim \gnp$.
    \\\\
    \textbf{Soundness:} We now argue soundness, i.e., that Algorithm \textsc{Triangle-Quality} rejects graphs $G$ such that $C_3(G) \not \in (1 \pm \varepsilon) \binom{n}{3} p^{3}$. First, if the graph additionally satisfies $\mathrm{arb}(G) > 200 \log^2 (n) \cdot n p$, Algorithm A rejects with high probability. If the graph satisfies $e(G) < \frac{1}{5} n^2p$ or $e(G) > 10 n^2 p$, Algorithm B rejects with high probability. Otherwise, if $C_3(G) < (1 - \varepsilon) \binom{n}{3} p^{3}$, Algorithm C outputs a $\widehat{t} < t$ in expected runtime and query complexity $O\left(\frac{1}{p} \right) \cdot \mathrm{poly}\left(\log(n), 1/\varepsilon \right)$ and the algorithm rejects. If $C_3(G) > (1 + \varepsilon) \binom{n}{3} p^{3}$, then Algorithm D outputs $$\widehat{t} > \left(1 - \frac{\varepsilon}{4}\right) \left(1 + \varepsilon\right) \binom{n}{3} p^{3} > \left(1 + \frac{\varepsilon}{2}\right)\binom{n}{3}p^3,$$ 
    for $\varepsilon < 1/2$, in expected runtime and query complexity $O\left(\frac{1}{p} \right) \cdot \mathrm{poly}\left(\log(n), 1/\varepsilon \right)$ and the algorithm rejects. Therefore, the algorithm rejects all $G$ such that $C_3(G) \not \in (1 \pm \varepsilon) \binom{n}{3} p^{3}$, with high probability over the randomness of the algorithm.
    \\\\
    \textbf{Query complexity and runtime:} The query complexity and runtime of Step 1 are $O(1/p)$, from \Cref{lem:claim3.23-eden2022}. Next, by definition of the algorithm, Steps 2 through 4 terminate after $O\left(\frac{1}{p} \right) \cdot \mathrm{poly}\left(\log(n), 1/\varepsilon \right)$ each. Therefore, we obtain the stated query complexity and runtime.
\end{proof}

\begin{remark}
    The algorithm above only used the following properties of $\gnp$: with high probability, the arboricity is bounded, and the number of edges and triangles are each concentrated around their expected values. Therefore, the algorithm and analysis extend to any distribution over graphs with these properties. Indeed, this is a recurring theme in our results: we utilize minimal structural properties about $\gnp$ and the results extend nicely to other distributions over graphs that satisfy these properties.
\end{remark}

\subsection{Query complexity lower bound} \label{sec:triangle-lowerbound}

Recall that $\rho_\Delta(G)$ counts the normalized fraction of triangles in $\gnp$, specifically given by $\rho_\Delta(G) := C_3(G)/(\binom{n}{3}p^3)$.
In this section we prove that the $(\gnp, \rho_\Delta)$-quality control problem, i.e., the quality control problem corresponding to counting the number of triangles in $\gnp$, requires $\Omega(1/p)$ queries when given access to (sorted) adjacency lists of vertices.

\begin{theorem}[Lower bound in \Cref{thm:triangles-usingliterature}]\label{thm:triangle-lb}
    There exists a constant $c>0$ and a polynomially growing function $n_0$ such that for every $p \in (0,1/2]$ and every $n \geq n_0(1/p)$ if an algorithm solves the $(\gnp, \rho_\Delta)$-quality control problem with $q(n,p)$ queries given (sorted) adjacency list access to the input graph, then $q(n,p) \geq c/p$.
\end{theorem}

A comparable result when given access to the adjacency matrix was given in \Cref{sec:motifs-lowerbound}; here, the focus is instead on adjacency list access, to complement the upper bound setting for triangle counts. We remark that the sorted adjacency list access model is a very powerful model for access to a graph. In particular, queries to the adjacency matrix can be answered using $O(\log n)$ queries to the sorted adjacency list. We further remark that our proof even allows ``degree'' queries for unit cost.

To prove \Cref{thm:triangle-lb} we use two distributions --- a YES distribution, which is forced to be $\gnp$ by the quality control problem having $\gnp$ as the target distribution, and a NO distribution which is simply $\gnp$ with a planted clique of size $\ell = O(pn)$.   The technical results show that (1) the two distributions are indistinguishable to algorithms making $o(n/\ell) = o(1/p)$ queries and (2) the distributions differ significantly in triangle counts. Together these results imply the query lower bound on quality control for triangle counting with respect to $\gnp$. 

For a graph $G$ on vertex set $[n]$ and a set $S \subseteq [n]$, we define $G'(G,S)$ to be the graph on vertex set $[n]$ whose edges are the union of edges of $G$ and of the complete graph on vertex set $S$. In what follows we set $D_Y = D_Y^n = \gnp$ and $D_N = D_N^{n,\ell}$ to be the distribution on graphs $G'(G,S)$ where $G \sim \gnp$ and $S \subseteq [n]$ with $|S|=\ell$ is chosen uniformly (independent of $G$). 

\begin{lemma}\label{lem:dy-dn-indist-triangle}
    Let $A$ be a deterministic algorithm making $q$ queries to the (sorted) adjacency list of an $n$ vertex graph and outputting $0/1$. Then 
    $$\left|\Pr_{G \in D_Y^n}[A(G) = 1] - \Pr_{G'\in D_N^{n,\ell}}[A(G')=1] \right| \leq \frac{2q\ell}{n-\ell}.$$
\end{lemma}

\begin{proof}
    We prove the following stronger statement that immediately implies the lemma. For every graph $G$
    $$\Pr_S[A(G'(G,S)) \ne A(G)] \leq \frac{2q\ell}{n-\ell}.$$
    Given $G$, let $\{(u_j,i_j)\}_{j=1}^q$ denote the queries of $A$ on input $G$ (where the response to the query $(u_j,i_j)$ includes the degree of $u_j$, the name $v_j$ of the $i_j$th neighbor of $u_j$ and the degree of $v_j$). 
    Let $Q = \{u_j | j \in [q]\} \cup \{v_j | j \in [q]\}$. Then we have $\Pr_S[S \cap Q \ne \emptyset] \leq \frac{2q\ell}{n-\ell}$. Conditioned on $S \cap Q = \emptyset$ we have that every query of $A$ to $G'$ gets the same answers as the corresponding query in $G$, and so $A(G') = A(G)$. This proves the lemma.
\end{proof}

\begin{lemma}\label{lem:triangle-sep}
     For every $p \in (0,1/2]$ and $n \geq 1000 p^{-1}$, and every $n\geq\ell \geq 2 n p$, the following hold:
    $$\Pr_{G \sim D_Y^n}[| \rho_{\Delta}(G) - 1| \leq 0.1] \geq 1 - 1/12, $$
    $$\mbox{ and } \Pr_{G' \sim D_N^{n,\ell}}[ \rho_{\Delta}(G') \geq 1.9] \geq 1 - 1/12.$$  
\end{lemma}

\begin{proof}
    First, the expected number of triangles in $\gnp$ is $\mu_3 = \binom{n}{3} p^3$. Summing covariances over triangles sharing an edge gives
    $$\mathrm{Var}(C_3(G))=\mu_3(1-p^3)+\binom{n}{2}(n-2)(n-3)p^5(1-p).$$
    Since $np\geq1000$ and $p\leq1/2$, this gives
    $$\frac{\mathrm{Var}(C_3(G))}{\mu_3^2}\leq\frac{12}{(np)^3}+\frac{18}{(np)^2}<\frac{0.1^2}{12}.$$
    Chebyshev's inequality proves $\Pr_{G \sim D_Y^n}[| C_3(G) - \mu_3| \leq 0.1 \cdot \mu_3] \geq 1 - 1/12$.

    Second, for $G' \sim D_N^{n, \ell}$, the planted clique alone has $\binom{\ell}{3}$ triangles. Since $\ell\geq2np$ and $np\geq2$,
    $$\frac{\binom{\ell}{3}}{\mu_3}\geq8\left(1-\frac{1}{2np}\right)\left(1-\frac{1}{np}\right)\geq3.$$
    Thus $\rho_\Delta(G')\geq1.9$ with probability $1$.
\end{proof}

\begin{proof}[Proof of \Cref{thm:triangle-lb}]
    Take $n_0(1/p)=1000/p$. If $p>1/8$, at least one query is necessary, since the complete graph has normalized triangle count $p^{-3}\geq8$. This gives the result after decreasing $c$. For $p\leq1/8$, set $\ell=\lceil2np\rceil$. Then $\ell\leq3np\leq n/2$. Let $A$ be an algorithm that makes $q$ queries, and accepts $G \sim \gnp$ with probability $\geq 2/3$ and accepts $G$ with $\rho_{\Delta}(G) > 1.5$ with probability at most $1/3$. By \Cref{lem:triangle-sep}, $\Pr_{G' \sim D_N^{n, \ell}}[\rho_{\Delta}(G')<1.9] \leq 1/12$, and therefore $\Pr_{G\sim \gnp,R}[A(G) = 1] - \Pr_{G'\sim D_N^{n, \ell},R}[A(G')=1] \geq 1/6$, for $R$ denoting the random coin tosses made by algorithm $A$. Thus, there exists $R_0$ satisfying $\E_{G\sim \gnp}[A'(G) = 1] - \E_{G'\sim D_N^{n, \ell}}[A'(G')=1] \geq 1/6$, for algorithm $A'$ that is the deterministic version of $A$ with the random tape set to be $R_0$. Combining this with \Cref{lem:dy-dn-indist-triangle} yields $q \geq (n-\ell)/(12 \ell)\geq1/(72p)$.
\end{proof}

\subsection{The challenge of extending to k-cliques}\label{sec:challenge-extending-literature}

Suppose we wanted to extend the approach above to quality control of $k$-cliques as compared to $\gnp$. When we apply \cite{eden2022approximating} (arboricity approximation) and \cite{eden2018faster} ($k$-clique count approximation in low arboricity graphs), we obtain a query complexity and runtime scaling with $1/p^{O(k^2)}\cdot \text{poly}(\log (n), 1/\varepsilon)$. Considering low arboricity graphs eliminates some lower-order $O(k)$ terms in the exponent, but the ratio between the arboricity and the number of $k$-cliques in $\gnp$ is not enough to bring the query complexity and runtime to $1/p^{O(k)}$. Indeed, additionally assuming access to uniform edge queries and using corresponding algorithms \cite{assadi2018simple} does not get rid of the $O(k^2)$ term in the exponent. The issue, at an intuitive level, arises because the number of $k$-cliques in $\gnp$, for sparse $p$, is so few, so bounds that scale inversely with the number of $k$-cliques seem to encounter a $1/p^{O(k^2)}$-style bound. It is an interesting future direction to adapt and simplify algorithms such as those in \cite{eden2018faster, assadi2018simple} for quality control of $k$-clique counts as compared to $\gnp$. In the next section, we instead take a different approach, iteratively understanding the number of $\ell$-cliques, for $\ell \leq k$, in the graph and developing a notion of ``robust quasirandomness'' of clique counts on subgraphs of the input graph.

\section{Quality control for Sum-NC0 graph parameters}\label{sec:NC0-properties}

In this section, we provide an alternative perspective on quality control of motif counts and its implications for the structure of an input graph. As we have noted, quality control on motif counts is a fundamental question for several reasons, including that motif counts can provide evidence that a graph
looks quasirandom (e.g. \cite{chung1989quasi, DBLP:journals/rsa/ShapiraY10}). Quality control for motif counts can further be seen to be a fundamental problem because quality control algorithms for motif counts can be used to construct quality control algorithms for more general parameters of graphs, which we show in this section.

We consider graph parameters (as defined in \Cref{def:graph-parameter}) computable by $\sum-NC^0$ circuits, defined formally below.

\begin{definition}\label{def:sum-nc0}
    A Sum-NC$^0$ ($\sum-NC^0$) circuit $C = \sum \circ ~C_{n, m}$ consists of a summation gate $\sum_m :\{0, 1\}^m \to \mathbb{N}$ composed with a circuit $C_{n, m}: \{0, 1\}^n \to \{0, 1\}^m$ with constant fan-in $k$, constant depth $d$, and $m = \mathrm{poly}(n)$, consisting of AND, OR, and NOT gates. 
\end{definition}

Recall the definition of a graph parameter from \Cref{def:graph-parameter}. For a graph parameter $F$, the corresponding quality control algorithm must, with high probability, accept $G \sim \gnp$ and reject graphs $G'$ for which $F(G') \not \in (1 \pm \varepsilon) \cdot \mathbb{E}_{G \sim \gnp}\left(F(G) \right)$.

For a graph parameter $F: \mathcal{G}(n) \to \mathbb{N}$, let its circuit representation have input from $\{0, 1\}^{\binom{n}{2}}$ representing which edges exist in the graph, and an output of the function value of the graph. Let $\rho_{F}(G) := F(G)/\mathbb{E}_{G' \sim \gnp}\left[ F(G')\right]$.

\begin{theorem}\label{thm:sum-nc0-parameter}
    Consider any graph parameter $F: \mathcal{G}(n) \to \mathbb{N}$ whose circuit representation can be expressed as a $\sum-NC^0$ circuit with gate fan-in at most $k$ and depth $d$, such that $\mathbb{E}_{G' \sim \gnp}\left[ F(x^{G'}) \right] > 0$. Suppose $0<p\leq1/2$ and $n\geq B_{k,d,\varepsilon,\tau}p^{-10(2k^d+2)}$. The $(\gnp, \rho_{F})$-quality control problem uses $1/p^{O(k^d)}$ queries to the adjacency matrix of the input graph, and runtime $1/p^{O(k^d)}$.
\end{theorem}

Graph parameters with Sum-NC$^0$ circuits are a generalization and alternative formulation of motif counts. These graph parameters count the number of appearances of local patterns in a graph.

We will, without loss of generality, assume and use that $\mathbb{E}_{G' \sim \gnp}\left[ F(x^{G'}) \right] \geq 1$. If this is not the case, scaling $F$ by some factor to obtain a graph parameter $F'$ is possible, and will allow us to study and perform quality control on $F$. $\rho_F$ is invariant under scaling.

The main technical work for proving \Cref{thm:sum-nc0-parameter} arises in reducing a graph parameter  $F$ whose circuit representation can be expressed as a Sum-NC$^0$ circuit with gate fan-in at most $k$ and depth $d$ to a sum of counts of motifs of size at most $2k^d$. Once we have done this, we can use the composition of quality control problems. We will perform $(\gnp, \rho_H)$-quality control for all motifs $H$ on at most $2k^d$ vertices. If any of the $(\gnp, \rho_H)$-quality control algorithms called rejects, then so does the algorithm we construct for $(\gnp, \rho_F)$-quality control, where $\rho_F(G) = F(G)/\mathbb{E}_{G' \sim \gnp}\left[F(G') \right]$ for graph parameter $F$. We argue that, due to the representation of $F$ as a sum of counts of small motifs, this suffices for $(\gnp, \rho_F)$-quality control.

Since quality control for any such Sum-NC$^0$ graph parameter can be accomplished via quality control of all small motif counts, this also demonstrates the robustness of the composability property of quality control problems.

\subsection{Reducing graph parameters to sums of motif counts}

In this section we prove the following lemma.

\begin{lemma}\label{lem:nc0-parameter-sum}
    Consider any graph parameter $F: \mathcal{G}(n) \to \mathbb{N}$ whose circuit representation can be expressed as a $\sum-NC^0$ circuit $C = \sum \circ ~C_{\binom{n}{2}, m}$ where $C_{\binom{n}{2}, m}$ has gate fan-in $k$ and depth $d$. Then, $F$ can be expressed as a nonnegative linear combination of labeled induced counts of motifs on at most $2 k^d$ vertices in the input graph.
\end{lemma}

\begin{proof}
    First, the sum of the output bits of any $NC^0$ circuit $C_{\binom{n}{2}, m}:\{0, 1\}^{\binom{n}{2}} \to \{0, 1\}^m$ with gate fan-in $k$ and depth $d$ can be written as a summation over at most $m \cdot 2^{k^d}$ AND gates with fan-in at most $k^d$. To see this, observe that each output bit can only depend on at most $k^d$ input variables/coordinates. Therefore, each output bit can be written as a summation over at most $2^{k^d}$ AND gates with fan-in $k^d$, taken over input coordinates and/or their negations. To see this, represent each bit of the output of the $NC^0$ circuit as a truth table. Since each output bit depends on at most $k^d$ input variables, and coordinates may be negated, the truth table corresponds to an OR taken over at most $2^{k^d}$ AND gates with fan-in $k^d$. This OR can be turned into a summation over AND gates, because this originated from a truth table in which only one AND gate will output 1 on an input.

    For $i\in[m\cdot2^{k^d}]$, and $S_i \subseteq [\binom{n}{2}]$, $|S_i| \leq k^d$, let $g_i:\{0,1\}^{S_i}\to\{0,1\}$ be the $i$-th AND gate in this new representation of $\Sigma-NC^0$ circuit $C$. Let $N_i \subseteq S_i$ consist of the indices of the input bits that should be negated before taking the AND gate. For $j\in S_i$, let $h_{i,j}:\{0,1\}\to\{0,1\}$ be defined as $h_{i,j}(b)=1-b$ if $j\in N_i$ and $h_{i,j}(b)=b$ otherwise.
    
    We can express any $g_i$ for $i \in [m \cdot 2^{k^d}]$ as:
    \begin{equation}\label{eq:nc0-g-to-h}
       g_i(x) = \prod_{j \in S_i} h_{i, j}(x_j). 
    \end{equation}
    Therefore, equating a graph $G$ with the string $x^G \in \{0, 1\}^{\binom{n}{2}}$ indicating the edges present, we can rewrite $C:\{0,1\}^{\binom{n}{2}}\to\mathbb{N}$ equivalently as $$C(x^G)=\sum_{i\in[m\cdot2^{k^d}]}g_i(x^G) =\sum_{i\in[m\cdot2^{k^d}]}\prod_{j\in S_i}h_{i,j}(x^G_j).$$

    Consider permutations $\sigma(x^G)$ of the \textit{vertices} of the $n$-vertex input graph $G$. First, since $F$ is invariant under graph isomorphism by definition of a graph parameter:
    $$F(x^G) = \frac{1}{n!} \sum_{\sigma} F(\sigma(x^G)) =  \frac{1}{n!} \sum_{\sigma}\sum_{i\in[m\cdot2^{k^d}]} g_i(\sigma(x^G))= \frac{1}{n!} \sum_{i\in[m\cdot2^{k^d}]} \sum_{\sigma}g_i(\sigma(x^G)).$$
    We now argue that $\sum_{\sigma}g_i(\sigma(x^G))$ corresponds to counting some induced motifs on at most $2k^d$ vertices in the graph. First, the set $\{\sigma(x^G)\}_{\sigma}$ includes all ordered sets of vertices in $G$. Therefore, using the formula \Cref{eq:nc0-g-to-h} for $g_i$, $\sum_{\sigma} g_i(\sigma(x^G))$ counts the number of appearances of a specific pattern $P$ of edges and non-edges among all labeled subsets of $|S_i|$ edges in the graph. This count can be computed by summing over all induced motifs on up to $2|S_i|$ vertices whose edge/non-edge pattern includes $P$. Therefore, $\sum_{\sigma} g_i(\sigma(x^G))$ is equivalent to a sum over a number of \textit{labeled induced counts} of some motifs on at most $2 |S_i| \leq 2 k^d$ vertices in the input graph.
More precisely, if the pattern uses $v$ vertices, each labeled embedding extends to $(n-v)!$ permutations. Completing all unspecified adjacencies gives a sum of induced counts with nonnegative coefficients. After dividing by $n!$ and summing over $i$, we obtain
$$F(G)=\sum_H a_H C_H(G),\qquad a_H\geq0,$$
where every $H$ has at most $2k^d$ vertices. Constant terms can be represented by the one-vertex count divided by $n$.
\end{proof}

\subsection{Quality control algorithm and analysis}

We now present the quality control algorithm for graph parameters computable by Sum-NC$^0$ circuits, as we have developed the technical lemmas needed to prove its correctness as a quality control algorithm. 

\begin{algorithm}[H]
\caption{Algorithm \textsc{Graph-Parameter-Quality} performs quality control for a graph parameter computable by a Sum-NC$^0$ circuit with gate fan-in $k$ and depth $d$, compared to $\gnp$. The query complexity and runtime of the algorithm is $1/p^{O(k^d)}$.}
\vspace{1em} ~~\textsc{Graph-Parameter-Quality}: On inputs $G$ and $n, p, \varepsilon, k, d$
\begin{quote}
\begin{enumerate}
    \item For all motifs $H$ on at most $2 k^d$ vertices, run Algorithm \textsc{Induced-Motif-Quality-Efficient} (on inputs $G, n, p, \gnp, \varepsilon$, and $H$) to perform quality control on the count of labeled induced copies of $H$ in $G$, allocating completeness error $\tau/M(k,d)$ to each call. 
    \item If the quality control algorithm \textit{rejects} on any motif $H$ on at most $2 k^d$ vertices, then output reject. Else, output accept.
\end{enumerate}
\end{quote}
\label{alg:graph-parameter-quality}
\end{algorithm}

Observe that the quality control algorithm is the \textit{same} for all graph parameter computable by a Sum-NC$^0$ circuit with fan-in $k$.

Towards proving \Cref{thm:sum-nc0-parameter}, we begin with the following lemma.

\begin{lemma}\label{lem:graph-parameter-to-motifs-outofrange}
    Let $F$ be a graph parameter $F: \mathcal{G}(n) \to \mathbb{N}$ whose circuit representation is a Sum-NC$^0$ circuit $C = \sum \circ ~C_{\binom{n}{2}, m}$, which satisfies $\mathbb{E}_{G' \sim \gnp}\left[ F(x^{G'}) \right] \geq 1$. Consider a graph $G \in \mathcal{G}(n)$ and the corresponding string with edges/non-edges $x^G \in \{0, 1\}^{\binom{n}{2}}$. If 
    $$\left| F(x^G) - \mathbb{E}_{G' \sim \gnp}\left[ F(x^{G'}) \right]\right| \geq \varepsilon \cdot \mathbb{E}_{G' \sim \gnp}\left[ F(x^{G'}) \right],$$
    then there exists a motif $H$ on at most $2 k^d$ vertices whose labeled induced count $C_H(G)$ satisfies:
    \begin{equation}\label{eq:nc0-1}
       \left| C_H(G) - \mathbb{E}_{G' \sim \gnp}\left[ C_H(G') \right]\right| \geq \varepsilon \cdot \mathbb{E}_{G' \sim \gnp}\left[ C_H(G') \right].
    \end{equation}
\end{lemma}

\begin{proof}

    Let $\mu^{(F)} := \mathbb{E}_{G' \sim \gnp}\left[ F(x^{G'}) \right]$. Let $\mu_H := \mathbb{E}_{G' \sim \gnp}\left[ C_H(G') \right]$. Let 
    $M(k, d) = \sum_{j = 1}^{2 k^d}2^{\binom{j}{2}}$, which is the number of induced subgraphs on at most $2 k^d$ vertices.

    By \Cref{lem:nc0-parameter-sum}, write $F(G)=\sum_Ha_HC_H(G)$ with $a_H\geq0$. Taking expectations gives $\mu^{(F)}=\sum_Ha_H\mu_H$. By the triangle inequality,
    $$\varepsilon\mu^{(F)}\leq|F(G)-\mu^{(F)}|\leq\sum_Ha_H|C_H(G)-\mu_H|.$$
    If every term with $a_H>0$ satisfied $|C_H(G)-\mu_H|<\varepsilon\mu_H$, the right-hand side would be strictly less than $\varepsilon\sum_Ha_H\mu_H=\varepsilon\mu^{(F)}$, a contradiction. Thus \Cref{eq:nc0-1} holds for some $H$.
\end{proof}

We now prove \Cref{thm:sum-nc0-parameter} for Algorithm \textsc{Graph-Parameter-Quality}.

\begin{proof}[Proof of \Cref{thm:sum-nc0-parameter}]

\textbf{Completeness:} Suppose the input graph $G$ is drawn from $\gnp$. Then, by the completeness of Algorithm \textsc{Induced-Motif-Quality-Efficient} proven in \Cref{sec:induced-motifs-algo}, each call of Algorithm \textsc{Induced-Motif-Quality-Efficient} in Step 1 accepts $G$ with probability at least $1 - o(1)$. By a union bound over the $M(k, d)$ (a constant) possible induced subgraphs over at most $2 k^d$ vertices, Step 2 then accepts $G \sim \gnp$ with probability $1 - o(1)$.

\textbf{Soundness:} Suppose the input $G$ satisfies
$$\left| F(x^G) - \mathbb{E}_{G' \sim \gnp}\left[ F(x^{G'}) \right]\right| > \varepsilon \cdot \mathbb{E}_{G' \sim \gnp}\left[ F(x^{G'}) \right].$$
By the proof of \Cref{lem:graph-parameter-to-motifs-outofrange}, there exists a motif $H$ on at most $2k^d$ vertices whose labeled induced count $C_H(G)$ satisfies:
    $$\left| C_H(G) - \mathbb{E}_{G' \sim \gnp}\left[ C_H(G') \right]\right| > \varepsilon \cdot \mathbb{E}_{G' \sim \gnp}\left[ C_H(G') \right].$$
By the soundness of Algorithm \textsc{Induced-Motif-Quality-Efficient} proven in \Cref{sec:induced-motifs-algo}, therefore the algorithm rejects input graph $G$ due to the count of motif $H$ with high probability.

\textbf{Query complexity and runtime:} These are inherited from the calls of Algorithm \textsc{Induced-Motif-Quality-Efficient} on motifs on at most $2k^d$ vertices. Therefore, the query complexity and runtime are $1/p^{O(k^d)}$.
\end{proof}

\section{Related work, expanded}\label{sec:appendix-related-work}

In this section, we review related work. The first four topics are alternative ways to work with distributional assumptions. The fifth topic discusses work in the literature that highlights the importance of motif counting.

\vspace{-2ex}
\paragraph{Average-case complexity} Quality control relates closely to the area of average-case complexity \cite{DBLP:journals/siamcomp/Levin86, DBLP:books/sp/goldreich2011/Goldreich11h, DBLP:journals/fttcs/BogdanovT06}, which aims to understand the average complexity of a problem when inputs are drawn from a distribution. Naturally, some problems in average-case complexity can be framed as quality control instances, though the quality control framework was not previously defined. Notable examples include refutations of random 3SAT \cite{DBLP:conf/stoc/Feige02} and natural proofs \cite{DBLP:journals/jcss/RazborovR97}. For the problem of refuting random 3SAT instances, a seminal definition of Feige~\cite{DBLP:conf/stoc/Feige02} requires that a refutation algorithm for 3SAT problem on distribution $\mathcal{D}$ should only accept {\em unsatisfiable} 3SAT inputs, while accepting most inputs drawn from $\mathcal{D}$. This corresponds exactly to a quality control problem with the parameter $\rho_{\mathrm{SAT}}(\phi)$ being $1$ if a formula $\phi$ is unsatisfiable and $0$ otherwise. 

Another class of problems that have been studied in the context of average case complexity are the ``planted solution problems'' --- notably planted SAT or planted clique problems; see, e.g., \cite{DBLP:journals/rsa/AlonKS98, racz2020finding, huleihel2024random}. In \Cref{sec:average-case-complexity}, we explain how our quality control algorithms give a fast average-case algorithm for approximate $k$-clique counting as well as planted detection of many adversarially planted $k$-cliques. One key difference between our motivation and the motivations in average case complexity is the typical {\em focus}: Our primary focus is on sublinear algorithms, whereas traditional work in average case complexity focused on polynomial time algorithms.

\vspace{-2ex}
\paragraph{Testing Assumptions of Algorithms}

The idea that one can sidestep the large sample lower bounds required to test distributions, and directly get to validating the conclusion, arose recently in work of Rubinfeld and Vasilyan~\cite{rubinfeld2023testing}, in the notion they call {\em testable learning}. They consider the setting of a learning algorithm that works with a distribution $\mathcal{D}$ with the assumption that $\mathcal{D} \in \mathcal{D}_0$. They show that it is possible to test that the output of the learner can be trusted with lower sample complexity than required to test the distribution alone, such that if the learned hypothesis is incorrect then the tester will reject, while if $\mathcal{D} \in \mathcal{D}_0$ the tester will accept. Our model is inspired by their model --- and the emphasis on testing the validity of the computation rather than the correctness of a distributional assumption.

In concurrent work \cite{TestableAlgorithmsEden2025}, Eden, Rubinfeld, and Vasilyan consider the problem of approximately counting edges and triangles in graphs in the sublinear access model within a framework they call ``testable approximate counting.'' In this framework, \textit{every} graph with a certain property (an arboricity bound) must be accepted with high probability and an estimate of the number of triangles must be output. Every graph for which the algorithm will not work well (typically, those without the property) must be rejected with high probability.

There are fundamental differences in the models of \cite{rubinfeld2023testing, TestableAlgorithmsEden2025} and ours: primarily, we are \textit{not} required to pass \textit{all} inputs $x$ for which $\rho(x) \approx 1$, whereas in both papers above we would be required to accept all such inputs. This yields very different results and techniques.      

\vspace{-2ex}
\paragraph{Distribution Testing} The concern of ``Does the input really come from distribution $\mathcal{D}$?'' has been approached in several other ways in the literature. One natural class of problems motivated by this concern are the distribution testing problems~(see, e.g., \cite{DBLP:journals/ftcit/Canonne22}) where the goal is to determine, given access to samples from a distribution $\mathcal{D}$ if $\mathcal{D} \in \mathcal{D}_0$ or if $\mathcal{D}$ is far from $\mathcal{D}_0$. However, these results require many samples from $\mathcal{D}$ to achieve their goals.\footnote{Indeed, if we had multiple samples, testing if the distribution over graphs we are receiving is close to the \erdosrenyi graph distribution would require at least exponentially many samples, in the size of the graph. Each sample itself is a very large graph, which could require significant computation to read, so this is computationally intractable for several reasons.} Our goal diverges from distribution testing because we aim to do something with {\em sublinear} access to {\em single} sample. 
\vspace{-2ex}
\paragraph{Quasirandomness} The study of quasirandom graphs aims to give ``deterministic'' definitions of randomness in graphs \cite{thomason1987pseudo, chung1989quasi, chung2002sparse, krivelevich2006pseudo, chung2008quasi, zhao2023graph}, often identifying equivalent deterministic structural properties possessed by graphs from $\gnp$ with high probability. One way to view these works in our language is that they identify specific quality parameters $\rho$ that concentrate around $1$ for random graphs. Motif counts are among these quality parameters. Many of the results in quasirandomness can then be described as implications showing, e.g., graphs with parameter $\rho\approx 1$ also have some parameter $\rho' \approx 1$ --- thus giving in our language reductions among quality control problems. In some cases, they also give efficient (polynomial-time) algorithms to compute some non-trivial parameters. Thus, these results form a useful toolkit for the design of quality control algorithms, and indeed, they play a significant role in our results.

\vspace{-2ex}
\paragraph{Motif counting} 

Motif counting is both theoretically and empirically essential to graph theory. On the theoretical side, beyond being a measure of graph quasirandomness \cite{chung1989quasi}, motif counts are also central to the theory of graph limits. For a sequence of dense graphs, if motif densities converge to some value for every finite motif, the sequence of graphs converges to a limiting object (graphon) that is characterized by limiting motif densities; see Chapter 4 of \cite{zhao2023graph}. Additionally, motif counting is a fundamental algorithmic problem (both for exact and approximate counting); see, for example, \cite{duke1995fast, DBLP:journals/jacm/AlonYZ95, milo2002network, DBLP:conf/wsdm/0003S20, DBLP:journals/csur/RibeiroPSAS21, DBLP:conf/www/SeshadhriT19}. Motif counting connects to topics such as average-case complexity and planted solution problems, and additionally in some contexts the counts of motifs in graphs parallel the role of moments of a distribution \cite{bickel2011method}.

Empirically, motif counting 
has applications across network analysis, including technological and social networks \cite{DBLP:conf/saci/DumaT14, DBLP:conf/globalsip/DeyGP17, DBLP:conf/icwsm/ColettoGGL17, DBLP:journals/access/LiHWHL18, dey2019network}, biological networks \cite{shen2002network, DBLP:journals/bmcsb/KimLWP11, wang2014identification, patra2020review}, brain  networks \cite{sporns2004motifs, battiston2017multilayer}, and ecological networks \cite{baker2015species, simmons2019motifs, stone2019network}. See \cite{DBLP:conf/www/SeshadhriT19} for references to dozens of other papers applying motif counting to network science, graph processing, bioinformatics, and more.

By now, there is also a vast body of work on sublinear algorithms (see, e.g., surveys by Rubinfeld and Shapira ~\cite{DBLP:journals/siamdm/RubinfeldS11} or Ron~\cite{DBLP:journals/fttcs/Ron09}). Approximately counting and sampling motifs in ``worst-case'' graphs with sublinear algorithms have been thoroughly studied --- e.g. \cite{fichtenberger2020sampling, biswas2021towards, eden2025approximately, DBLP:journals/corr/AliakbarpourBGP16} --- and optimal bounds have been proven in various settings \cite{assadi2018simple, eden2017approximately, eden2018faster, eden2020almost}. We stress two main differences in our setting: Sublinear time algorithms are required to work for worst-case inputs, while we intend to take advantage of the distribution of inputs --- leading to a weaker completeness requirement in our case. Moreover, unlike some sublinear algorithms which allow algorithms to produce outputs that would be consistent with the output on a small perturbation of the input according to some distance measure, our model targets a fixed (real-valued) parameter.
\section{Connection to average-case complexity}\label{sec:average-case-complexity}

\subsection{Average-case complexity and examples}\label{sec:appendix-avg-case-examples}

Given the proximity of quality control to the notion of average case complexity, one may anticipate seeing instances of quality control problems arising naturally in the field. Indeed, some notable classes of such problems are ``refutations of random 3SAT'' problems \cite{DBLP:conf/stoc/Feige02}, strong refutation \cite{DBLP:journals/cacm/Feige25, DBLP:journals/cacm/GuruswamiKM25}, and natural proofs \cite{DBLP:journals/jcss/RazborovR97}. Additionally, as noted by \cite{Trevisan2006post}, certifying properties of random instances is a challenge and a crucial step in explicit constructions from random objects. The problem of certifying properties of random objects in these various settings can be framed as a quality control problem.

Beyond these, another related class of problems whose average-case performance is studied is ``planted solution problems,'' such as planted SAT or planted clique problems. Broadly, in this class of problems, there is a natural distribution $\mathcal{D}_0$ and then a planted distribution $\mathcal{D}_1$ that is far in total variation distance from $\mathcal{D}_0$ but has many locally similar features --- and the challenge is to distinguish the two classes. This problem has been previously studied when the algorithm has access to the entire graph (e.g., \cite{DBLP:journals/rsa/AlonKS98}), as well as in the sublinear access setting with query access to the adjacency matrix of the graph \cite{racz2020finding, huleihel2024random}. If one were to consider the quality control problem $(\mathcal{D}_1, \rho)$ where $\rho(x)=1$ if $x$ has a planted solution, then a solution to the quality control problem also solves the associated planted solution problem in both the sublinear and global frameworks. Indeed, our results imply that there is an algorithm with query complexity and runtime $p^{-O(\Delta(H))}$ that distinguishes between $\gnp$ and $\gnp$ with an adversarially placed structure for which the resulting graph has a labeled induced count of $H$ outside $(1\pm\varepsilon)\mu_H$, for a constant-sized subgraph $H$ and $\mu_H:= \E_{G \sim \gnp}[C_H(G)]$.

For some references to the traditional cases of the study of average-case clique counting, see e.g., \cite{boix2021average} for counting cliques over $\gnp$ and \cite{DBLP:conf/focs/GoldreichR18} for worst-case to average-case reductions. Quality control differs in its framework as well as the lens through which we study it (here, sublinear algorithms). As we observe in \Cref{sec:avg-case-approximate-clique}, through quality control we give a fast average-case motif counting algorithm.

\subsection{From average-case complexity to quality control and back}\label{sec:general-avg-case-qc}

Quality control has a close connection to the \textit{average-case complexity} of algorithmic / computational problems, which we further explore in this section. We discuss how to view average-case problems through the perspective of quality control, and vice versa.

One can consider two notions of what it means for an algorithm $A$ to compute a function $\rho$ on distribution $\mathcal{D}$: A strong definition may require that $A$ is correct with high probability on every input, but the distribution of the runtime is nice (e.g., has low-expected value, or is mostly fast) when inputs are drawn from $\mathcal{D}$. A weak requirement may be that $A$ is always fast, but only correct on most inputs drawn from $\mathcal{D}$. When translated to our setting (of $\rho$ being a real-valued parameter that we want to test for closeness to $1$), a quality control algorithm for $(\mathcal{D}, \rho)$ is implied by a strong average case algorithm; and implies a weak average case algorithm. In this section, we prove that a quality control algorithm for $(\mathcal{D}, \rho)$ also implies a strong average case algorithm with the same runtime, under stronger assumptions about the concentration of $\mathcal{D}$ around its expected value. 

\paragraph{Average-case problems to quality control problems}

Works on average-case algorithms typically define a ``good'' property of inputs with the following guarantees. With probability at least $1 - 1/t_2$, for $t_2 = \omega(1)$, an input drawn from $\mathcal{D}$ possesses the good property, and thus there exists an algorithm $A$ with worst-case runtime $t_1$ for this input. This ``good'' property is possible to certify/verify. On the other hand, for worst-case inputs, the algorithm $A$ has worst-case runtime $t_2$. In such a set-up, the average-case complexity of the algorithm $A$ is:
$$(1 - 1/t_2) \cdot t_1 + (1/t_2) \cdot t_2  = O(t_1).$$
See \cite{DBLP:journals/jal/DyerF89, marcussen2025fast} for examples of papers that define and work with such a good property.

When an average-case problem has such a ``good'' property, we can define a quality control problem as follows. We consider $(\mathcal{D}, \rho)$-quality control, with the same distribution $\mathcal{D}$ over inputs, and $\rho$ defined as:
$$\rho(x) = \begin{cases}
    1 ~~ \text{ if $A$ halts on $x$ within $t_1$ steps} \\ 
    0 ~~ \text{ otherwise.}
\end{cases}$$

Consider the algorithm $B$ that, on input $x$, accepts if $A$ halts within $\leq t_1$ steps on input $x$. Under the good-property assumption above, and for $1/t_2\leq1/3$, algorithm $B$ is an algorithm for $(\mathcal{D}, \rho)$-quality control. For completeness, $x \sim \mathcal{D}$ is \textit{good} with high probability and thus $\rho(x) = 1$. For soundness, if the runtime is longer than $t_1$, this is detectable and $\rho(x) = 0$.

\paragraph{Quality control problems to average-case problems}

Quality control problems serve as an intermediate between the problems of ``Constructing an algorithm that is correct with high probability over $\mathcal{D}$'' and average-case complexity, i.e.,  ``Constructing an algorithm that is correct on every input, and has average-case time $T$ over $\mathcal{D}$.'' 

$(\mathcal{D}, \rho)$-quality control serves as an intermediate because, first, it checks which inputs the algorithm is correct with (by soundness) as opposed to the former, but does not give worst-case algorithmic guarantees for computing $\rho(x)$ when $\rho(x) \not \approx \mathbb{E}_{y \sim \mathcal{D}}[\rho(y)]$, which average-case algorithms must.

In general, we have the following.

\begin{proposition}\label{prop:avg-from-qc}
    Let $\mathcal{D}$ be a distribution supported over $\mathcal{X}$, and let $\rho:\mathcal{X}\to\mathbb{R}_{\geq0}$ have mean $\mu=\mathbb{E}_{x\sim\mathcal{D}}\rho(x)>0$, given as input or computable in $O(t_3)$ time. Suppose $0<\varepsilon\leq1$ and $C$ is a $(\mathcal{D},\rho/\mu,\varepsilon/2)$-quality control algorithm with query complexity and runtime $t_3$, satisfying
    $$\Pr_{x\sim\mathcal{D},R}[C(x)=\textsc{reject}]\leq O(t_3/t_4).$$
    Suppose $E$ computes a $(1\pm\varepsilon)$-approximation of $\rho(x)$ for any $x\in\mathcal{X}$ in query complexity and runtime $t_4$, with failure probability at most $1/6$. Then there exists an algorithm that approximates $\rho(x)$ to a $(1\pm\varepsilon)$ factor with high constant probability on every input and has average-case query complexity and runtime $O(t_3)$ over $\mathcal{D}$.
\end{proposition}
The algorithm for approximating $\rho(x)$ is given below, with $C$ amplified to soundness error at most $1/6$ by a constant number of independent repetitions.

\begin{algorithm}[H]
\caption{Algorithm \textsc{Average-Case-From-Quality-Control} uses a $(\mathcal{D}, \rho/\mu)$-quality control algorithm to construct an average-case algorithm for computing $\rho$.}
\vspace{1em} ~~ \textsc{Average-Case-From-Quality-Control:} On inputs $x$ and quality control algorithm $C$:
\begin{quote}
\begin{enumerate}
    \item If $C(x) = $ \textsc{accept}, then output $ \mathbb{E}_{y \sim \mathcal{D}}[\rho(y)]$.
    \item Else, output $E(x)$.
\end{enumerate}
\end{quote}
\label{alg:average-case-from-quality-control}
\end{algorithm}

\begin{proof}[Proof of \Cref{prop:avg-from-qc}]
First, if $|\rho(x)/\mu-1|\leq\varepsilon/2$, then $\mu$ is a $(1\pm\varepsilon)$-approximation of $\rho(x)$. Otherwise, $C$ rejects with probability at least $5/6$. Together with the guarantee of $E$, this gives a correct approximation with probability at least $2/3$ on every input.

We now argue about the runtime. Amplification changes the rejection probability and runtime by at most a constant factor, so the average-case runtime is at most
$$O(t_3)+t_4\Pr_{x\sim\mathcal{D},R}[C(x)=\textsc{reject}]=O(t_3).$$
The same calculation bounds the average-case number of queries.
\end{proof}

\subsection{Approximate motif counting}\label{sec:avg-case-approximate-clique}

We now instantiate the framework of moving from quality control problems to average-case problems for the case of approximately counting the number of copies of a motif in a graph. Such an algorithm must be correct with high constant probability (over the algorithm's randomness) on all inputs.

To put the various types of average-case questions, including quality control, in comparison for motif counting, recall that $G \sim \gnp$ satisfies the property that the number of copies of a motif $H$ is approximately equal to its expected amount under $\gnp$, with high probability. However, outputting this expectation for all inputs $G$ is not meaningful. Quality control gives us a way of efficiently verifying, with high probability, whether the number of copies of $H$ is indeed close to its expectation under $\gnp$, telling us whether this is indeed a good estimate for the count of the motif $H$ in the specific graph $G$. Suppose we want an algorithm that outputs a good estimate for \textit{all} $G$, with good average-case query complexity and runtime over $\gnp$. As described more generally in \Cref{sec:general-avg-case-qc}, we can do so by combining quality control and a worst-case algorithm that is called very infrequently.

Specifically, from \Cref{sec:general-avg-case-qc}, we obtain the following.

\begin{theorem}\label{thm:avg-case-k-clique-count}
    Let $H$ be a fixed motif on $k$ vertices and suppose $0<p\leq1/2$. For sufficiently large constants $c$ and $A=A(H,\varepsilon)$, suppose
    $$\frac{n}{\log^2 n}\geq Ap^{-c(\Delta(H)+1)}.$$
    The average-case query complexity and runtime over $\gnp$ of approximately counting labeled induced copies of $H$, with high constant probability on every input, are $1/p^{O(\Delta(H))}$. For an edgeless motif, they are $O_{H,\varepsilon}(1)$.
\end{theorem}
For other distributions, the same conclusion holds if the rejection probability satisfies the bound in \Cref{prop:avg-from-qc}.

Let $C_H(G)$ be the number of labeled induced copies of motif $H$ in $G$. We utilize the following lemma, whose proof uses higher moments of the same indicator sum as \Cref{lem:Gnp-motif-concentration}.

\begin{lemma}\label{lem:gnp-extra-concentration-average-case}
    Let $H$ be a fixed motif. For every fixed $a>0$ and accuracy $\eta>0$, there exists $A=A(H,a,\eta)$ such that, for $0<p\leq1/2$ and $n\geq Ap^{-2m_H}$,
    $$\Pr_{G\sim\gnp}\left[|C_H(G)-\mathbb{E}C_H(G)|>\eta\mathbb{E}C_H(G)\right]\leq n^{-a}.$$
\end{lemma}

\begin{proof}
    Write $X=C_H(G)=\sum_{\phi}I_{\phi}$ over injective embeddings, and fix an integer $b>a+1$. In the expansion of $\mathbb{E}[(X-\mathbb{E}(X))^{2b}]$, join embeddings that share an edge variable. Terms with an isolated embedding vanish. Within each component, list the embeddings so that each embedding after the first shares an edge variable with an earlier embedding, and therefore repeats at least two vertices. Thus a contributing tuple uses $2bk-u$ vertices with $u\geq2b$. An embedding repeating $v$ vertices repeats at most $m_Hv$ required edges. Shared non-edge constraints cost only a constant, since $1-p\geq1/2$. Summing over overlap patterns gives
    $$\frac{\mathbb{E}[(X-\mathbb{E}(X))^{2b}]}{(\mathbb{E}(X))^{2b}}
    \leq A_{H,b}\sum_{u=2b}^{2bk}(np^{m_H})^{-u}.$$
    Since $np^{m_H}\geq\sqrt{n}$, this is at most $A'_{H,b}n^{-b}$. Markov's inequality and a sufficiently large $A$ give the lemma.
\end{proof}

For our worst-case algorithm (which returns an approximation for the number of $H$ copies for any graph), we use the trivial algorithm that enumerates over all $|H|$-size subgraphs of the input graph and checks if this subgraph is a copy of $H$. (There may be more efficient algorithms, but we want something that works independent of the number of $H$ copies in the graph or other properties, and this suffices.)

Let $\rho_H(G) = C_H(G)/\mathbb{E}_{G' \sim \gnp}\left[ C_H(G')\right]$, for $C_H(G)$.

\begin{proof}[Proof of \Cref{thm:avg-case-k-clique-count}]
    Run Algorithm \textsc{Induced-Motif-Quality-Efficient} with accuracy $\varepsilon/2$ in independent stages. At stage $i$, use samples $O(i)$ times larger than the base samples. If a stage accepts, output $\mu_H=\mathbb{E}_{G\sim\gnp}C_H(G)$. After $R=\lceil(k+3)\log_2 n\rceil$ rejected stages, count $H$ exhaustively. The exponential soundness bounds (or Hoeffding's inequality for direct sampling) make the false-acceptance probability at most $2^{-i}/12$ at stage $i$. A union bound, as in \Cref{prop:avg-from-qc}, gives correctness on every input.

    We now argue about the average-case runtime. Let $t_0=p^{-O(\Delta(H))}$ bound the base sample sizes in the sparse case. Call $G$ good if, for every non-induced test with tolerance $\delta$, its global degrees are in $(n-1)(p\pm\delta/8)$, its global codegrees are in $(n-2)(p^2\pm\delta/8)$, and $C_H(G)\in(1\pm\varepsilon/8)\mu_H$. By Chernoff bounds and \Cref{lem:gnp-extra-concentration-average-case}, $G\sim\gnp$ is good except with probability $n^{-(k+3)}$.

    For a fixed good graph, condition on the sampled vertices whose degree or codegree is checked. Chernoff bounds on the remaining independent positions and a union bound give probability at least $7/8$ that all checks pass. The size condition absorbs the $O(1/n)$ corrections and bounds the probability of a repeated vertex by $O((Rt_0)^2/n)\leq1/8$. Thus each stage accepts with probability at least $1/2$. The same holds for direct sampling by Hoeffding's inequality.

    Independence bounds the probability of reaching stage $i$ by $2^{-(i-1)}$. Since stage $i$ costs $O_H(i^3t_0^3)$, the expected cost on a good graph is at most
    $$O_H(t_0^3)\sum_{i\geq1}i^3 2^{-(i-1)}+O_H(n^k)2^{-R}=p^{-O(\Delta(H))}.$$
    Bad graphs contribute at most $O_H(n^k)n^{-(k+3)}=O_H(1)$. The same argument bounds queries. For the first branch of Algorithm~\ref{alg:induced-motif-quality-efficient}, stage $i$ costs $O_{H,\varepsilon}(i)$, so the same argument gives constant expected cost.
\end{proof}
\section{Concentration of subgraph counts}\label{sec:appendix}

For the completeness arguments of the paper -- i.e., the statement that a $(\gnp, \rho)$-quality control algorithm accepts $G \sim \gnp$ with high probability -- we will utilize that, in $\gnp$, the counts of small motifs are well-concentrated around their expectations. Additionally, for lower-bounding the query complexity of quality control for motif counts, we will consider an alternative random graph distribution known as the \textit{stochastic block model} and need concentration of counts of small motifs. In this section, we provide proofs of concentration of small motif counts for these two models, relying only on the second-moment method.

\subsection{\erdosrenyi graphs}\label{sec:appendix-gnp}

In this section, we prove \Cref{lem:Gnp-motif-concentration} using the second moment method. For convenience, we re-state \Cref{lem:Gnp-motif-concentration}.

\begin{proposition}[\Cref{lem:Gnp-motif-concentration}, restated]
 Let $H$ be an induced motif on $k$ vertices, and suppose $0<p\leq1/2$. Consider parameters $n, p$ with $n \geq \frac{\omega_{1/p}(1)}{\varepsilon^2} \cdot p^{-m_H}$, where $m_H := \max_{F \subseteq H} e(F)/v(F)$.
 
For $G \sim \gnp$, with probability at least $1 - o_{1/p}(1)$, the number of copies of $H$ in $G$ is in the range $(1 \pm \frac{\varepsilon}{100}) \binom{n}{v(H)} v(H)! \cdot p^{e(H)} (1-p)^{\binom{v(H)}{2} - e(H)}$. 
\end{proposition}

\begin{proof}
    Let $X_G$ be the number of labeled induced copies of $H$ in a graph $G$. First, observe that $\mathbb{E}\left[X_G \right] = \binom{n}{v(H)} v(H)! \cdot p^{e(H)} (1-p)^{\binom{v(H)}{2} - e(H)}$.

    Let's now analyze the variance of $X_G$. For an ordered set $S$, let $S \equiv H$ mean that $S$ forms a labeled induced copy of $H$. We can write $X_G$ as:
    $$X_G = \sum_{S \in [n]^k} \mathbbm{1}(S \equiv H).$$
    Tuples with repeated vertices contribute zero. Write $I_S=\mathbbm{1}(S\equiv H)$, and let $\theta=\mathbb{E}I_S$ for an injective tuple. Therefore, the variance is
    $$\mathrm{Var}[X_G]=\sum_{S,T}\bigl(\mathbb{E}[I_SI_T]-\mathbb{E}[I_S]\mathbb{E}[I_T]\bigr).$$
    If $S$ and $T$ share at most one vertex, their indicators depend on disjoint edge variables and their covariance is zero. Suppose they share $i\geq2$ vertices. If their requirements on the shared pairs disagree, the joint probability is zero. Otherwise, the shared vertices induce a submotif $H'\subseteq H$, and
    $$\mathbb{E}[I_SI_T]=\theta^2p^{-e(H')}(1-p)^{-\binom{i}{2}+e(H')}.$$
    For each $S$, the number of tuples $T$ sharing $i$ vertices is at most $\binom{k}{i}(k)_i(n-k)_{k-i}$, where $(a)_j=a(a-1)\cdots(a-j+1)$. Thus, for a constant $C_H$ depending only on $H$,
    $$\frac{\mathrm{Var}[X_G]}{\mathbb{E}[X_G]^2}\leq C_H\sum_{i=2}^{k}n^{-i}\max_{\substack{H'\subseteq H\\|H'|=i}}p^{-e(H')}(1-p)^{-\binom{i}{2}+e(H')}.$$
    This includes fully overlapping tuples, including different orderings of the same vertices. Since $1-p\geq1/2$ and $e(H')\leq m_H i$, we obtain
    $$\frac{\mathrm{Var}[X_G]}{\mathbb{E}[X_G]^2}\leq C'_H\sum_{i=2}^{k}(np^{m_H})^{-i}.$$
    Since $n \geq \frac{\omega_{1/p}(1)}{\varepsilon^2} \cdot p^{-m_H}$, this gives us:
    $$\text{Var}\left[ X_G\right]\leq \mathbb{E}\left[X_G \right]^2 \cdot \frac{\varepsilon^2}{\omega(1)} = \mathbb{E}\left[X_G \right]^2 \cdot \varepsilon^2 \cdot o_{1/p}(1).$$
    Finally, by Chebyshev's inequality, 
    $$\mathbb{P}\left[\left|X_G - \mathbb{E}\left[ X_G\right] \right| \geq (\varepsilon/100)\mathbb{E}\left[ X_G\right] \right] \leq \frac{\text{Var}\left[X_G \right]}{(\varepsilon/100)^2 \mathbb{E}\left[ X_G\right]^2} = o_{1/p}(1).\eqno \qedhere$$
\end{proof}

\subsection{Stochastic Block Models}\label{sec:appendix-sbm}

We now prove that motifs are concentrated in stochastic block models with two communities by using the second moment method.

Consider a two-community stochastic block model. Community 1 has $n_1$ vertices (with $n_1 \geq n/2$), and in-community edge probability $p$. Community 2 has $n_2$ vertices, and in-community edge probability $1/2$. The between-community edge probability is $1/2$. Let $\mathcal{G}_{\mathrm{SBM}, n, p}$ be the distribution over graphs given by this model. Without loss of generality, suppose Community 1 is on the first $n_1$ vertices of the graph, according to some arbitrary ordering.

\begin{proposition}\label{thm:sbm-H-concentration}

    Consider parameters $n,n_1,n_2,p$ with $0<p\leq1/2$ and $0<\varepsilon\leq1$, satisfying  $n_1 + n_2 = n$, $n\geq C_k p^{-\Delta(H)}/\varepsilon^2$ for a sufficiently large constant $C_k$, and  $2^{2\binom{k}2+1} n p^{\Delta(H)/2} \leq n_2 \leq n/2$, where $\Delta(H)$ is the maximum vertex degree of $H$.   
    
    Consider any motif $H$ on $k$ vertices. With probability at least $\frac{11}{12}$, the number of labeled induced copies of $H$ in $G \sim \mathcal{G}_{\mathrm{SBM}, n, p}$ is within a $(1 \pm \varepsilon)$ multiplicative factor of its expectation. 
\end{proposition}

We prove this by proving the following stronger statement.

\begin{lemma}\label{lem:sbm-H-concentration-specific}
    Consider parameters $n,n_1,n_2,p$ with $0<p\leq1/2$ and $0<\varepsilon\leq1$, satisfying  $n_1 + n_2 = n$, $n\geq C_k p^{-\Delta(H)}/\varepsilon^2$ for a sufficiently large constant $C_k$, and  $2^{2\binom{k}2+1} n p^{\Delta(H)/2} \leq n_2 \leq n/2$, where $\Delta(H)$ is the maximum vertex degree of $H$.   
    
    Consider any motif $H$ on $k$ vertices, and any fixed labeled vertex subset $V(H')\subseteq V(H)$ with $|H'|=k'$. For a graph $G$ with $n$ vertices, let $N_{H', H}(G)$ be the number of copies of $H$ in $G$, such that $H'$ is placed within the first $n_1$ vertices of $G$, and $H \setminus H'$ is placed within the other $n_2$ vertices.
    
    With probability at least $1 - \frac{1}{12} \cdot \frac{1}{2^k}$ over $G \sim \mathcal{G}_{\mathrm{SBM}, n, p}$, 
    \begin{equation}\label{eq:sbm-good-range}
      N_{H', H}(G) \in (1 \pm \varepsilon) \binom{n_1}{k'}\binom{n_2}{k-k'}k'!(k-k')! \cdot p^{e(H')} (1-p)^{\binom{k'}{2} - e(H')} (1/2)^{\binom{k}{2} - \binom{k'}{2}}.  
    \end{equation}
\end{lemma}

\begin{proof}
    We apply Chebyshev's inequality; therefore, we must analyze the expectation and variance of $N_{H', H}(G)$. For ease of notation, we refer to $ N_{H', H}(G)$ as $X_G$

    First, we see
    $$\mathbb{E}_{G \sim \mathcal{G}_{\mathrm{SBM}, n, p}}\left[ X_G\right] = \binom{n_1}{k'}\binom{n_2}{k-k'}k'!(k-k')! \cdot p^{e(H')} (1-p)^{\binom{k'}{2} - e(H')} (1/2)^{\binom{k}{2} - \binom{k'}{2}}.$$

    Next, we analyze the variance. For an ordered set $S$, let $S \equiv H$ mean that $S$ forms a labeled induced copy of $H$. First, observe that we can write $X_G$ as:
    $$X_G = \sum_{\substack{S_1 \in [n_1]^{k'} \\ S_2 \in [n_2]^{k - k'}}} \mathbbm{1}(S_1 \equiv H') \mathbbm{1}(S_1 \cup S_2 \equiv H).$$
    Therefore, $$\mathrm{Var}\left[X_G \right] = \mathbb{E}\left[ X_G^2\right] - \mathbb{E}\left[ X_G\right]^2 $$ $$= \sum_{\substack{S_1 \in [n_1]^{k'} \\ S_2 \in [n_2]^{k - k'}}} \sum_{\substack{T_1 \in [n_1]^{k'} \\ T_2 \in [n_2]^{k - k'}}} \mathbb{E}\left[\mathbbm{1}(S_1 \equiv H') \mathbbm{1}(S_1 \cup S_2 \equiv H) \mathbbm{1}(T_1 \equiv H') \mathbbm{1}(T_1 \cup T_2 \equiv H)\right] - \mathbb{E}\left[ X_G\right]^2.$$

    We break up the summation according to the intersections in the two communities. Write $\theta$ for the probability that a fixed admissible tuple induces $H$, and let $M=(n_1)_{k'}(n_2)_{k-k'}$, so $\mathbb{E}X_G=M\theta$.

    Consider pairs of tuples for which $S_1$ and $T_1$ intersect in $i$ vertices, and $S_2$ and $T_2$ intersect in $j$ vertices. We include $0\leq i\leq k'$ and $0\leq j\leq k-k'$, including full overlaps. As in the previous proof, we sum covariances, and tuples whose constraints disagree have joint probability zero.

    Suppose $i\geq1$. Let $F$ be a compatible overlap on the $i$ shared vertices in Community 1. The shared constraints give the following upper bound on the joint probability:
    \begin{equation}\label{eq:variance-sum}
       \mathbb{E}[I_SI_T]\leq\theta^2p^{-e(F)}(1-p)^{-\binom{i}{2}+e(F)}2^{ij+\binom{j}{2}}.
    \end{equation}
    The exponent $ij$ accounts for shared pairs between the communities. There are at most $K_k M^2 n_1^{-i}n_2^{-j}$ pairs with this overlap, for a constant $K_k$ depending only on $k$. Since $e(F)\leq i\Delta(H)/2$ and $1-p\geq1/2$, the contribution to the variance is at most
    \begin{equation}\label{eq:XG-expectation}
       \mathbb{E}[X_G]^2 K'_k\left(n_1p^{\Delta(H)/2}\right)^{-i}n_2^{-j},
    \end{equation}
    after absorbing the factors of $2$ into $K'_k$.

    First, when $i\geq1$, the size assumptions give $n_1p^{\Delta(H)/2}\geq C_k/(2\varepsilon^2)$ and $n_2\geq C_k/(2\varepsilon^2)$. Choosing $C_k\geq24\cdot2^k(k+1)^2K'_k$ bounds \Cref{eq:XG-expectation} by
    $$\mathbb{E}[X_G]^2\cdot\frac{\varepsilon^2}{12\cdot2^k(k+1)^2}.$$

    Next, when $i=0$ and $j\geq1$, the shared pairs all have probability $1/2$. Their contribution is at most
    $$\mathbb{E}[X_G]^2 K'_k n_2^{-j}\leq\mathbb{E}[X_G]^2\cdot\frac{\varepsilon^2}{12\cdot2^k(k+1)^2}.$$
    Finally, when $i=j=0$, the indicators are independent and their covariance is zero. The same is true when the total overlap has only one vertex.

    Combining the at most $(k+1)^2$ overlap types, we find that
    $$\mathrm{Var}[X_G]\leq\mathbb{E}[X_G]^2\cdot\frac{\varepsilon^2}{12\cdot2^k}.$$

    By Chebyshev's inequality, 
    $$\mathbb{P}\left[\left|X_G - \mathbb{E}\left[ X_G\right] \right| \geq \varepsilon \mathbb{E}\left[ X_G\right] \right] \leq \frac{\text{Var}\left[X_G \right]}{\varepsilon^2 \mathbb{E}\left[ X_G\right]^2} \leq \frac{1}{12} \cdot \frac{\varepsilon^2}{2^k} \cdot \frac{1}{\varepsilon^2} = \frac{1}{12\cdot2^k}.\eqno \qedhere$$
\end{proof}

\begin{proof}[Proof of \Cref{thm:sbm-H-concentration} from \Cref{lem:sbm-H-concentration-specific}]
Apply \Cref{lem:sbm-H-concentration-specific} over every $H' \subseteq H$. There are at most $2^k$ such $H'$. By a union bound, the probability (over $G \sim \mathcal{G}_{\mathrm{SBM}, n, p}$) that there exists an $N_{H', H}(G)$ not in the range given in \Cref{eq:sbm-good-range} is at most $\frac{1}{12}$. Since the number of labeled induced copies of $H$ in $G$ is the sum over $H'$ of the number of labeled induced copies of $H$ with $H'$ in Community 1, we find that, with probability at least $\frac{11}{12}$, the number of labeled induced copies of $H$ in $G \sim \mathcal{G}_{\mathrm{SBM}, n, p}$ is in a $(1 \pm \varepsilon)$ multiplicative factor of its expectation.
\end{proof}

\end{document}